\documentclass[10pt,a4paper]{article}
\usepackage[utf8]{inputenc}
\usepackage[T1]{fontenc}
\usepackage{amsmath,amsfonts,amssymb,amsthm}
\usepackage{graphicx}
\usepackage{enumitem,aliascnt,multirow,natbib,url}
\usepackage[hidelinks]{hyperref}
\usepackage{placeins}

\allowdisplaybreaks[4]

\numberwithin{equation}{section}

\newtheorem{theorem}{Theorem}[section]

\newaliascnt{corollary}{theorem}
\newtheorem{corollary}[corollary]{Corollary}
\aliascntresetthe{corollary}

\newaliascnt{lemma}{theorem}
\newtheorem{lemma}[lemma]{Lemma}
\aliascntresetthe{lemma}

\newaliascnt{proposition}{theorem}

\aliascntresetthe{proposition}

\theoremstyle{definition}
\newaliascnt{definition}{theorem}

\aliascntresetthe{definition}

\newtheorem*{definition*}{Definition}

\theoremstyle{remark}
\newaliascnt{remark}{theorem}

\aliascntresetthe{remark}

\newtheorem*{remark*}{Remark}

\theoremstyle{remark}
\newaliascnt{example}{theorem}

\aliascntresetthe{example}

\newtheorem*{example*}{Example}

\title{Depth based inference on conditional distribution with infinite dimensional data}
\author{Joydeep Chowdhury and Probal Chaudhuri\\ Indian Statistical Institute, Kolkata}
\date{ }

\begin{document}
\maketitle

\begin{abstract}
We develop inference and testing procedures for conditional dispersion and skewness in a nonparametric regression setup based on statistical depth functions.
The methods developed can be applied in situations, where the response is multivariate and the covariate is a random element in a metric space. This includes regression with functional covariate as a special case.
We construct measures of the center, the spread and the skewness of the conditional distribution of the response given the covariate using depth based nonparametric regression procedures.
We establish the asymptotic consistency of those measures and develop a test for heteroscedasticity and a test for conditional skewness.
We present level and power study for the tests in several simulated models.
The usefulness of the methodology is also demonstrated in a real dataset. In that dataset, our responses are the nutritional contents of different meat samples measured by their protein, fat and moisture contents, and the functional covariate is the absorbance spectra of the meat samples.
\end{abstract}

\section{Introduction} \label{sec:1}
A statistical depth function provides an ordering of the points such that the points close to the center of the distribution have higher depth values than those that are away from the center. Several depth functions are defined in the literature, e.g., the halfspace depth (\cite{tukey1975mathematics}, \cite{donoho1992breakdown}), the simplicial depth (\cite{liu1990notion}), the spatial depth (\cite{vardi2000multivariate}, \cite{serfling2002depth}), the projection depth (\cite{zuo2000general,zuo2000structural}), and many others.
Depth functions have been used earlier for various purposes like detecting outliers (\cite{chen2009outlier}), clustering (\cite{jornsten2004clustering}) and classification (\cite{ghosh2005maximum}, \cite{dutta2012robust}, \cite{li2012dd}).
We develop nonparametric regression methods using statistical depth functions. The response in our setup is multivariate, and the covariate is a random element in a separable metric space, which includes the cases of finite dimensional as well as infinite dimensional or function-valued covariates.
Nonparametric regression with functional covariate and real valued response has been studied in \cite{ferraty2006nonparametric}, \cite{ferraty2007nonparametric}, \cite{rachdi2007nonparametric}, \cite{chagny2014adaptive,chagny2016adaptive}, etc. Most of the authors investigated the conditional mean of a real valued or multivariate response and a functional covariate.
Nonparametric quantile regression with real valued or multivariate responses and functional covariates was investigated in \cite{ferraty2006nonparametric}, \cite{gardes2010functional} and \cite{chaouch2013nonparametric,chaouch2015vector}.

The traditional nonparametric regression methods estimate the conditional mean or median (see, e.g., \cite{hardle1990applied}, \cite{fan1996local}, \cite{wand1994kernel}, etc.), and provide information about only the center of the conditional distribution of the response given the covariate. These methods are not useful if one is interested in other features of the conditional distribution, like the conditional spread and heteroscedasticity in the sample. The methods we develop will provide simultaneous information about the center as well as other parts of the conditional distribution of the response.
In particular, we construct measures of conditional spread and conditional skewness of the response, and develop tests for heteroscedasticity and variation of conditional skewness over the covariate.

Tests for heteroscedasticity were studied in \cite{gupta1984distribution}, \cite{dette1998testing}, \cite{dette2002consistent}, \cite{holgersson2004testing} and \cite{dette2010robust}. These authors considered either parametric tests or nonparametric tests for real valued response.
In \cite{godfrey1991testing}, a test for symmetry of the distribution of the regression error was developed in the context of linear regression with univariate response.
Procedures to test for conditional symmetry in the context of regression in time series data and in autoregressive setup with econometric and financial data were developed in \cite{harvey1999autoregressive,harvey2000conditional}, \cite{bai2001consistent}, \cite{charoenrook2004conditional}, \cite{lanne2007modeling}, \cite{delgado2007nonparametric}, \cite{smith2007conditional}, \cite{bali2008role}, \cite{kuosmanen2009neoclassical}, \cite{grigoletto2009looking} and \cite{ghysels2011conditional}.
However, these authors too restricted themselves to mostly parametric tests, or nonparametric tests for real valued response in time series regression.

In this article, we develop methodology based on statistical depths, which can be used to investigate conditional spread and skewness in a nonparametric regression setup. Our proposed measures of conditional spread and skewness yield statistical tests for heteroscedasticity and conditional skewness when the response is multivariate and the covariate may be infinite dimensional.
To the best of our knowledge, there is no other nonparametric procedure to investigate heteroscedasticity and conditional skewness in such a general setup.

In \autoref{sec:2}, conditional depth based central regions are defined and several conditional depths are described. The asymptotic consistency of the estimates of the conditional depths and the conditional central regions are established in \autoref{sec:3}. In \autoref{sec:4}, measures of conditional spread are defined based on the conditional central regions, the consistency of their estimates is established, and a test of heteroscedasticity is proposed. In \autoref{sec:5}, depth based conditional median and conditional trimmed means are defined and the consistency of their estimates is discussed. Based on the conditional median, conditional trimmed means and the measure of conditional spread, measures of conditional skewness are defined in \autoref{sec:5}, and it is established that their estimates are asymptotically consistent. A test of conditional skewness is also proposed in \autoref{sec:5}. The methodology and the tests developed in this article are demonstrated using a real data in \autoref{sec:6}. Some concluding remarks are made in \autoref{sec:7}. The proofs of the theorems are provided in \autoref{sec:8}. In \autoref{supsec:1}, the performance of the tests is demonstrated in several simulated models, and in \autoref{supsec:2}, some additional mathematical details are provided.

\section{Conditional depth and central regions} \label{sec:2}
Let $ \mathbf{Y} $ be a random vector in $ \mathbb{R}^p $ and $ \mathbf{X} $ be a random element in a complete separable metric space $ ( \mathcal{C}, d ) $. Let $ \mu( \cdot \,|\, \mathbf{z} ) $ be the conditional probability measure of $ \mathbf{Y} $ given $ \mathbf{X} = \mathbf{z} $, and $ \rho( \cdot \,|\, \mathbf{z} ) $ be a conditional depth function on the response space $ \mathbb{R}^p $ related to $ \mu( \cdot \,|\, \mathbf{z} ) $. Values of the depth functions usually lie between 0 and 1, and consequently they are uniformly bounded.

Let $ \mathbf{x} $ be a fixed element in $ \mathcal{C} $. Define $ D( \alpha \,|\, \mathbf{x} ) = \{ \mathbf{y} \in \mathbb{R}^p \,|\, \rho( \mathbf{y} \,|\, \mathbf{x} ) \ge \alpha \} $ for $ \alpha \in \mathbb{R} $. $ D( \alpha \,|\, \mathbf{x} ) $ is called the conditional $ \alpha $-trimmed region \index{conditional $ \alpha $-trimmed region} of $ \mathbf{Y} $ given $ \mathbf{X} = \mathbf{x} $ corresponding to the conditional depth $ \rho( \cdot \,|\, \mathbf{x} ) $ (cf. \cite{zuo2000general,zuo2000structural}).
For $ 0 \le r < 1 $, let $ \alpha( r ) = \sup\{ \alpha \,|\, \mu( D( \alpha \,|\, \mathbf{x} ) \,|\, \mathbf{x} ) \ge r \} $, which is finite when $ \rho( \mathbf{y} \,|\, \mathbf{x} ) $ is uniformly bounded over $ \mathbf{y} $. We define the set $ D( \alpha( r ) \,|\, \mathbf{x} ) $ as the conditional depth based $ 100 r $\% central region \index{conditional central region} of $ \mathbf{Y} $ given $ \mathbf{X} = \mathbf{x} $ with respect to the conditional depth $ \rho( \cdot \,|\, \mathbf{x} ) $.
Clearly, $ \mu( D( \alpha( r ) \,|\, \mathbf{x} ) \,|\, \mathbf{x} ) \ge r $, and $ \mu( D( \alpha( r ) \,|\, \mathbf{x} ) \,|\, \mathbf{x} ) = r $ whenever $ \mu( \{ \mathbf{y} \in \mathbb{R}^p \,|\, \rho( \mathbf{y} \,|\, \mathbf{x} ) = \alpha( r ) \} \,|\, \mathbf{x} ) = 0 $. Further, for any $ \mathbf{y}_1 \in D( \alpha( r ) \,|\, \mathbf{x} ) $ and $ \mathbf{y}_2 \not\in D( \alpha( r ) \,|\, \mathbf{x} ) $, $ \rho( \mathbf{y}_1 \,|\, \mathbf{x} ) \ge \rho( \mathbf{y}_2 \,|\, \mathbf{x} ) $.

The conditional $ \alpha $-depth contour \index{conditional $ \alpha $-depth contour} $ \delta( \alpha \,|\, \mathbf{x} ) $ of $ \mathbf{Y} $ given $ \mathbf{X} = \mathbf{x} $ is defined as the boundary of $ D( \alpha \,|\, \mathbf{x} ) $, and for $ 0 < r < 1 $, the conditional $ 100 r $\% central region contour \index{conditional central region contour} of $ \mathbf{Y} $ given $ \mathbf{X} = \mathbf{x} $ is defined as $ \delta( \alpha( r ) \,|\, \mathbf{x} ) $. The conditional central region contour determine the shape of the conditional central region.

To estimate $ \rho( \cdot \,|\, \mathbf{x} ) $, $ D( \alpha \,|\, \mathbf{x} ) $ and $ D( \alpha( r ) \,|\, \mathbf{x} ) $ based on a random sample $ ( \mathbf{X}_1, \mathbf{Y}_1 ), \allowbreak \ldots, ( \mathbf{X}_n, \mathbf{Y}_n ) $, we adopt a nonparametric regression procedure. Let $ W_{i,n}( \mathbf{x} ) $ be the weight on the observation pair $ ( \mathbf{X}_i, \mathbf{Y}_i ) $, $ i = 1, \ldots, n $, $ W_{i,n}( \mathbf{x} ) \ge 0 $ for each $ i $ and $ \sum_{i=1}^{n} W_{i,n}( \mathbf{x} ) = 1 $. The sample conditional probability measure of $ \mathbf{Y} $ given $ \mathbf{X} = \mathbf{x} $ is defined as
\begin{align*}
\mu_n( B \,|\, \mathbf{x} ) = \sum_{i=1}^{n} \mathbb{I}( \mathbf{Y}_i \in B ) W_{i,n}( \mathbf{x} ) ,
\end{align*}
where $ B $ is any Borel set. The conditional sample depth function $ \rho_n( \cdot \,|\, \mathbf{x} ) $ is related to $ \mu_n( \cdot \,|\, \mathbf{x} ) $ in the same way as the conditional population depth function $ \rho( \cdot \,|\, \mathbf{x} ) $ is related to $ \mu( \cdot \,|\, \mathbf{x} ) $. The conditional sample $ \alpha $-trimmed region of $ \mathbf{Y} $ given $ \mathbf{X} = \mathbf{x} $ is defined as $ D_n( \alpha \,|\, \mathbf{x} ) = \{ \mathbf{y} \in \mathbb{R}^p \,|\, \rho_n( \mathbf{y} \,|\, \mathbf{x} ) \ge \alpha \} $. The conditional sample $ 100 r $\% central region of $ \mathbf{Y} $ given $ \mathbf{X} = \mathbf{x} $ is $ D_n( \alpha_n( r ) \,|\, \mathbf{x} ) $, where $ \alpha_n( r ) = \sup\{ \alpha \,|\, \mu_n( D_n( \alpha \,|\, \mathbf{x} ) \,|\, \mathbf{x} ) \ge r \} $.
The estimates of $ \delta( \alpha \,|\, \mathbf{x} ) $ and $ \delta( \alpha( r ) \,|\, \mathbf{x} ) $ are denoted as $ \delta_n( \alpha \,|\, \mathbf{x} ) $ and $ \delta_n( \alpha_n( r ) \,|\, \mathbf{x} ) $, respectively, and they are defined as the boundaries of $ D_n( \alpha \,|\, \mathbf{x} ) $ and $ D_n( \alpha_n( r ) \,|\, \mathbf{x} ) $, respectively. 

The weights $ \{  W_{i,n}( \mathbf{x} ) \} $ are constructed based on the covariate values $ \mathbf{X}_1, \ldots, \mathbf{X}_n $. There are several methods of selecting such weights. In the kernel regression method, we choose a kernel function $ K( \cdot ) $ and a bandwidth $ h $, and the weight $ W_{i,n}( \mathbf{x} ) $ is
\begin{align*}
W_{i,n}( \mathbf{x} ) = \frac{K( h^{-1} d( \mathbf{x}, \mathbf{X}_i ) )}{\sum_{i=1}^{n} K( h^{-1} d( \mathbf{x}, \mathbf{X}_i ) )} ,
\end{align*}
where $ d( \cdot, \cdot ) $ is the metric in the covariate space. This leads to a Nadaraya-Watson type kernel estimate (\cite{nadaraya1964estimating}, \cite{watson1964smooth}). In the nearest neighbor method, we choose a positive integer $ k $ for the number of nearest neighbors to be considered, and define
\begin{align*}
h( \mathbf{x}, k, n ) = \min\left\{ h \middle\arrowvert \sum_{i=1}^{n} \mathbb{I}( d( \mathbf{x}, \mathbf{X}_i ) \le h ) \ge k \right\} .
\end{align*}
The weight $ W_{i,n}( \mathbf{x} ) $ in this case is
\begin{align*}
W_{i,n}( \mathbf{x} ) = \frac{\mathbb{I}( d( \mathbf{x}, \mathbf{X}_i ) \le h( \mathbf{x}, k, n ) )}{\sum_{i=1}^{n} \mathbb{I}( d( \mathbf{x}, \mathbf{X}_i ) \le h( \mathbf{x}, k, n ) )} .
\end{align*}

We now describe the conditional versions of three well-known depth functions. The conditional halfspace depth \index{conditional halfspace depth} (\cite{tukey1975mathematics}, \cite{donoho1992breakdown}) of $ \mathbf{Y} $ given $ \mathbf{X} = \mathbf{x} $ is defined as $ \rho( \mathbf{y} \,|\, \mathbf{x} ) = \inf\{ \mu( \{ \mathbf{v} \in \mathbb{R}^p \,|\, \mathbf{u}^t \mathbf{v} \ge \mathbf{u}^t \mathbf{y} \} \,|\, \mathbf{x} ) \,|\, \mathbf{u} \in \mathbb{R}^p \} $. Its estimate is $ \rho_n( \mathbf{y} \,|\, \mathbf{x} ) = \inf\{ \mu_n( \{ \mathbf{v} \in \mathbb{R}^p \,|\, \mathbf{u}^t \mathbf{v} \ge \mathbf{u}^t \mathbf{y} \} \,|\, \mathbf{x} ) \,|\, \mathbf{u} \in \mathbb{R}^p \} $.
The conditional projection depth \index{conditional projection depth} (\cite{zuo2000general,zuo2000structural}) is defined as
\begin{align*}
\rho( \mathbf{y} \,|\, \mathbf{x} ) = \left[ 1 + \sup_{\| \mathbf{u} \| = 1} \frac{| \mathbf{u}^t \mathbf{y} - m( \mathbf{u}^t \mathbf{Y} \,|\, \mathbf{x} ) |}{m( | \mathbf{u}^t \mathbf{Y} - m( \mathbf{u}^t \mathbf{Y} \,|\, \mathbf{x} ) | \,|\, \mathbf{x} )} \right]^{-1} ,
\end{align*}
where $ m( \mathbf{u}^t \mathbf{Y} \,|\, \mathbf{x} ) $ and $ m( | \mathbf{u}^t \mathbf{Y} - m( \mathbf{u}^t \mathbf{Y} \,|\, \mathbf{x} ) | \,|\, \mathbf{x} ) $ are the conditional medians of $ \mathbf{u}^t \mathbf{Y} $ and $ | \mathbf{u}^t \mathbf{Y} - m( \mathbf{u}^t \mathbf{Y} \,|\, \mathbf{x} ) | $ given $ \mathbf{X} = \mathbf{x} $, respectively.
The conditional sample projection depth is
\begin{align*}
\rho_n( \mathbf{y} \,|\, \mathbf{x} ) = \left[ 1 + \sup_{\| \mathbf{u} \| = 1} \frac{| \mathbf{u}^t \mathbf{y} - m_n( \mathbf{u}^t \mathbf{Y} \,|\, \mathbf{x} ) |}{m_n( | \mathbf{u}^t \mathbf{Y} - m_n( \mathbf{u}^t \mathbf{Y} \,|\, \mathbf{x} ) | \,|\, \mathbf{x} )} \right]^{-1} ,
\end{align*}
where $ m_n( \mathbf{u}^t \mathbf{Y} \,|\, \mathbf{x} ) $ and $ m_n( | \mathbf{u}^t \mathbf{Y} - m_n( \mathbf{u}^t \mathbf{Y} \,|\, \mathbf{x} ) | \,|\, \mathbf{x} ) $ are the sample analogues of $ m( \mathbf{u}^t \mathbf{Y} \,|\, \mathbf{x} ) $ and $ m( | \mathbf{u}^t \mathbf{Y} - m( \mathbf{u}^t \mathbf{Y} \,|\, \mathbf{x} ) | \,|\, \mathbf{x} ) $, respectively.
The conditional spatial depth \index{conditional spatial depth} (\cite{vardi2000multivariate}, \cite{serfling2002depth}) is defined as $ \rho( \mathbf{y} \,|\, \mathbf{x} ) = 1 - \| \mathbb{E}[ \| \mathbf{y} - \mathbf{Y} \|^{-1} ( \mathbf{y} - \mathbf{Y} ) \,|\, \mathbf{X} = \mathbf{x} ] \| $.
Here, we adopt the convention of defining $ \left\| \mathbf{v} \right\|^{-1} \mathbf{v} = \mathbf{0} $ when $ \mathbf{v} = \mathbf{0} $.
The estimate $ \rho_n( \mathbf{y} \,|\, \mathbf{x} ) $ is given by
\begin{align*}
\rho_n( \mathbf{y} \,|\, \mathbf{x} ) = 1 - \left\| \sum_{i=1}^{n} \| \mathbf{y} - \mathbf{Y}_i \|^{-1} ( \mathbf{y} - \mathbf{Y}_i ) W_{i,n}( \mathbf{x}  ) \right\| .
\end{align*}
There are good algorithms available for computing these depth functions for multivariate data, and we use them in our numerical investigation later.
When the response is univariate, the conditional 50\% central regions for the depth functions described above correspond to the box in the conditional boxplot of the response.

\section{Asymptotic consistency of conditional depth and central regions} \label{sec:3}
Recall the definitions of $ \mu( \cdot \,|\, \mathbf{z} ) $, $ \mu_n( \cdot \,|\, \mathbf{x} ) $, $ \rho( \cdot \,|\, \mathbf{x} ) $ and $ \rho_n( \cdot \,|\, \mathbf{x} ) $ from \autoref{sec:2}.
We shall first establish the uniform asymptotic consistency for some standard conditional depth functions.
We need the following conditions on the weights $ W_{i,n}( \mathbf{x} ) $ and the conditional probability $ \mu( \cdot \,|\, \mathbf{z} ) $ of $ \mathbf{Y} $ given $ \mathbf{X} = \mathbf{z} $.
\begin{enumerate}[label=(A\arabic*), ref=(A\arabic*)]
\item \label{a1}
$ ( \log n ) \sum_{i=1}^{n} W_{i,n}^2( \mathbf{x} ) \stackrel{a.s.}{\longrightarrow} 0 $ and $ ( \log n ) \max_{1 \le i \le n} W_{i,n}( \mathbf{x} ) \stackrel{a.s.}{\longrightarrow} 0 $ as $ n \to \infty $,
and for any $ \delta > 0 $, $ \sum_{i=1}^{n} W_{i,n}( \mathbf{x} ) \mathbb{I}( d( \mathbf{x}, \mathbf{X}_i ) \ge \delta ) \allowbreak \stackrel{a.s.}{\longrightarrow} 0 $ as $ n \to \infty $.

\item \label{a2}
There is a collection of indices $ \mathbf{S}_n \subset \{ 1, 2, \ldots, n \} $ with cardinality $ k_n $ such that for any $ \delta > 0 $, $ \mathbf{S}_n \subset \{ i \,|\, d( \mathbf{x}, \mathbf{X}_i ) < \delta , \, i = 1, \ldots, n \} $ for all sufficiently large $ n $. Also,
\begin{align*}
& \frac{\log n}{k_n} \to 0
\text{ as } n \to \infty , \\
& \frac{\min_{i \in \mathbf{S}_n} W_{i,n}( \mathbf{x} )}{\max_{i \in \mathbf{S}_n} W_{i,n}( \mathbf{x} )}
\text{ is \emph{almost surely} bounded away from 0 as } n \to \infty , \\
& \text{and } \frac{\sum_{\{ i_1, \ldots, i_{p+1} \} \not\subset \mathbf{S}_n} W_{i_1,n}( \mathbf{x} ) \ldots W_{i_{p+1},n}( \mathbf{x} )}{\sum_{\{ j_1, \ldots, j_{p+1} \} \subset \mathbf{S}_n} W_{j_1,n}( \mathbf{x} ) \ldots W_{j_{p+1},n}( \mathbf{x} )}
\stackrel{a.s.}{\longrightarrow} 0
\text{ as }
n \to \infty .
\end{align*}

\item \label{a3}
$ \mu( \cdot \,|\, \mathbf{z} ) \stackrel{w}{\longrightarrow} \mu( \cdot \,|\, \mathbf{x} ) $ as $ \mathbf{z} \to \mathbf{x} $.
\end{enumerate}
Suppose that $ \mathbb{P}[ d( \mathbf{x}, \mathbf{X} ) \le h ] > 0 $ for all $ h > 0 $.
Then, in the nearest neighbor method of constructing weights $ W_{i,n}( \mathbf{x} ) $, conditions \ref{a1} and \ref{a2} hold when $ k_n = \lfloor ( \log n )^2 \rfloor + 1 $.
In the kernel method, suppose that the kernel function $ K( \cdot ) $ satisfies $ l \mathbb{I}( 0 \le s \le 1 ) \le K( s ) \le u \mathbb{I}( 0 \le s \le 1 ) $ for some constants $ 0 < l, u < \infty $.
Depending on the distribution of the covariate $ \mathbf{X} $, we can choose a sequence of bandwidths $ \{ h_n \} $ such that $ h_n \to 0 $ and $ [ n \mathbb{P}[ d( \mathbf{x}, \mathbf{X} ) \le h_n ] ]^{-1} ( \log n ) \to 0 $ as $ n \to \infty $.
For this choice of $ \{ h_n \} $ and the kernel function, it can be shown using the Bernstein inequality (see \citet[p.~95, Lemma A]{serfling2009approximation}) and the Borel-Cantelli Lemma that conditions \ref{a1} and \ref{a2} are satisfied.

Assumption \ref{a3} holds in many common situations that we encounter. As an example, consider the location-scale model defined by $ \mathbf{Y} = l( \mathbf{X} ) + s( \mathbf{X} ) \mathbf{G} $. Here, $ \mathbf{G} $ is a random vector independent of $ \mathbf{X} $, and the functions $ l( \cdot ) : \mathcal{C} \to \mathbb{R}^p $ and $ s( \cdot ) : \mathcal{C} \to \mathbb{R} $ are both continuous at $ \mathbf{x} $. In such a setup, $ \mu( \cdot \,|\, \mathbf{z} ) \stackrel{w}{\longrightarrow} \mu( \cdot \,|\, \mathbf{x} ) $ as $ \mathbf{z} \to \mathbf{x} $.
Under \ref{a1} and \ref{a3}, it follows that
\begin{align}
\mu_n( \cdot \,|\, \mathbf{x} ) \stackrel{w}{\longrightarrow} \mu( \cdot \,|\, \mathbf{x} )
\quad
\text{\emph{almost surely} as }
n \to \infty .
\label{eq1}
\end{align}
See \autoref{lemma1} for the proof of \eqref{eq1}.
The following theorem states the uniform consistency for several depth functions.
\begin{theorem} \label{thm1}
Let $ \rho( \mathbf{y} \,|\, \mathbf{x} ) $ be any of the three conditional depth functions described in \autoref{sec:2}. Suppose that there is a neighborhood of $ \mathbf{x} $ such that for all $ \mathbf{z} $ in that neighborhood, the conditional distribution of $ \mathbf{Y} $ given $ \mathbf{X} = \mathbf{z} $ has a continuous positive density $ f( \cdot \,|\, \mathbf{z} ) $ that is continuous in $ \mathbf{z} $.
Then, under conditions \ref{a1}, \ref{a2} and \ref{a3},
\begin{align}
\sup_{ \mathbf{y} \in \mathbb{R}^p } | \rho_n( \mathbf{y} \,|\, \mathbf{x} ) - \rho( \mathbf{y} \,|\, \mathbf{x} ) | \stackrel{P}{\longrightarrow} 0
\quad
\text{as }
n \to \infty .
\label{eq2}
\end{align}
\end{theorem}
Examples of continuous positive conditional densities $ f( \cdot \,|\, \mathbf{z} ) $ that are continuous in $ \mathbf{z} $ include the location-scale model: $ \mathbf{Y} = l( \mathbf{X} ) + s( \mathbf{X} ) \mathbf{G} $, where $ \mathbf{G} $ is a random vector independent of $ \mathbf{X} $ and has a continuous positive density on $ \mathbb{R}^p $, and the functions $ l( \cdot ) : \mathcal{C} \to \mathbb{R}^p $ and $ s( \cdot ) : \mathcal{C} \to \mathbb{R} $ are both continuous in a neighborhood of $ \mathbf{x} $.

We now proceed to state the asymptotic consistency of the conditional central regions. Recall the definitions of $ D( \cdot \,|\, \mathbf{x} ) $, $ D_n( \cdot \,|\, \mathbf{x} ) $, $ \alpha( r ) $ and $ \alpha_n( r ) $ from \autoref{sec:2}. We need the following conditions.
\begin{enumerate}[label=(A\arabic*), ref=(A\arabic*)]
\setcounter{enumi}{3}
\item \label{a4}
$ \mu( \{ \mathbf{y} \,|\, \rho( \mathbf{y} \,|\, \mathbf{x} ) = \alpha \} \,|\, \mathbf{x} ) = 0 $ for all $ \alpha $.

\item \label{a5}
For any $ 0 \le \gamma_1 < \gamma_2 \le \max_{\mathbf{y}} \rho( \mathbf{y} \,|\, \mathbf{x} ) $, $ \mu( \{ \mathbf{y} \,|\, \gamma_1 \le \rho( \mathbf{y} \,|\, \mathbf{x} ) \le \gamma_2 \} \,|\, \mathbf{x} ) > 0 $.

\item \label{a6}
$ \rho( \mathbf{y} \,|\, \mathbf{x} ) $ is a continuous function of $ \mathbf{y} $, and $ \rho( \mathbf{y} \,|\, \mathbf{x} ) \to 0 $ as $ \| \mathbf{y} \| \to \infty $.

\item \label{a7}
Define $ D_0( \alpha \,|\, \mathbf{x} ) = \{ \mathbf{y} \,|\, \rho( \mathbf{y} \,|\, \mathbf{x} ) > \alpha \} $. Then, the closure of $ D_0( \alpha \,|\, \mathbf{x} ) $ is $ D( \alpha \,|\, \mathbf{x} ) $ for any $ 0 < \alpha < \rho( \mathbf{m}( \mathbf{x} ) \,|\, \mathbf{x} ) $.
\end{enumerate}
Assumptions \ref{a4} and \ref{a5} imply that the distribution function of the random variable $ \rho( \mathbf{Y} \,|\, \mathbf{x} ) $ is continuous and strictly increasing. Assumptions \ref{a6} and \ref{a7} are related to the smoothness of the conditional depth and the corresponding central regions.

Suppose $ \mu( \cdot \,|\, \mathbf{x} ) $ has a probability density with a convex support. Then, from the proof of Lemma 6.1 in \cite{donoho1992breakdown}, we get that \ref{a6} is satisfied for the conditional halfspace depth. From the arguments in the proof of Lemma 6.3 in \cite{donoho1992breakdown}, we get that $ D( \alpha \,|\, \mathbf{x} ) $ is convex for all $ 0 < \alpha < \rho( \mathbf{m}( \mathbf{x} ) \,|\, \mathbf{x} ) $, and consequently, \ref{a4}, \ref{a5} and \ref{a7} are also satisfied for the conditional halfspace depth.

Now suppose that $ \mu( \cdot \,|\, \mathbf{x} ) $ has a positive probability density on $ \mathbb{R}^p $. Then, the assumptions \ref{a4}, \ref{a5} and \ref{a6} hold for the conditional spatial depth. Note that when assumption \ref{a6} holds, assumption \ref{a7} is also satisfied if the conditional depth has no local maximum. From Theorem 1 in \citet[p.~230]{chowdhury2017nonparametric}, we get that the conditional spatial median is unique. Using arguments similar to those in the proofs of Lemma 2.5 and Lemma 2.6 in the supplement of \cite{chowdhury2019nonparametric}, one can show that the conditional spatial depth $ \rho( \mathbf{y} \,|\, \mathbf{x} ) $ has non-zero Fr\'{e}chet derivative with respect to $ \mathbf{y} $ for $ \mathbf{y} $ not being the conditional spatial median. Consequently, the conditional spatial depth cannot have any local maximum and hence it satisfies assumption \ref{a7}.

From Theorem 3.4 in \cite{zuo2000structural}, it follows that it is sufficient for $ \mu( \cdot \,|\, \mathbf{x} ) $ to have an elliptically symmetric probability density on $ \mathbb{R}^p $ for the associated conditional projection depth to satisfy assumptions \ref{a4}, \ref{a5}, \ref{a6} and \ref{a7}.

In the theorem below, we establish an asymptotic consistency result for the conditional sample central regions.
Note that the Hausdorff distance $ d_H( A, B ) $ between two nonempty closed and bounded subsets of a metric space is defined as $ d_H( A, B ) = \inf\{ \epsilon \,|\, A \subseteq B^\epsilon , B \subseteq A^\epsilon \} $, where $ A^\epsilon $ and $ B^\epsilon $ denote the $ \epsilon $-neighborhoods of $ A $ and $ B $, respectively (see, e.g., \citet[p.~281]{munkres2000topology}).
\begin{theorem} \label{thm2}
Suppose \ref{a4}, \ref{a5}, \eqref{eq1} and \eqref{eq2} are satisfied. Then, for any $ \epsilon > 0 $ and any $ 0 < r < 1 $,
\begin{align*}
\mathbb{P}\left[ D( \alpha( r ) + \epsilon \,|\, \mathbf{x} ) \subseteq D_n( \alpha_n( r ) \,|\, \mathbf{x} ) \subseteq D( \alpha( r ) - \epsilon \,|\, \mathbf{x} ) \right] \to 1
\text{ as } n \to \infty .
\end{align*}
If in addition, conditions \ref{a6} and \ref{a7} hold, then
\begin{align*}
d_H( D_n( \alpha_n( r ) \,|\, \mathbf{x} ) , D( \alpha( r ) \,|\, \mathbf{x} ) ) \stackrel{P}{\longrightarrow} 0
\end{align*}
as $ n \to \infty $ for any $ 0 < r < 1 $.
\end{theorem}

A consequence of \autoref{thm2}, stated below in \autoref{coro1}, is that for any band around the conditional central region contour $ \delta( \alpha( r ) \,|\, \mathbf{x} ) $, which may be of arbitrarily small width, the conditional sample central region contour $ \delta_n( \alpha_n( r ) \,|\, \mathbf{x} ) $ lies inside that band with high probability for large $ n $. Since the contours determine the shapes of the central regions, this implies that the shapes of the sample central regions are good approximations of their population counter-parts in large samples.
\begin{corollary} \label{coro1}
Suppose \ref{a4}, \ref{a5}, \eqref{eq1} and \eqref{eq2} are satisfied. Then, for any $ \epsilon > 0 $,
\begin{align*}
\mathbb{P}\left[ \delta_n( \alpha_n( r ) \,|\, \mathbf{x} ) \subseteq \{ \mathbf{y} \in \mathbb{R}^p \,|\, \alpha( r ) - \epsilon \le \rho( \mathbf{y} \,|\, \mathbf{x} ) < \alpha( r ) + \epsilon \} \right] \to 1
\quad \text{as } n \to \infty .
\end{align*}
\end{corollary}

Convergence of sample central regions were studied earlier in the unconditional setup by \cite{he1997convergence} and \cite{zuo2000structural}. In \cite{he1997convergence}, the authors assumed that the population depth contours are elliptic in nature in order to establish the convergence of the contours of the central regions. In \cite{zuo2000structural}, the authors restricted themselves to elliptic distributions for proving the convergence of the central regions in the Hausdorff distance.

\section{Measure of conditional spread and study of heteroscedasticity} \label{sec:4}
In this section, we define a measure of conditional spread based on the conditional central region, establish the asymptotic consistency of its estimate and propose a test of heteroscedasticity based on it.

A large conditional central region indicates a high spread of the conditional distribution of the response.
In \cite{liu1999multivariate}, a measure of spread based on central regions in an unconditional setup was considered. The analogous measure of conditional spread in our regression setup is
$ \Lambda( r \,|\, \mathbf{x} ) = \text{volume}( D( \alpha( r ) \,|\, \mathbf{x} ) ) = \lambda( D( \alpha( r ) \,|\, \mathbf{x} ) ) $,
where $ \lambda( \cdot ) $ is the Lebesgue measure in $ \mathbb{R}^p $ and $ 0 < r < 1 $.
We can estimate $ \Lambda( r \,|\, \mathbf{x} ) $ by its sample analogue $ \Lambda_n( r \,|\, \mathbf{x} ) = \text{volume}( D_n( \alpha_n( r ) \,|\, \mathbf{x} ) ) = \lambda( D_n( \alpha_n( r ) \,|\, \mathbf{x} ) ) $. In \autoref{thm3}, we show that $ \Lambda_n( r \,|\, \mathbf{x} ) $ is a consistent estimate of $ \Lambda( r \,|\, \mathbf{x} ) $. In \cite{liu1999multivariate}, a different estimate of $ \Lambda( r \,|\, \mathbf{x} ) $ was considered, which is the volume of the convex hull of the sample observations lying in $ D_n( \alpha_n( r ) \,|\, \mathbf{x} ) $. But this estimate may not be consistent when the conditional central region $ D( \alpha( r ) \,|\, \mathbf{x} ) $ is not convex. The measure $ \Lambda_n( r \,|\, \mathbf{x} ) $ is computationally intensive.
For this reason, we consider an alternative spread measure $ \Delta( r \,|\, \mathbf{x} ) $ in our analyses.

We define the measure $ \Delta( r \,|\, \mathbf{x} ) $ of conditional spread \index{measure of conditional spread} of $ \mathbf{Y} $ given $ \mathbf{X} = \mathbf{x} $ as the diameter of the set $ D( \alpha( r ) \,|\, \mathbf{x} ) $, i.e., $ \Delta( r \,|\, \mathbf{x} ) = \sup\{ \| \mathbf{y}_1 - \mathbf{y}_2 \| \,|\, \mathbf{y}_1, \mathbf{y}_2 \in D( \alpha( r ) \,|\, \mathbf{x} ) \} $.
For a real valued response and for the conditional depth measures described in \autoref{sec:2}, $ \Delta( 0.5 \,|\, \mathbf{x} ) $ coincides with the conditional interquartile range.
We estimate $ \Delta( r \,|\, \mathbf{x} ) $ by
\begin{align*}
\Delta_n( r \,|\, \mathbf{x} ) = \max\{ \| \mathbf{Y}_i - \mathbf{Y}_j \| \,|\, \rho_n( \mathbf{Y}_i \,|\, \mathbf{x} ) , \rho_n( \mathbf{Y}_j \,|\, \mathbf{x} ) \ge \alpha_n( r );\, i,j = 1, \ldots, n \} .
\end{align*}
In \autoref{thm3}, we show that the estimates $ \Lambda_n( r \,|\, \mathbf{x} ) $ and $ \Delta_n( r \,|\, \mathbf{x} ) $ are asymptotically consistent.
\begin{theorem} \label{thm3}
(a) Suppose \ref{a4}, \ref{a5}, \eqref{eq1} and \eqref{eq2} hold, and
$ \mu( \cdot \,|\, \mathbf{x} ) $ has a continuous positive density on $ \mathbb{R}^p $.
Then,
$ \Lambda_n( r \,|\, \mathbf{x} ) \stackrel{P}{\longrightarrow} \Lambda( r \,|\, \mathbf{x} ) $
as $ n \to \infty $.

(b) 
Let $ 0 < r < 1 $, and suppose that for any open set $ G $ with $ G \cap D( \alpha( r ) \,|\, \mathbf{x} ) \neq \emptyset $, we have $ \mathbb{P}\left[ \mathbf{Y} \in G \cap D( \alpha( r ) \,|\, \mathbf{x} ) \right] > 0 $.
Then, under \ref{a4}, \ref{a5}, \ref{a6}, \ref{a7}, \eqref{eq1} and \eqref{eq2}, $ \Delta_n( r \,|\, \mathbf{x} ) \stackrel{P}{\longrightarrow} \Delta( r \,|\, \mathbf{x} ) $ as $ n \to \infty $.
\end{theorem}

The condition in the above theorem that for any open set $ G $ with $ G \cap D( \alpha( r ) \,|\, \mathbf{x} ) \neq \emptyset $, $ \mathbb{P}\left[ \mathbf{Y} \in G \cap D( \alpha( r ) \,|\, \mathbf{x} ) \right] > 0 $ is trivially satisfied for any depth function when the support of the distribution of $ \mathbf{Y} $ is the whole response space.
On the other hand, for particular depths like the conditional halfspace depth, the conditional central regions are bounded convex sets. So, if the support of the distribution of $ \mathbf{Y} $ is a bounded convex set, then also that condition is satisfied for the conditional halfspace depth.

In \cite{dette1998testing}, \cite{dette2002consistent} and \cite{dette2010robust}, some tests of heteroscedasticity were considered in a nonparametric regression setup with real valued response and covariate. Some tests of heteroscedasticity were also proposed in the case of parametric regression with multivariate response (\cite{gupta1984distribution}, \cite{holgersson2004testing}).
We now propose a nonparametric test for heteroscedasticity based on the conditional spread measure $ \Delta( r \,|\, \mathbf{x} ) $. In the presence of heteroscedasticity, $ \Delta( r \,|\, \mathbf{x} ) $ will vary with $ \mathbf{x} $. Our hypotheses are
\begin{center}
\begin{minipage}[c]{.4\textwidth}
\begin{itemize}
\item[$ H_0 : $] $ \Delta( r \,|\, \mathbf{x} ) $ is constant over $ \mathbf{x} $,
\item[$ H_A : $] $ \Delta( r \,|\, \mathbf{x} ) $ varies with $ \mathbf{x} $.
\end{itemize}
\end{minipage}
\end{center}
To capture the variation of $ \Delta( r \,|\, \mathbf{x} ) $ over $ \mathbf{x} $, we compute $ \Delta_n( r \,|\, \mathbf{X}_i ) $ for $ i = 1, \ldots, n $. We define our test-statistic as
\begin{align*}
\mathbf{T}_n = \frac{1}{n} \sum_{i=1}^{n} \left[ \Delta_n( r \,|\, \mathbf{X}_i ) - \left( \frac{1}{n} \sum_{j=1}^{n} \Delta_n( r \,|\, \mathbf{X}_j ) \right) \right]^2 .
\end{align*}
Large values of $ \mathbf{T}_n $ bear evidence against $ H_0 $.

Finding the actual distribution of the test statistic $ \mathbf{T}_n $ is difficult, and we compute the p-value of the test using a permutation procedure. Let $ \mathbf{S}_n $ be the set of all permutations of the integers $ 1, \ldots, n $, defined as $ \mathbf{S}_n = \{ \sigma \,|\, \sigma : \{ 1, \ldots, n \} \to \{ 1, \ldots, n \}, \, \sigma \text{ is one-to-one and onto} \} $. Consider $ \mathbb{S}_n = \{ ( \mathbf{X}_1, \mathbf{Y}_{\sigma( 1 )} ), \ldots, ( \mathbf{X}_n, \mathbf{Y}_{\sigma( n )} ) \allowbreak \,|\, \sigma \in \mathbf{S}_n \} $, the set of all permuted samples, where the response is freely permuted.
We compute the value of $ \mathbf{T}_n $ for all the permuted samples in $ \mathbb{S}_n $, and the empirical distribution of those values can be taken as an approximation of the null distribution of $ \mathbf{T}_n $. The p-value is computed as the proportion of those values of $ \mathbf{T}_n $, which are larger than the actually observed value of $ \mathbf{T}_n $. In practice, the number of all permutations is too large for even moderate sample sizes, and we consider a fixed number of random permutations of $ 1, \ldots, n $ to calculate the p-value based on that.

In \autoref{sec:6}, we demonstrate the conditional spread measure and the test of heteroscedasticity in a real dataset. A size and power study is presented for the test in simulated models in \autoref{supsec:1}.

\section{Measure of conditional skewness and related inference} \label{sec:5}
In this section, we define conditional medians and trimmed means based on statistical depths and establish the asymptotic consistency of their estimates. Based on these measures of the center of the conditional distribution of the response and the measure of conditional spread defined in \autoref{sec:4}, we define measures of conditional skewness of the response. We show that the estimates of these measures are asymptotically consistent and propose a test of conditional skewness based on them.

In the context of linear regression with univariate responses, \cite{godfrey1991testing} developed a test for symmetry of the distribution of the regression error.
In \cite{bai2001consistent} and \cite{delgado2007nonparametric}, procedures were developed to test for conditional symmetry in the case of regression involving time series data.
In \cite{kuosmanen2009neoclassical}, tests for skewness of regression errors were developed for an econometric regression problem.
\cite{harvey1999autoregressive} studied the conditional skewness in asset returns in an autoregressive setup.
In \cite{harvey2000conditional}, the authors investigated the economic importance of accounting for systematic conditional skewness in the distribution of the asset returns, and how this can be employed to improve expected returns.
\cite{bali2008role} investigated the effect of conditional skewness in the estimation of conditional value at risk in an autoregressive setup.
\cite{charoenrook2004conditional} investigated the effect of conditional skewness in aggregate market returns.
\cite{lanne2007modeling} and \cite{grigoletto2009looking} proposed different autoregressive models for conditional skewness in stock returns and financial times series.
In \cite{ghysels2011conditional}, a robust measure of conditional skewness was proposed, and the economic significance of conditional skewness in portfolio allocation was investigated.
In \cite{smith2007conditional}, the implication of conditional skewness in asset pricing was studied.

In the work described above, the authors investigated conditional skewness with either a finite dimensional covariate or in a time series regression setup with a real valued response. For multivariate random variables in an unconditional setup, statistical depths have been employed to investigate skewness in the data.
In \cite{rousseeuw2002depth}, the halfspace depth was used to develop a test for angular symmetry in multivariate data.
In \cite{dyckerhoff2015depth}, an affine invariant robust test for symmetry in bivariate data was constructed based on statistical depths.
We consider the case of nonparametric regression with multivariate response and a covariate taking values in a metric space, and develop methods to investigate conditional skewness using statistical depths.

The depth based conditional median \index{depth based conditional median} $ \mathbf{m}( \mathbf{x} ) $ of $ \mathbf{Y} $ given $ \mathbf{X} = \mathbf{x} $ with respect to the conditional depth $ \rho( \cdot \,|\, \mathbf{x} ) $ is a point such that $ \rho( \mathbf{m}( \mathbf{x} ) \,|\, \mathbf{x} ) \ge \rho( \mathbf{y} \,|\, \mathbf{x} ) $ for every $ \mathbf{y} $. Note that $ \mathbf{m}( \mathbf{x} ) $ may not be unique. For a univariate response and the depth functions described in \autoref{sec:2}, the point $ \mathbf{m}( \mathbf{x} ) $ becomes the usual conditional median, which can be viewed as the center of the conditional boxplot, and the set $ D( \alpha( r ) \,|\, \mathbf{x} ) $ becomes the conditional interquartile interval for $ r = 0.5 $, which corresponds to the box in the conditional boxplot as already noted. $ \mathbf{m}( \mathbf{x} ) $ along with the conditional central region $ D( \alpha( r ) \,|\, \mathbf{x} ) $ yields information about the center and the spread of the conditional distribution of the response.
The sample conditional median $ \mathbf{m}_n( \mathbf{x} ) $ is a point, which may not be unique, such that $ \rho_n( \mathbf{m}_n( \mathbf{x} ) \,|\, \mathbf{x} ) \ge \rho_n( \mathbf{y} \,|\, \mathbf{x} ) $ for every $ \mathbf{y} $.

Trimmed means based on depth functions in an unconditional setup were investigated earlier in \cite{donoho1992breakdown}, \cite{liu1999multivariate}, \cite{zuo2006multidimensional}, \cite{masse2009multivariate}, etc.
The conditional $ 100 r $\% trimmed mean \index{conditional trimmed mean} $ \mathbf{m}( r \,|\, \mathbf{x} ) $ of $ \mathbf{Y} $ given $ \mathbf{X} = \mathbf{x} $ is defined as
\begin{align*}
\mathbf{m}( r \,|\, \mathbf{x} ) = \frac{\int \mathbf{y} \mathbb{I}( \mathbf{y} \in D( \alpha( 1 - r ) \,|\, \mathbf{x} ) ) \mu( \mathrm{d} \mathbf{y} \,|\, \mathbf{x} )}{\mu( D( \alpha( 1 - r ) \,|\, \mathbf{x} ) \,|\, \mathbf{x} )} .
\end{align*}
Unlike the conditional median, the conditional trimmed mean is always unique.
For a real valued response, the depth based conditional trimmed mean coincides with the usual conditional trimmed mean for the depths considered in \autoref{sec:2}.
The sample conditional $ 100 r $\% trimmed mean $ \mathbf{m}_n( r \,|\, \mathbf{x} ) $ is
\begin{align*}
\mathbf{m}_n( r \,|\, \mathbf{x} ) = \frac{\int \mathbf{y} \mathbb{I}( \mathbf{y} \in D_n( \alpha_n( 1 - r ) \,|\, \mathbf{x} ) ) \mu_n( \mathrm{d} \mathbf{y} \,|\, \mathbf{x} )}{\mu_n( D_n( \alpha_n( 1 - r ) \,|\, \mathbf{x} ) \,|\, \mathbf{x} )} .
\end{align*}
We denote the conditional mean of $ \mathbf{Y} $ given $ \mathbf{X} = \mathbf{x} $ as $ \mathbf{M}( \mathbf{x} ) = \int \mathbf{y} \mu( \mathrm{d} \mathbf{y} \,|\, \mathbf{x} ) $, and its estimate is the conditional sample mean $ \mathbf{M}_n( \mathbf{x} ) = \int \mathbf{y} \mu_n( \mathrm{d} \mathbf{y} \,|\, \mathbf{x} ) $.

Conditional means or trimmed means along with conditional medians can be used to detect conditional skewness in the data. When the distribution is symmetric, the mean, trimmed means and the median coincide, whereas for a skewed distribution, the mean and the trimmed means lie away from the median. The distance between the mean or the trimmed mean and the median depends on the spread of the distribution, which needs to be accounted for while defining a measure of conditional skewness based on this distance. One advantage in using the distance between the conditional $ 100 r $\% trimmed mean with an appropriate value of $ r $ and the conditional median over the distance between the conditional mean and the conditional median in constructing a measure of conditional skewness is that the former is resistant to the presence of outliers. However, some information about the conditional distribution is lost in trimming. We define the measure $ \Psi_1( r_1, r_2 \,|\, \mathbf{x} ) $ of conditional skewness \index{measure of conditional skewness} of $ \mathbf{Y} $ given $ \mathbf{X} = \mathbf{x} $ as $ \Psi_1( r_1, r_2 \,|\, \mathbf{x} ) = \| \mathbf{m}( r_1 \,|\, \mathbf{x} ) - \mathbf{m}( \mathbf{x} ) \| / \Delta( r_2 \,|\, \mathbf{x} ) $. Its estimate is the corresponding sample version $ \Psi_{1,n}( r_1, r_2 \,|\, \mathbf{x} ) = \| \mathbf{m}_n( r_1 \,|\, \mathbf{x} ) - \mathbf{m}_n( \mathbf{x} ) \| / \Delta_n( r_2 \,|\, \mathbf{x} ) $. The measure $ \Psi_1( r_1, r_2 \,|\, \mathbf{x} ) $ is based on the distance between the conditional trimmed mean and the conditional median. Another measure of conditional skewness, $ \Psi_2( r \,|\, \mathbf{x} ) $, based on the distance between the conditional mean and the conditional median, is defined as $ \Psi_2( r \,|\, \mathbf{x} ) = \| \mathbf{M}( \mathbf{x} ) - \mathbf{m}( \mathbf{x} ) \| / \Delta( r \,|\, \mathbf{x} ) $. Its estimate is $ \Psi_{2,n}( r \,|\, \mathbf{x} ) = \| \mathbf{M}_n( \mathbf{x} ) - \mathbf{m}_n( \mathbf{x} ) \| / \Delta_n( r \,|\, \mathbf{x} ) $.

The next theorem states the asymptotic consistency of the sample conditional median and the sample conditional trimmed means.
\begin{theorem} \label{thm4}
Let $ \mathcal{M}( \mathbf{x} ) $ be the set of all conditional medians corresponding to the conditional depth function $ \rho( \mathbf{y} \,|\, \mathbf{x} ) $.
Suppose that conditions \ref{a4}, \ref{a5}, \ref{a6}, \ref{a7} are satisfied and \eqref{eq1} and \eqref{eq2} hold. Then,
for any sequence of sample conditional medians $ \{ \mathbf{m}_n( \mathbf{x} ) \} $,
$ \inf_{ \mathbf{m} \in \mathcal{M}( \mathbf{x} ) } \| \mathbf{m}_n( \mathbf{x} ) - \mathbf{m} \|
\stackrel{P}{\longrightarrow} 0 $
as $ n \to \infty $.
Also, $ \mathbf{m}_n( r \,|\, \mathbf{x} ) \stackrel{P}{\longrightarrow} \mathbf{m}( r \,|\, \mathbf{x} ) $ as $ n \to \infty $ for any $ 0 < r < 1 $.
\end{theorem}
The asymptotic consistency of the conditional skewness estimates $ \Psi_{1,n}( r_1, r_2 \,|\, \mathbf{x} ) $ and $ \Psi_{2,n}( r \,|\, \mathbf{x} ) $ follows from \autoref{thm3} and \autoref{thm4}.
\begin{corollary} \label{coro2} 
Let the conditional median $ \mathbf{m}( \mathbf{x} ) $ of $ \mathbf{Y} $ given $ \mathbf{X} = \mathbf{x} $ corresponding to $ \rho( \mathbf{y} \,|\, \mathbf{x} ) $ be unique.
Let $ 0 < r < 1 $, and suppose that for any open set $ G $ with $ G \cap D( \alpha( r ) \,|\, \mathbf{x} ) \neq \emptyset $, we have $ \mathbb{P}\left[ \mathbf{Y} \in G \cap D( \alpha( r ) \,|\, \mathbf{x} ) \right] > 0 $.
Then, under \ref{a4}, \ref{a5}, \ref{a6}, \ref{a7}, \eqref{eq1} and \eqref{eq2}, $ \Psi_{1,n}( r_1, r_2 \,|\, \mathbf{x} ) \stackrel{P}{\longrightarrow} \Psi_1( r_1, r_2 \,|\, \mathbf{x} ) $ as $ n \to \infty $.
If in addition, \ref{a1} holds, $ \mathbb{E}[ \| \mathbf{Y} \|^2 \,|\, \mathbf{X} = \mathbf{z} ] $ is uniformly bounded over $ \mathbf{z} $ and $ \mathbb{E}[ \mathbf{Y} \,|\, \mathbf{X} = \mathbf{z} ] $ is continuous at $ \mathbf{z} = \mathbf{x} $, then
$ \Psi_{2,n}( r \,|\, \mathbf{x} ) \stackrel{P}{\longrightarrow} \Psi_2( r \,|\, \mathbf{x} ) $ as $ n \to \infty $.
\end{corollary}

Using the conditional skewness measures $ \Psi_{1,n}( r_1, r_2 \,|\, \mathbf{x} ) $ and $ \Psi_{2,n}( r \,|\, \mathbf{x} ) $, we propose a test for conditional skewness based on bootstrap.
We describe the procedure of the test based on $ \Psi_{1,n}( r_1, r_2 \,|\, \mathbf{x} ) $. The testing procedure based on $ \Psi_{2,n}( r \,|\, \mathbf{x} ) $ is the same after replacing $ \Psi_{1,n}( r_1, r_2 \,|\, \mathbf{x} ) $ by $ \Psi_{2,n}( r \,|\, \mathbf{x} ) $.
Let $ \mathbf{x} $ be a fixed covariate value. If the conditional distribution of $ \mathbf{Y} $ given $ \mathbf{X} = \mathbf{x} $ is symmetric around a point $ \boldsymbol{\mu} $, then $ \boldsymbol{\mu} $ will be the deepest point, i.e., the conditional median of $ \mathbf{Y} $ given $ \mathbf{X} = \mathbf{x} $, and the conditional distributions of $ \mathbf{Y} - \boldsymbol{\mu} $ and $ \boldsymbol{\mu} - \mathbf{Y} $ given $ \mathbf{X} = \mathbf{x} $ will be identical.
Based on this idea, we devise a test for conditional skewness for the distribution of $ \mathbf{Y} $ given $ \mathbf{X} = \mathbf{x} $.
Our hypotheses for the test of conditional skewness are
\begin{center}
\begin{minipage}[c]{.9\textwidth}
\begin{itemize}
\item[$ H_0 : $] the conditional distribution of $ \mathbf{Y} $ given $ \mathbf{X} = \mathbf{x} $ is symmetric,
\item[$ H_A : $] the conditional distribution of $ \mathbf{Y} $ given $ \mathbf{X} = \mathbf{x} $ is not symmetric.
\end{itemize}
\end{minipage}
\end{center}
Let $ n( \mathbf{x} ) $ be the number of indices $ i $ such that $ W_{i,n}( \mathbf{x} ) > 0 $.
We estimate the conditional median $ \mathbf{m}_n( \mathbf{x} ) $ of $ \mathbf{Y} $ given $ \mathbf{X} = \mathbf{x} $ based on weights $ W_{i,n}( \mathbf{x} ) $, and consider the sample $ \{ \mathbf{Y}_i - \mathbf{m}_n( \mathbf{x} ) \,|\, W_{i,n}( \mathbf{x} ) > 0 \} \cup \{ \mathbf{m}_n( \mathbf{x} ) - \mathbf{Y}_i \,|\, W_{i,n}( \mathbf{x} ) > 0 \} $. From these $ 2 n( \mathbf{x} ) $ observations, we randomly choose $ n( \mathbf{x} ) $ elements $ \mathbf{Z}_1, \ldots, \mathbf{Z}_{n( \mathbf{x} )} $, and consider the collection $ \{ \mathbf{Z}_1 + \mathbf{m}_n( \mathbf{x} ), \ldots, \mathbf{Z}_{n( \mathbf{x} )} + \mathbf{m}_n( \mathbf{x} ) \} $ of $ n( \mathbf{x} ) $ elements as our bootstrap sample. In this bootstrap sample, we calculate the measure of conditional skewness $ \Psi_{1,n}( r_1, r_2 \,|\, \mathbf{x} ) $ from the conditional trimmed mean $ \mathbf{m}_n( r_1 \,|\, \mathbf{x} ) $ and the conditional spread measure $ \Delta_n( r_2 \,|\, \mathbf{x} ) $ of the bootstrap sample and the conditional median $ \mathbf{m}_n( \mathbf{x} ) $ of the original sample. We repeat this procedure a large number of times, say, $ M $ times, and obtain $ M $ values of $ \Psi_{1,n}( r_1, r_2 \,|\, \mathbf{x} ) $.
We reject $ H_0 $ at level $ \alpha $ if the value of $ \Psi_{1,n}( r_1, r_2 \,|\, \mathbf{x} ) $ based on the original sample is higher than the $ ( 1 - \alpha ) $-quantile of the $ M $ values of $ \Psi_{1,n}( r_1, r_2 \,|\, \mathbf{x} ) $ obtained from the bootstrap samples.

The conditional skewness measures and the associated test are demonstrated in a real dataset in \autoref{sec:6}. A size and power study for the test of conditional skewness is presented in simulated models in \autoref{supsec:1}.

\section{Data demonstration} \label{sec:6}
In this section, we demonstrate the conditional central regions, the conditional medians and trimmed means, the measures of conditional spread and skewness, the test of heteroscedasticity and the test of conditional skewness in the Tecator data.
The Tecator data is available in the `caret' package in R. This dataset contains the percentage values of moisture, fat and protein contents of 215 meat samples along with their absorbance spectra in the wavelength range 850--1050 nm. The moisture, the fat and the protein contents were measured by analytical chemistry, while a Tecator Infratec Food and Feed Analyzer was used to record the absorbance spectrum. Being able to predict the nutritional contents of a meat sample from its absorbance spectra is economically beneficial since obtaining the spectra is relatively cheaper.

We first demonstrate the conditional central regions, the conditional medians and the conditional trimmed means.
Here, we consider the pair of fat and protein contents as the response, and the curve of absorbance spectra as the covariate. So, the response is bivariate and the covariate is functional. The covariate is considered to be a random element in the $ L_2 $ space.
We choose the weights $ \{ W_{i,n}( \mathbf{x} ) \} $ in a way which ensures the asymptotic consistency of the estimates. For this, we have employed the nearest neighbor approach, where the integer $ k = k_n $ is
\begin{align*}
k_n = \lfloor ( \log n )^2 \rfloor + 1 ,
\end{align*}
$ \lfloor r \rfloor $ being the largest integer less than or equal to the real number $ r $. So,
\begin{align*}
h( \mathbf{x}, k_n, n ) = \min\left\{ h \middle\arrowvert \sum_{i=1}^{n} \mathbb{I}( d( \mathbf{x}, \mathbf{X}_i ) \le h ) > ( \log n )^2 \right\} .
\end{align*}
This approach is equivalent to the kernel method of choosing the weights when the kernel function is $ K( u ) = \mathbb{I}( 0 \le u \le 1 ) $, and the bandwidth $ h = h( \mathbf{x}, k_n, n ) $.
We noted in \autoref{sec:3} that such a choice ensures the asymptotic consistency of the estimates.
With our choice of the weights, the conditional sample depth $ \rho_n( \cdot \,|\, \mathbf{x} ) $ becomes the corresponding unconditional sample depth based on the local response values of $ \mathbf{x} $. We present the plots of central regions, medians and trimmed means based on the conditional halfspace depth using algorithms by \cite{rousseeuw1996algorithm,rousseeuw1998constructing}, \cite{ruts1996computing}, \cite{rousseeuw1998computing} and \cite{rousseeuw1999bagplot}, which are developed for the unconditional case.
Other depths like the conditional projection depth and the conditional spatial depth produce very similar plots and we do not present them here.
\begin{figure}
\centering
\includegraphics[width=1\linewidth]{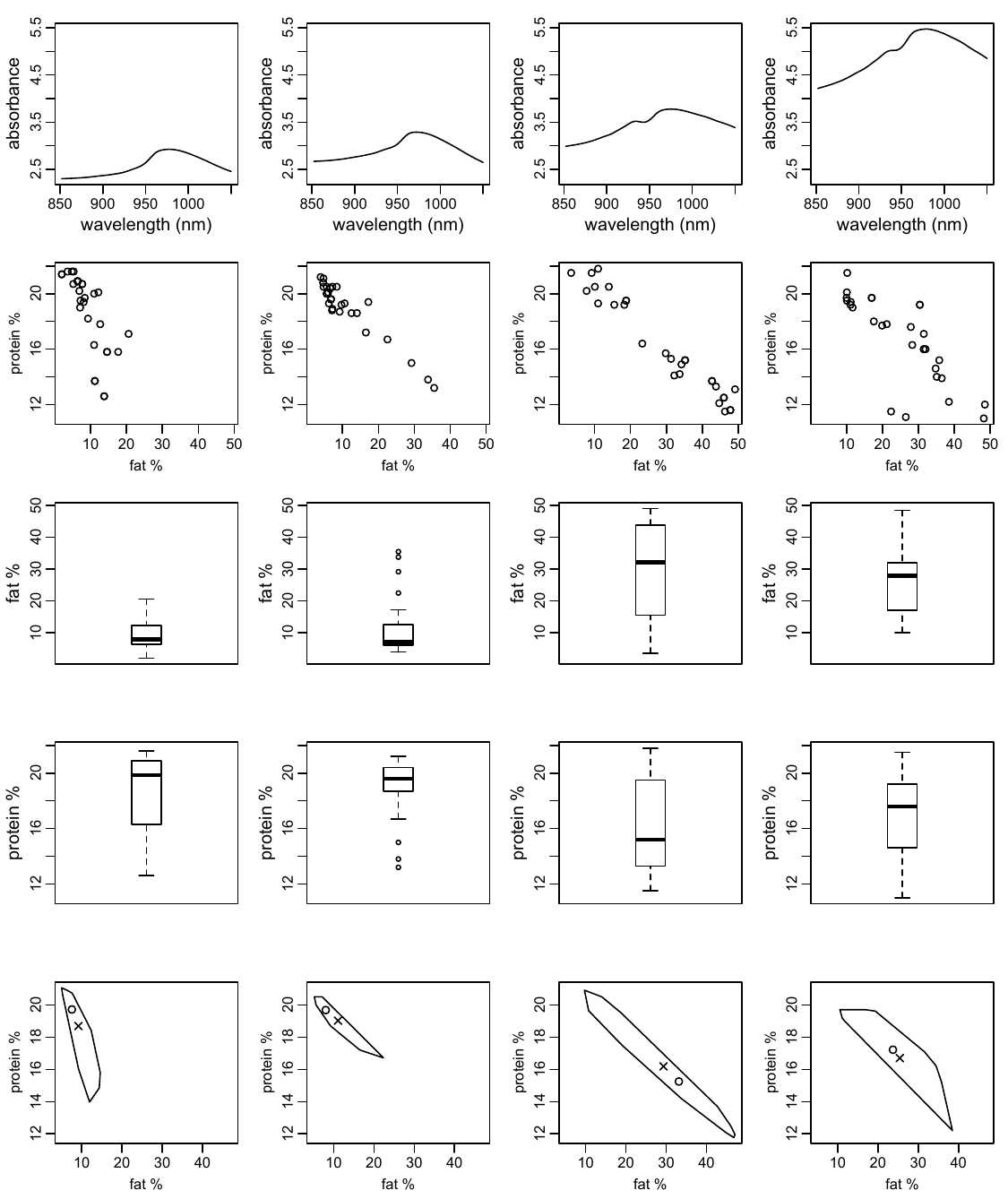}
\caption[Plots of the local response values, local boxplots and the conditional central regions in the Tecator data]{The selected covariate curves ($ 1^\text{st} $ row), the scatter plots of the local response values ($ 2^\text{nd} $ row), the local boxplots for the fat (\%) ($ 3^\text{rd} $ row) and the protein (\%) ($ 4^\text{th} $ row), and the conditional 50\% central regions along with the conditional medians (circle) and the conditional 10\% trimmed means (cross) ($ 5^\text{th} $ row) in the Tecator data.}
\label{fig:sliceplottecator}
\end{figure}
We present the 50\% conditional central regions for the conditional halfspace depth corresponding to four selected covariate values in the Tecator data in \autoref{fig:sliceplottecator}, along with the scatter plots of the local response values and the corresponding local boxplots. The circles and the crosses inside the conditional central regions denote the conditional medians and the 10\% conditional trimmed means, respectively.
The scatter plots of the local response values in \autoref{fig:sliceplottecator} indicate the correlation between the two coordinates of the response variables. The local boxplots and the conditional central regions both demonstrate the heteroscedasticity present in the datasets. The local boxplots, the conditional medians and the conditional trimmed means reflect the variation of conditional skewness of the response over the covariate values in both the datasets.

We now demonstrate the conditional spread measure and the conditional skewness measure using the conditional halfspace depth, the conditional spatial depth and the conditional projection depth.
For the computation of the conditional spatial median, we have used the algorithm given by \cite{chaudhuri1996geometric} for the unconditional case. For conditional projection depth and median computation, we have used the algorithms given in \cite{zuo2011exact} and \cite{liu2013exactly}.
We consider two regression problems here. In the first case, we take the pair of fat and protein contents as a bivariate response, and in the second case, we take the triplet of moisture, fat and protein contents as a trivariate response. The covariate is same in both cases, namely, the curve of absorbance spectra.
For both the datasets, we compute $ \Delta_n( 0.5 \,|\, \mathbf{x} ) $ with $ \mathbf{x} $ varying over the sample covariate values $ \mathbf{X}_i, i = 1, \ldots, n $. We compute the first principal component scores of the observations with respect to the sample dispersion operator of the functional covariate, and denote it as $ P_1 $.
We plot $ \Delta_n( 0.5 \,|\, \mathbf{x} ) $, $ \Psi_{1,n}( 0.1, 0.9 \,|\, \mathbf{x} ) $ and $ \Psi_{2,n}( 0.5 \,|\, \mathbf{x} ) $ for both the bivariate response and the trivariate response against $ P_1 $ in \autoref{fig:spreadtecator} and \autoref{fig:skewtecator}.
We present the plots against the first principal component scores because the first principal component represents the direction of the highest variation in the covariate distribution. Further, the plots of the conditional spread and skewness measures against other principal components do not show clear patterns unlike those visible in the case of the first principal component.
\begin{figure}
\centering
\includegraphics[width=1\linewidth]{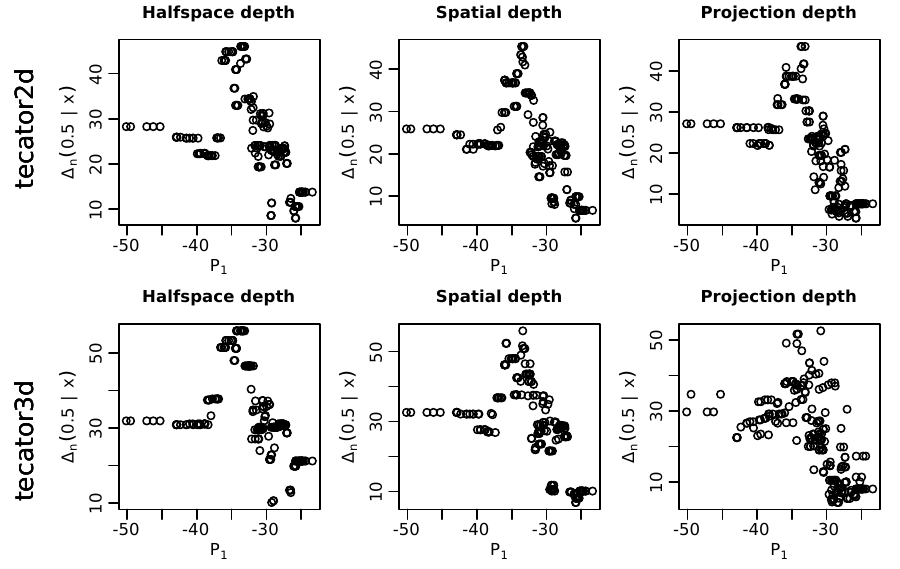}
\caption[Plots of the depth based conditional spread measure in the Tecator data]{Plots of $ \Delta_n( 0.5 \,|\, \mathbf{x} ) $ against $ P_1 $ for bivariate (fat \% and protein \%) and trivariate (moisture \%, fat \% and protein \%) responses in the Tecator data for different depths.}
\label{fig:spreadtecator}
\end{figure}
\begin{figure}
\centering
\includegraphics[width=1\linewidth]{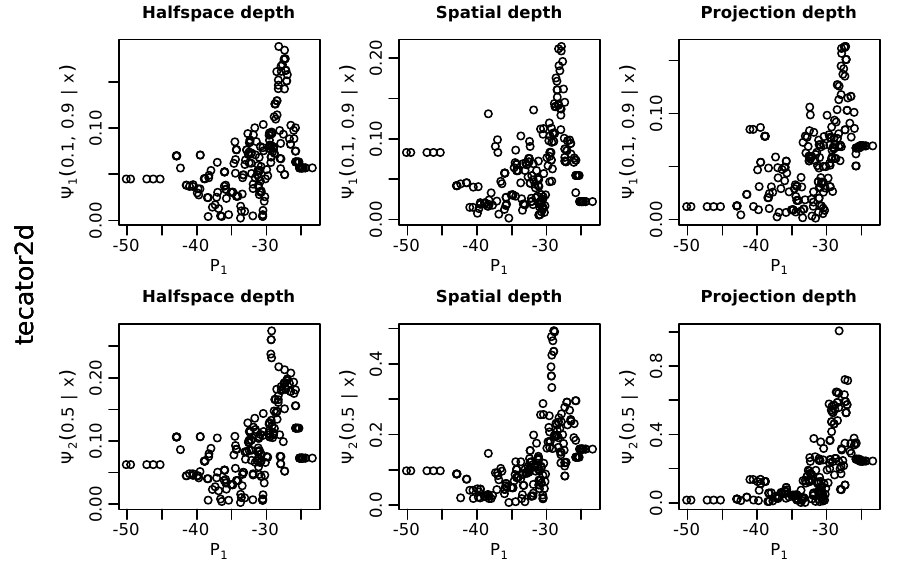}
\includegraphics[width=1\linewidth]{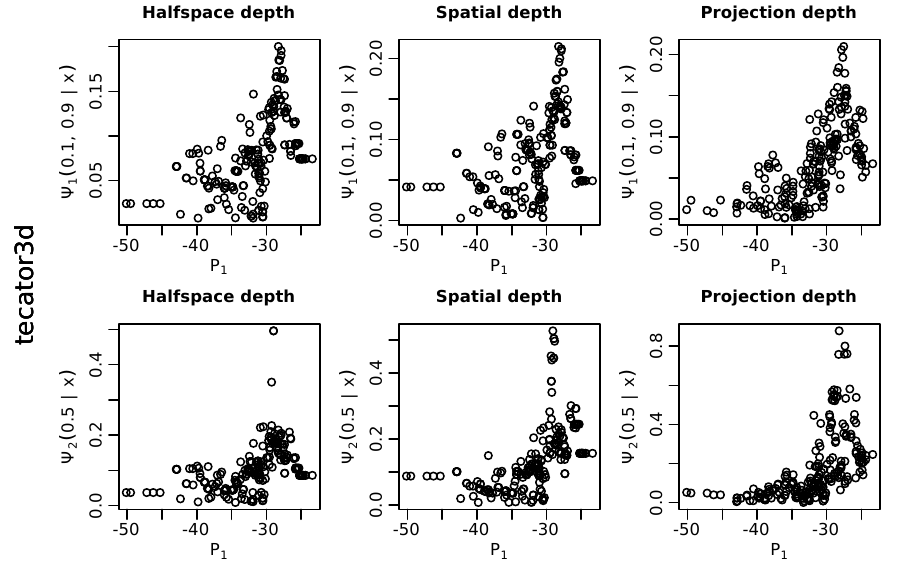}
\caption[Plots of the depth based conditional skewness measures in the Tecator data]{Plots of $ \Psi_{1,n}( 0.1, 0.9 \,|\, \mathbf{x} ) $ and $ \Psi_{2,n}( 0.5 \,|\, \mathbf{x} ) $ against $ P_1 $ for bivariate (fat \% and protein \%) and trivariate (moisture \%, fat \% and protein \%) responses in the Tecator data for different depths.}
\label{fig:skewtecator}
\end{figure}
We can see clear patterns in the plots of $ \Delta_n( 0.5 \,|\, \mathbf{x} ) $ against $ P_1 $ for both the bivariate and the trivariate responses in \autoref{fig:spreadtecator}, which indicates the presence of heteroscedasticity. The conditional skewness measures $ \Psi_{1,n}( 0.1, 0.9 \,|\, \mathbf{x} ) $ and $ \Psi_{2,n}( 0.5 \,|\, \mathbf{x} ) $ also vary with $ P_1 $ in \autoref{fig:skewtecator}, indicating some variation of the conditional skewness present in the sample.

We present the computed p-values for the test of heteroscedasticity based on 500 random permutations and the conditional halfspace depth, the conditional spatial depth and the conditional projection depth in the Tecator data in \autoref{table:hettecator}. The p-values indicate strong presence of heteroscedasticity, and they are consistent with the plots in \autoref{fig:spreadtecator} as both indicate presence of heteroscedasticity.
\begin{table}[h]
\caption{p-values for the test of heteroscedasticity based on several conditional depths in the Tecator data}
\begin{center}
\begin{tabular}{ccccc}
\hline
Data		& Response type		& Halfspace		& Spatial		& Projection	\\\hline
Tecator		& Bivariate			& 0				& 0				& 0				\\
Tecator		& Trivariate		& 0.042			& 0				& 0.002			\\\hline
\end{tabular}
\end{center}
\label{table:hettecator}
\end{table}

To demonstrate the test of conditional skewness in the Tecator data, we randomly select 50 covariate values from the sample, and compute the p-values of the test of conditional skewness at those fixed covariate values for both the bivariate and the trivariate responses. The number of bootstrap samples utilized to compute a p-value is taken as 1000.
We again use the conditional halfspace depth, the conditional spatial depth and the conditional projection depth for the test of conditional skewness.
\begin{figure}
\centering
\includegraphics[width=1\linewidth]{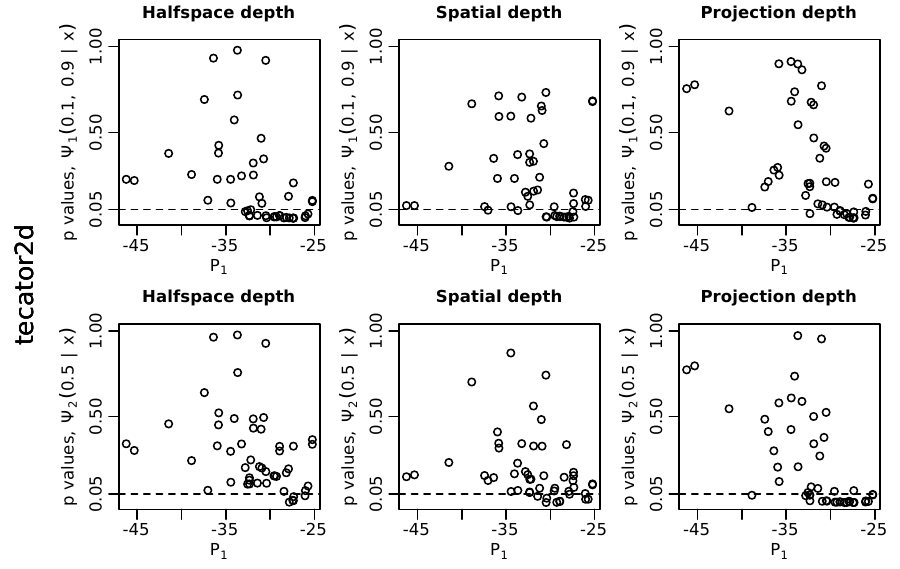}
\includegraphics[width=1\linewidth]{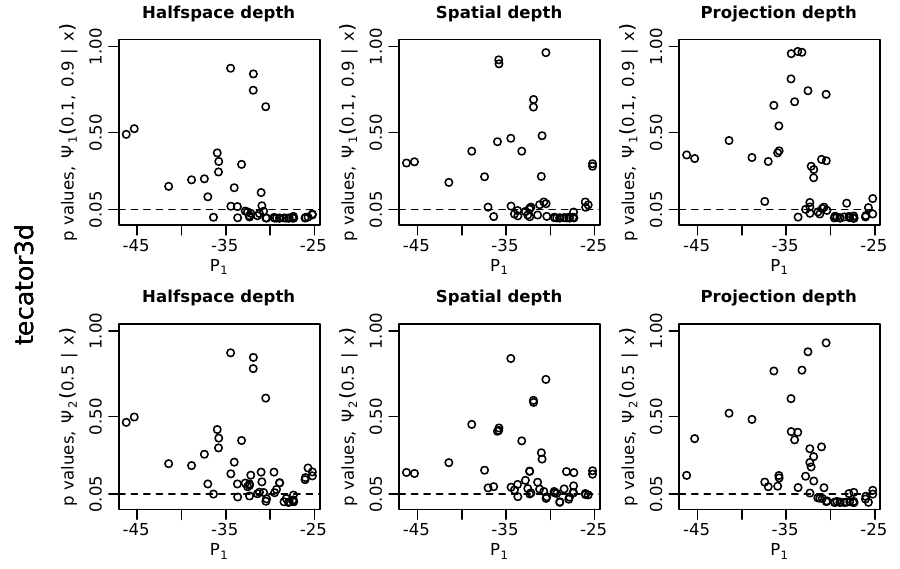}
\caption[Plots of the p-values of the tests of conditional skewness in the Tecator data]{p-values for the test of conditional skewness based on of $ \Psi_{1,n}( 0.1, 0.9 \,|\, \mathbf{x} ) $ and $ \Psi_{2,n}( 0.5 \,|\, \mathbf{x} ) $ against $ P_1 $ for bivariate (fat \% and protein \%) and trivariate (moisture \%, fat \% and protein \%) responses in the Tecator data for different depths.}
\label{fig:skewtesttecator}
\end{figure}
In \autoref{fig:skewtesttecator}, we plot the computed p-values for the test of conditional skewness based on $ \Psi_{1,n}( 0.1, 0.9 \,|\, \mathbf{x} ) $ and $ \Psi_{2,n}( 0.5 \,|\, \mathbf{x} ) $ against the first principal component scores of the corresponding selected covariate curves based on the sample covariance operator of the functional covariate.
In \autoref{fig:skewtesttecator}, we notice a stronger evidence for the presence of conditional skewness in the trivariate response of the Tecator data than in the case of the bivariate response at nominal level 5\%.

\section{Concluding remarks} \label{sec:7}
In this article, nonparametric methods for investigating the conditional distribution of a multivariate response given a covariate are developed based on statistical depths, where the covariate is a random element in a metric space. Unlike traditional mean regression methods, the methods developed here provide information about the center as well as other features of the conditional distribution of the response, like the conditional spread and skewness.

We propose a nonparametric test of heteroscedasticity, where the response is multivariate and the covariate is a random element in a complete separable metric space. To the best of our knowledge, no such test has been investigated in such a setup in the past. Various past studies of heteroscedasticity that are available in the literature consider real valued response and finite dimensional covariate.

We also propose a test of conditional skewness for a multivariate response and a covariate taking values in a metric space. There is an extensive literature on the study of conditional skewness, but almost everywhere, the response considered is real valued and the covariate is finite dimensional.

There are several challenges in extending the depth based regression methods developed here to the case of infinite dimensional responses. The popular finite dimensional depths like the halfspace depth and the projection depth cannot be properly extended for infinite dimensional random elements (see, e.g., Theorem 1 in \cite{chakraborty2014data}). Though, the spatial depth is well-defined for infinite dimensional random elements (see \cite{chakraborty2014spatial}). The asymptotic results proved here use the structure of the finite dimensional space $ \mathbb{R}^p $, and it is difficult to establish them for infinite dimensional spaces.
Some attempts to develop depth based methods for infinite dimensional responses were taken in \cite{chowdhury2019nonparametric}. Only the spatial depth was considered there, and the consistency of the depth based measures like the conditional central regions, conditional medians and trimmed means, measures of conditional spread and skewness were not established. Also, no depth based testing procedure was developed there.

\section{Proofs and mathematical details} \label{sec:8}
Here, we provide the proofs of the theorems in \autoref{sec:3}. The proofs require several lemmas, which are stated in this section, and their proofs are provided in \autoref{supsec:2}. The following lemmas are required for the proof of \autoref{thm1}.
\begin{lemma} \label{lemma1}
Under \ref{a1} and \ref{a3}, for any $ \mu( \cdot \,|\, \mathbf{x} ) $-continuity set $ B $, and given any $ k > 0 $ and $ \epsilon > 0 $, we have
\begin{align*}
& \mathbb{P}\left[ \left| \mu_n( B \,|\, \mathbf{x} ) - \mu( B \,|\, \mathbf{x} ) \right| > \epsilon \middle\arrowvert \mathbf{X}_1, \mathbf{X}_2, \ldots \right]
< 2 n^{-k}
\end{align*}
\emph{almost surely} for all sufficiently large $ n $.
Further, if $ \mathcal{B} $ is a VC class of $ \mu( \cdot \,|\, \mathbf{x} ) $-continuity sets, then
\begin{align*}
\mathbb{P}\left[ \sup\{ \left| \mu_n( B \,|\, \mathbf{x} ) - \mu( B \,|\, \mathbf{x} ) \right| \,|\, B \in \mathcal{B} \} > \epsilon \middle\arrowvert \mathbf{X}_1, \mathbf{X}_2, \ldots \right]
\stackrel{a.s.}{\longrightarrow} 0
\quad
\text{as }
n \to \infty .
\end{align*}
\end{lemma}

\begin{lemma} \label{lemma2}
Let $ \rho( \mathbf{y} \,|\, \mathbf{x} ) $ be the conditional projection depth function described in \autoref{sec:2}. Suppose that there is a neighborhood of $ \mathbf{x} $ such that for all $ \mathbf{z} $ in that neighborhood, the conditional distribution of $ \mathbf{Y} $ given $ \mathbf{X} = \mathbf{z} $ has a continuous positive density $ f( \cdot \,|\, \mathbf{z} ) $ on $ \mathbb{R}^p $, which is continuous in $ \mathbf{z} $.
Then, under \ref{a1}, \ref{a2} and \ref{a3}, we have for all sufficiently large $ n $, $ \sup_{\| \mathbf{y} \| > C} \rho_n( \mathbf{y} \,|\, \mathbf{x} ) \stackrel{a.s.}{\longrightarrow} 0 $ as $ C \to \infty $, and $ \sup_{\| \mathbf{y} \| > C} \rho( \mathbf{y} \,|\, \mathbf{x} ) \to 0 $ as $ C \to \infty $.
Further,
\begin{align*}
& \sup_{\| \mathbf{y} \| \le C} \left| \sup_{\| \mathbf{u} \| = 1} \frac{| \mathbf{u}^t \mathbf{y} - m_n( \mathbf{u}^t \mathbf{Y} \,|\, \mathbf{x} ) |}{m_n( | \mathbf{u}^t \mathbf{Y} - m_n( \mathbf{u}^t \mathbf{Y} \,|\, \mathbf{x} ) | \,|\, \mathbf{x} )} 
- \sup_{\| \mathbf{u} \| = 1} \frac{| \mathbf{u}^t \mathbf{y} - m( \mathbf{u}^t \mathbf{Y} \,|\, \mathbf{x} ) |}{m( | \mathbf{u}^t \mathbf{Y} - m( \mathbf{u}^t \mathbf{Y} \,|\, \mathbf{x} ) | \,|\, \mathbf{x} )} \right| \\
& \stackrel{a.s.}{\longrightarrow} 0
\quad
\text{as } n \to \infty .
\end{align*}
Here, $ m( \mathbf{u}^t \mathbf{Y} \,|\, \mathbf{x} ) $ and $ m( | \mathbf{u}^t \mathbf{Y} - m( \mathbf{u}^t \mathbf{Y} \,|\, \mathbf{x} ) | \,|\, \mathbf{x} ) $ are the conditional medians of $ \mathbf{u}^t \mathbf{Y} $ and $ | \mathbf{u}^t \mathbf{Y} - m( \mathbf{u}^t \mathbf{Y} \,|\, \mathbf{x} ) | $ given $ \mathbf{X} = \mathbf{x} $, respectively. And $ m_n( \mathbf{u}^t \mathbf{Y} \,|\, \mathbf{x} ) $ and $ m_n( | \mathbf{u}^t \mathbf{Y} - m_n( \mathbf{u}^t \mathbf{Y} \,|\, \mathbf{x} ) | \,|\, \mathbf{x} ) $ are the sample analogues of $ m( \mathbf{u}^t \mathbf{Y} \,|\, \mathbf{x} ) $ and $ m( | \mathbf{u}^t \mathbf{Y} - m( \mathbf{u}^t \mathbf{Y} \,|\, \mathbf{x} ) | \,|\, \mathbf{x} ) $, respectively. 
\end{lemma}

\begin{proof}[Proof of \autoref{thm1}]
Let $ \rho( \mathbf{y} \,|\, \mathbf{x} ) $ be the conditional halfspace depth.
Define $ H( \mathbf{u}, \mathbf{y} ) = \{ \mathbf{v} \in \mathbb{R}^p \,|\, \mathbf{u}^t \mathbf{v} \le \mathbf{u}^t \mathbf{y} \} $. From (6.6) in \cite{donoho1992breakdown}, we get
\begin{align*}
\sup_{\mathbf{y} \in \mathbb{R}^p} | \rho_n( \mathbf{y} \,|\, \mathbf{x} ) - \rho( \mathbf{y} \,|\, \mathbf{x} ) | 
\le 
\sup\{ | \mu_n( H( \mathbf{u}, \mathbf{y} ) \,|\, \mathbf{x} ) - \mu( H( \mathbf{u}, \mathbf{y} ) \,|\, \mathbf{x} ) | \allowbreak \,|\, \mathbf{u}, \mathbf{y} \in \mathbb{R}^p \} .
\end{align*}
Also, the class of sets $ \{ H( \mathbf{u}, \mathbf{y} ) \,|\, \mathbf{u}, \mathbf{y} \in \mathbb{R}^p \} $ forms a VC-class \citep[p.~152]{van1996weak}.
Under the conditions in the theorem, every $ H( \mathbf{u}, \mathbf{y} ) $ is a $ \mu( \cdot \,|\, \mathbf{x} ) $-continuity set, and
from \autoref{lemma1}, we get that given any $ \epsilon > 0 $,
\begin{align*}
\mathbb{P}\left[ \sup\{ | \mu_n( H( \mathbf{u}, \mathbf{y} ) \,|\, \mathbf{x} ) - \mu( H( \mathbf{u}, \mathbf{y} ) \,|\, \mathbf{x} ) | \,|\, \mathbf{u}, \mathbf{y} \in \mathbb{R}^p \} > \epsilon \,|\, \mathbf{X}_1, \mathbf{X}_2, \ldots \right]
\stackrel{a.s.}{\longrightarrow} 0
\end{align*}
as $ n \to \infty $. Consequently, from an application of the dominated convergence theorem, we get that \eqref{eq2} is satisfied for the conditional halfspace depth.

Next, let $ \rho( \mathbf{y} \,|\, \mathbf{x} ) $ be the conditional spatial depth.
Define $ \mathbf{Q}( \mathbf{y} \,|\, \mathbf{z} ) = \mathbb{E}[ \| \mathbf{y} - \mathbf{Y} \|^{-1} ( \mathbf{y} - \mathbf{Y} ) \,|\, \mathbf{X} = \mathbf{z} ] $. Under the conditions in the theorem, we have $ \mathbf{Q}( \mathbf{y} \,|\, \mathbf{z} ) $ continuous at $ \mathbf{z} = \mathbf{x} $.
Further,
$ \sup_{ \mathbf{y} \in \mathbb{R}^p } | \rho_n( \mathbf{y} \,|\, \mathbf{x} ) - \rho( \mathbf{y} \,|\, \mathbf{x} ) | 
\le \sup_{ \mathbf{y} \in \mathbb{R}^p } \| \mathbf{Q}_n( \mathbf{y} \,|\, \mathbf{x} ) - \mathbf{Q}( \mathbf{y} \,|\, \mathbf{x} ) \| $, where $ \mathbf{Q}_n( \mathbf{y} \,|\, \mathbf{x} ) = \sum_{i=1}^{n} \| \mathbf{y} - \mathbf{Y}_i \|^{-1} ( \mathbf{y} - \mathbf{Y}_i )W_{i,n}( \mathbf{x} ) $.
Under conditions \ref{a1} and \ref{a3}, the continuity of $ \mathbf{Q}( \mathbf{y} \,|\, \mathbf{z} ) $ at $ \mathbf{z} = \mathbf{x} $, using arguments similar to those in the proof of \autoref{lemma1} in \autoref{supsec:2} and some arguments related to the properties of VC-subgraph classes similar to those in the proof of Theorem 5.5 in \cite{koltchinskii1997m}, we get that given any $ \epsilon > 0 $,
$ \mathbb{P}[ \sup_{ \mathbf{y} \in \mathbb{R}^p } \| \mathbf{Q}_n( \mathbf{y} \,|\, \mathbf{x} ) - \mathbf{Q}( \mathbf{y} \,|\, \mathbf{x} ) \| \allowbreak > \epsilon \,|\, \mathbf{X}_1, \mathbf{X}_2, \ldots ]
\allowbreak
\stackrel{a.s.}{\longrightarrow} 0 $
as $ n \to \infty $. Again, from an application of the dominated convergence theorem, we get that \eqref{eq2} is satisfied for the conditional spatial depth.

Finally, we consider $ \rho( \mathbf{y} \,|\, \mathbf{x} ) $ to be the conditional projection depth.
Under the conditions in the theorem, from \autoref{lemma2} we get that for all sufficiently large $ n $, $ \sup_{\| \mathbf{y} \| > C} \rho_n( \mathbf{y} \,|\, \mathbf{x} ) \stackrel{a.s.}{\longrightarrow} 0 $ as $ C \to \infty $, and $ \sup_{\| \mathbf{y} \| > C} \rho( \mathbf{y} \,|\, \mathbf{x} ) \to 0 $ as $ C \to \infty $. Hence, it is sufficient to show that for any $ C > 0 $, $ \sup_{\| \mathbf{y} \| \le C} | \rho_n( \mathbf{y} \,|\, \mathbf{x} ) - \rho( \mathbf{y} \,|\, \mathbf{x} ) | \stackrel{a.s.}{\longrightarrow} 0 $ as $ n \to \infty $. Under the conditions in the theorem, from \autoref{lemma2} we have
\begin{align*}
& \sup_{\| \mathbf{y} \| \le C} \left| \sup_{\| \mathbf{u} \| = 1} \frac{| \mathbf{u}^t \mathbf{y} - m_n( \mathbf{u}^t \mathbf{Y} \,|\, \mathbf{x} ) |}{m_n( | \mathbf{u}^t \mathbf{Y} - m_n( \mathbf{u}^t \mathbf{Y} \,|\, \mathbf{x} ) | \,|\, \mathbf{x} )} 
- \sup_{\| \mathbf{u} \| = 1} \frac{| \mathbf{u}^t \mathbf{y} - m( \mathbf{u}^t \mathbf{Y} \,|\, \mathbf{x} ) |}{m( | \mathbf{u}^t \mathbf{Y} - m( \mathbf{u}^t \mathbf{Y} \,|\, \mathbf{x} ) | \,|\, \mathbf{x} )} \right| \\
& \stackrel{a.s.}{\longrightarrow} 0
\quad
\text{as } n \to \infty ,
\end{align*}
and this implies $ \sup_{\| \mathbf{y} \| \le C} | \rho_n( \mathbf{y} \,|\, \mathbf{x} ) - \rho( \mathbf{y} \,|\, \mathbf{x} ) | \stackrel{a.s.}{\longrightarrow} 0 $ as $ n \to \infty $. Hence, \eqref{eq2} holds for the conditional projection depth.
\end{proof}

The following results are required for the proof of \autoref{thm2}.
\begin{lemma} \label{lemma3}
Suppose \eqref{eq2} holds. Then, for any $ \epsilon > 0 $, $ \delta > 0 $ and any sequence $ \{ \alpha_n \} $ with $ \alpha_n \stackrel{P}{\longrightarrow} \alpha $ as $ n \to \infty $, where $ \alpha $ is a real number, we have
\begin{align*}
\mathbb{P}\left[ D( \alpha + \delta \,|\, \mathbf{x} ) \subseteq D_n( \alpha_n \,|\, \mathbf{x} ) \subseteq D( \alpha - \delta \,|\, \mathbf{x} ) \right] > 1 - \epsilon
\end{align*}
for all sufficiently large $ n $.
\end{lemma}

\begin{lemma} \label{lemma4}
Suppose \ref{a4}, \eqref{eq1} and \eqref{eq2} hold.
Then, for any $ \beta $, $ \mu_n( D_n( \beta \,|\, \mathbf{x} ) \,|\, \mathbf{x} ) \stackrel{P}{\longrightarrow} \mu( D( \beta \,|\, \mathbf{x} ) \,|\, \mathbf{x} ) $ as $ n \to \infty $.
\end{lemma}

\begin{lemma} \label{lemma5}
Under \ref{a4}, \ref{a5}, \eqref{eq1} and \eqref{eq2}, $ \alpha_n( r ) \stackrel{P}{\longrightarrow} \alpha( r ) $ as $ n \to \infty $ for any $ 0 < r < 1 $.
\end{lemma}

\begin{lemma} \label{lemma6}
Suppose that \ref{a6}, \ref{a7} and \eqref{eq2} hold. Then, for any sequence $ \{ \alpha_n \} $ with $ \alpha_n \stackrel{P}{\longrightarrow} \alpha $ as $ n \to \infty $, $ d_H( D_n( \alpha_n \,|\, \mathbf{x} ), D( \alpha \,|\, \mathbf{x} ) ) \allowbreak \stackrel{P}{\longrightarrow} 0 $ as $ n \to \infty $.
\end{lemma}

\begin{proof}[Proof of \autoref{thm2}]
When \eqref{eq2} holds, from \autoref{lemma3} we get that that for any $ \epsilon, \delta > 0 $ and any sequence $ \{ \alpha_n \} $ with $ \alpha_n \stackrel{P}{\longrightarrow} \alpha $ as $ n \to \infty $, we have
\begin{align}
\mathbb{P}\left[ D( \alpha + \delta \,|\, \mathbf{x} ) \subseteq D_n( \alpha_n \,|\, \mathbf{x} ) \subseteq D( \alpha - \delta \,|\, \mathbf{x} ) \right] > 1 - \epsilon
\label{thm2eq1}
\end{align}
for all sufficiently large $ n $. Also, under \ref{a4}, \ref{a5}, \eqref{eq1}, \eqref{eq2} and for any $ 0 < r < 1 $, from \autoref{lemma5} we get
\begin{align}
\alpha_n( r ) \stackrel{P}{\longrightarrow} \alpha( r )
\quad\text{as } n \to \infty .
\label{thm2eq2}
\end{align}
From \eqref{thm2eq1} and \eqref{thm2eq2}, we get that for any $ \epsilon > 0 $ and any $ 0 < r < 1 $,
\begin{align}
\mathbb{P}\left[ D( \alpha( r ) + \epsilon \,|\, \mathbf{x} ) \subseteq D_n( \alpha_n( r ) \,|\, \mathbf{x} ) \subseteq D( \alpha( r ) - \epsilon \,|\, \mathbf{x} ) \right] \to 1
\quad\text{as } n \to \infty .
\label{thm2eq3}
\end{align}
Next, under \ref{a6}, \ref{a7} and \eqref{eq2}, for any sequence of random variables $ \{ \alpha_n \} $ with $ \alpha_n \stackrel{P}{\longrightarrow} \alpha $ as $ n \to \infty $, from \autoref{lemma6} we get
$ d_H( D_n( \alpha_n \,|\, \mathbf{x} ), D( \alpha \,|\, \mathbf{x} ) ) \allowbreak \stackrel{P}{\longrightarrow} 0 $
as $ n \to \infty $. From this and \eqref{thm2eq2}, we have $ d_H( D_n( \alpha_n( r ) \,|\, \mathbf{x} ) , D( \alpha( r ) \,|\, \mathbf{x} ) ) \stackrel{P}{\longrightarrow} 0 $ as $ n \to \infty $.
\end{proof}

\begin{proof}[Proof of \autoref{coro1}]
The proof follows directly from \eqref{thm2eq3} in the proof of \autoref{thm2}.
\end{proof}

The following results are required for the proofs of \autoref{thm3} and \autoref{thm4}.
\begin{lemma} \label{lemma7}
Let $ \{ \alpha_n \} $ be a sequence of random variables with $ \alpha_n \stackrel{P}{\longrightarrow} \alpha > 0 $ as $ n \to \infty $. Suppose that \ref{a4}, \ref{a6}, \eqref{eq1} and \eqref{eq2} hold. Then,
\begin{align*}
\int \mathbf{y} \mathbb{I}( \mathbf{y} \in D_n( \alpha_n \,|\, \mathbf{x} ) ) \mu_n( \mathrm{d} \mathbf{y} \,|\, \mathbf{x} )
\stackrel{P}{\longrightarrow}
\int \mathbf{y} \mathbb{I}( \mathbf{y} \in D( \alpha \,|\, \mathbf{x} ) ) \mu( \mathrm{d} \mathbf{y} \,|\, \mathbf{x} )
\quad
\text{as } n \to \infty .
\end{align*}
\end{lemma}

\begin{lemma} \label{lemma8}
Suppose \eqref{eq2} and \ref{a4} hold, and
$ \mu( \cdot \,|\, \mathbf{x} ) $ has a continuous density on $ \mathbb{R}^p $, which is everywhere positive.
Then, for any sequence $ \{ \alpha_n \} $ with $ \alpha_n \stackrel{P}{\longrightarrow} \alpha > 0 $ as $ n \to \infty $,
$ \lambda( D_n( \alpha_n \,|\, \mathbf{x} ) )
\stackrel{P}{\longrightarrow} \lambda( D( \alpha \,|\, \mathbf{x} ) ) $
as $ n \to \infty $, where $ \lambda( \cdot ) $ is the Lebesgue measure on $ \mathbb{R}^p $.
\end{lemma}

\begin{lemma} \label{lemma9}
For any pair of sets $ A $ and $ B $, $ | \text{Diameter}( A ) - \text{Diameter}( B ) | \le 2 d_H( A, B ) $.
\end{lemma}

\begin{lemma} \label{lemma10}
Let $ \{ \alpha_n \} $ be a sequence of random variables with $ \alpha_n \stackrel{P}{\longrightarrow} \alpha $ as $ n \to \infty $.
Define $ D_n'( \alpha_n \,|\, \mathbf{x} ) = \{ \mathbf{Y}_i \,|\, \rho_n( \mathbf{Y}_i \,|\, \mathbf{x} ) \ge \alpha_n;\; i = 1, \ldots, n \} $.
Suppose that for any open set $ G $ with $ G \cap D( \alpha \,|\, \mathbf{x} ) \neq \emptyset $, we have $ \mathbb{P}\left[ \mathbf{Y} \in G \cap D( \alpha \,|\, \mathbf{x} ) \right] > 0 $.
Then, under \ref{a6}, \ref{a7} and \eqref{eq2}, $ d_H( D_n'( \alpha_n \,|\, \mathbf{x} ), D( \alpha \,|\, \mathbf{x} ) ) \stackrel{P}{\longrightarrow} 0 $ as $ n \to \infty $.
\end{lemma}

\begin{proof}[Proof of \autoref{thm3}]
\textit{Proof of (a):}
Suppose \ref{a4} and \eqref{eq2} are satisfied, and
$ \mu( \cdot \,|\, \mathbf{x} ) $ has a continuous positive density on $ \mathbb{R}^p $. Then, from \autoref{lemma8}, for any sequence $ \{ \alpha_n \} $ with $ \alpha_n \stackrel{P}{\longrightarrow} \alpha > 0 $ as $ n \to \infty $, we have
$ \lambda( D_n( \alpha_n \,|\, \mathbf{x} ) ) )
\stackrel{P}{\longrightarrow} \lambda( D( \alpha \,|\, \mathbf{x} ) ) $
as $ n \to \infty $, where $ \lambda( \cdot ) $ is the Lebesgue measure on $ \mathbb{R}^p $.
Using this fact and \autoref{lemma5}, the proof of part (a) is complete.

\textit{Proof of (b):}
For any pair of sets $ A $ and $ B $, from \autoref{lemma9} we get $ | \text{Diameter}( A ) - \text{Diameter}( B ) | \le 2 d_H( A, B ) $.
Define
\begin{align*}
D_n'( \alpha_n \,|\, \mathbf{x} ) = \{ \mathbf{Y}_i \,|\, \rho_n( \mathbf{Y}_i \,|\, \mathbf{x} ) \ge \alpha_n;\; i = 1, \ldots, n \} .
\end{align*}
Note that
\begin{align*}
& \Delta_n( r \,|\, \mathbf{x} ) = \text{Diameter}( D_n'( \alpha_n( r ) \,|\, \mathbf{x} ) ) 
\quad\text{and}\quad \Delta( r \,|\, \mathbf{x} ) = \text{Diameter}( D( \alpha( r ) \,|\, \mathbf{x} ) ) .
\end{align*}
Let $ \{ \alpha_n \} $ be a sequence of random variables with $ \alpha_n \stackrel{P}{\longrightarrow} \alpha $ as $ n \to \infty $.
Suppose that for any open set $ G $ with $ G \cap D( \alpha \,|\, \mathbf{x} ) \neq \emptyset $, we have
\begin{align*}
\mathbb{P}\left[ \mathbf{Y} \in G \cap D( \alpha \,|\, \mathbf{x} ) \right] > 0 .
\end{align*}
Then, under \ref{a6}, \ref{a7} and \eqref{eq2}, from \autoref{lemma10}, we get
\begin{align}
d_H\left( D_n'( \alpha_n \,|\, \mathbf{x} ), D( \alpha \,|\, \mathbf{x} ) \right) \stackrel{P}{\longrightarrow} 0
\quad\text{as } n \to \infty .
\label{thm3eq1}
\end{align}
The proof of part (b) follows from \eqref{thm3eq1} and \autoref{lemma5}.
\end{proof}

\begin{proof}[Proof of \autoref{thm4}]
We denote $ \rho( \mathcal{M}( \mathbf{x} ) \,|\, \mathbf{x} ) = \max_\mathbf{y} \rho( \mathbf{y} \,|\, \mathbf{x} ) = \rho( \mathbf{m} \,|\, \mathbf{x} ) $ for any $ \mathbf{m} \in \mathcal{M}( \mathbf{x} ) $.
From \eqref{eq2}, we get that given any $ \epsilon, \delta > 0 $, for $ \mathbf{m} \in \mathcal{M}( \mathbf{x} ) $,
\begin{align*}
& \mathbb{P}[ \rho_n( \mathbf{m}_n( \mathbf{x} ) \,|\, \mathbf{x} ) \le \rho( \mathbf{m}_n( \mathbf{x} ) \,|\, \mathbf{x} ) + \epsilon \le \rho( \mathcal{M}( \mathbf{x} ) \,|\, \mathbf{x} ) + \epsilon ]
> 1 - \frac{\delta}{2} \\
& \text{and}\quad \mathbb{P}[ \rho_n( \mathbf{m}_n( \mathbf{x} ) \,|\, \mathbf{x} ) \ge \rho_n( \mathbf{m} \,|\, \mathbf{x} ) \ge \rho( \mathcal{M}( \mathbf{x} ) \,|\, \mathbf{x} ) - \epsilon ]
> 1 - \frac{\delta}{2}
\end{align*}
for all sufficiently large $ n $, which implies that $ \rho_n( \mathbf{m}_n( \mathbf{x} ) \,|\, \mathbf{x} ) \stackrel{P}{\longrightarrow} \rho( \mathcal{M}( \mathbf{x} ) \,|\, \mathbf{x} ) $ as $ n \to \infty $.
From this fact and \autoref{lemma3}, we have, given any $ \epsilon, \delta > 0 $,
\begin{align*}
\mathbb{P}\left[ D_n( \rho_n( \mathbf{m}_n( \mathbf{x} ) \,|\, \mathbf{x} ) \,|\, \mathbf{x} ) \subseteq D( \rho( \mathcal{M}( \mathbf{x} ) \,|\, \mathbf{x} ) - \delta \,|\, \mathbf{x} ) \right] > 1 - \epsilon
\end{align*}
for all sufficiently large $ n $.
Note that $ D_n( \rho_n( \mathbf{m}_n( \mathbf{x} ) \,|\, \mathbf{x} ) \,|\, \mathbf{x} ) $ is the collection of all sample conditional medians, and $ \mathcal{M}( \mathbf{x} ) = D( \rho( \mathcal{M}( \mathbf{x} ) \,|\, \mathbf{x} ) \,|\, \mathbf{x} ) \subset D( \rho( \mathcal{M}( \mathbf{x} ) \,|\, \mathbf{x} ) - \delta \,|\, \mathbf{x} ) $ for all $ \delta > 0 $.
Whenever $ D_n( \rho_n( \mathbf{m}_n( \mathbf{x} ) \,|\, \mathbf{x} ) \,|\, \mathbf{x} ) \subseteq D( \rho( \mathcal{M}( \mathbf{x} ) \,|\, \mathbf{x} ) - \delta \,|\, \mathbf{x} ) $, we have
\begin{align*}
d_H( D_n( \rho_n( \mathbf{m}_n( \mathbf{x} ) \,|\, \mathbf{x} ) \,|\, \mathbf{x} ) , \mathcal{M}( \mathbf{x} ) )
\le
d_H( D( \rho( \mathcal{M}( \mathbf{x} ) \,|\, \mathbf{x} ) - \delta \,|\, \mathbf{x} ) , \mathcal{M}( \mathbf{x} ) ) .
\end{align*}
From \eqref{lemma6eq2} in the proof of \autoref{lemma6} in \autoref{supsec:2}, we get
\begin{align*}
\lim_{\delta \to 0^+} d_H( \mathcal{M}( \mathbf{x} ) , D( \rho( \mathcal{M}( \mathbf{x} ) \,|\, \mathbf{x} ) - \delta \,|\, \mathbf{x} ) ) = 0 .
\end{align*}
Therefore, $ d_H( D_n( \rho_n( \mathbf{m}_n( \mathbf{x} ) \,|\, \mathbf{x} ) \,|\, \mathbf{x} ) , \mathcal{M}( \mathbf{x} ) )
\stackrel{P}{\longrightarrow}
0 $ as $ n \to \infty $, which implies
$ \inf_{ \mathbf{m} \in \mathcal{M}( \mathbf{x} ) } \| \mathbf{m}_n( \mathbf{x} ) - \mathbf{m} \|
\stackrel{P}{\longrightarrow} 0 $
as $ n \to \infty $ for any sequence of sample conditional medians $ \{ \mathbf{m}_n( \mathbf{x} ) \} $.

We now proceed to prove the consistency of $ \mathbf{m}_n( r \,|\, \mathbf{x} ) $.
Let $ \{ \alpha_n \} $ be a sequence of random variables with $ \alpha_n \stackrel{P}{\longrightarrow} \alpha > 0 $ as $ n \to \infty $. Suppose that \ref{a4}, \ref{a6}, \eqref{eq1} and \eqref{eq2} are satisfied. Then, from \autoref{lemma7}, we have
\begin{align}
\int \mathbf{y} \mathbb{I}( \mathbf{y} \in D_n( \alpha_n \,|\, \mathbf{x} ) ) \mu_n( \mathrm{d} \mathbf{y} \,|\, \mathbf{x} )
\stackrel{P}{\longrightarrow}
\int \mathbf{y} \mathbb{I}( \mathbf{y} \in D( \alpha \,|\, \mathbf{x} ) ) \mu( \mathrm{d} \mathbf{y} \,|\, \mathbf{x} )
\label{thm4eq1}
\end{align}
as $ n \to \infty $.
For any $ 0 < r < 1 $, from conditions \ref{a4} and \ref{a5}, we get that $ \alpha( 1 - r ) > 0 $. So, from \autoref{lemma5} and \eqref{thm4eq1}, we have
$ \int \mathbf{y} \mathbb{I}( \mathbf{y} \in D_n( \alpha_n( 1 - r ) \,|\, \mathbf{x} ) ) \mu_n( \mathrm{d} \mathbf{y} \,|\, \mathbf{x} )
\stackrel{P}{\longrightarrow}
\int \mathbf{y} \mathbb{I}( \mathbf{y} \in D( \alpha( 1 - r ) \,|\, \mathbf{x} ) ) \mu( \mathrm{d} \mathbf{y} \,|\, \mathbf{x} ) $
as $ n \to \infty $.
Also, from \ref{a4}, \eqref{eq1} and \autoref{thm2}, it follows that
$ \mu_n( D_n( \alpha_n( 1 - r ) \,|\, \mathbf{x} ) \,|\, \mathbf{x} ) \stackrel{P}{\longrightarrow} \mu( D( \alpha( 1 - r ) \,|\, \mathbf{x} ) \,|\, \mathbf{x} ) $
as $ n \to \infty $.
Hence, $ \mathbf{m}_n( r \,|\, \mathbf{x} ) \stackrel{P}{\longrightarrow} \mathbf{m}( r \,|\, \mathbf{x} ) $ as $ n \to \infty $.
\end{proof}

\begin{proof}[Proof of \autoref{coro2}]
From \autoref{thm3} and \autoref{thm4}, we get $ \Psi_{1,n}( r_1, r_2 \,|\, \mathbf{x} ) \allowbreak \stackrel{P}{\longrightarrow} \Psi_1( r_1, r_2 \,|\, \mathbf{x} ) $ as $ n \to \infty $.
Next, note that
\begin{align*}
& \mathbb{E}\left[ \left\| \mathbf{M}_n( \mathbf{x} ) - \mathbf{M}( \mathbf{x} ) \right\|^2 \middle\arrowvert \mathbf{X}_1, \mathbf{X}_2, \ldots \right] \\
& = \sum_{i=1}^{n} \mathbb{E}\left[ \left\| \mathbf{Y}_i - \mathbb{E}\left[ \mathbf{Y}_i \mid \mathbf{X}_i \right] \right\|^2 \middle\arrowvert \mathbf{X}_i \right] W_{i,n}^2( \mathbf{x} ) \\
& \quad + \left\| \sum_{i=1}^{n} \mathbb{E}\left[ \mathbf{Y}_i \mid \mathbf{X}_i \right] W_{i,n}( \mathbf{x} ) - \mathbb{E}\left[ \mathbf{Y} \mid \mathbf{X} = \mathbf{x} \right] \right\|^2 .
\end{align*}
Now, from \ref{a1} and using the condition that $ \mathbb{E}\left[ \left\| \mathbf{Y} \right\|^2 \middle\arrowvert \mathbf{X} = \mathbf{z} \right] $ is uniformly bounded over $ \mathbf{z} $, we get $ \sum_{i=1}^{n} \mathbb{E}\left[ \left\| \mathbf{Y}_i - \mathbb{E}\left[ \mathbf{Y}_i \mid \mathbf{X}_i \right] \right\|^2 \middle\arrowvert \mathbf{X}_i \right] W_{i,n}^2( \mathbf{x} ) \stackrel{a.s.}{\longrightarrow} 0 $ as $ n \to \infty $.
Again using \ref{a1} and the continuity of $ \mathbb{E}\left[ \mathbf{Y} \mid \mathbf{X} = \mathbf{z} \right] $ at $ \mathbf{z} = \mathbf{x} $, we get that $ \left\| \sum_{i=1}^{n} \mathbb{E}\left[ \mathbf{Y}_i \mid \mathbf{X}_i \right] W_{i,n}( \mathbf{x} ) - \mathbb{E}\left[ \mathbf{Y} \mid \mathbf{X} = \mathbf{x} \right] \right\|^2 \stackrel{a.s.}{\longrightarrow} 0 $ as $ n \to \infty $.
Therefore, we have $ \mathbb{E}\left[ \left\| \mathbf{M}_n( \mathbf{x} ) - \mathbf{M}( \mathbf{x} ) \right\|^2 \middle\arrowvert \mathbf{X}_1, \mathbf{X}_2, \ldots \right] \stackrel{a.s.}{\longrightarrow} 0 $ as $ n \to \infty $, which implies $ \mathbb{E}\left[ \left\| \mathbf{M}_n( \mathbf{x} ) - \mathbf{M}( \mathbf{x} ) \right\|^2 \right] \to 0 $ as $ n \to \infty $. Consequently, from the Markov inequality, we have
\begin{align}
\mathbf{M}_n( \mathbf{x} ) \stackrel{P}{\longrightarrow} \mathbf{M}( \mathbf{x} )
\quad\text{as } n \to \infty .
\label{coro2eq1}
\end{align}
From \eqref{coro2eq1}, \autoref{thm3} and \autoref{thm4}, we get $ \Psi_{2,n}( r \,|\, \mathbf{x} ) \stackrel{P}{\longrightarrow} \Psi_2( r \,|\, \mathbf{x} ) $ as $ n \to \infty $.
\end{proof}

\appendix

\section{Demonstration in simulated models} \label{supsec:1}
We demonstrate here the performance of the tests of heteroscedasticity and conditional skewness in simulated models.
First, we demonstrate the test of heteroscedasticity in two simulation models.
Let $ \Sigma_p = (( \sigma_{ij} ))_{p \times p} $ with $ \sigma_{ij} = 0.5 + 0.5 \mathbb{I}( i = j ) $, and $ a $ be a positive number.
The covariate $ \mathbf{X} $ is a random function with $ \mathbf{X}( t ) = \mathbf{B} e^{t} $, where $ \mathbf{B} \sim Uniform[ 0, 1 ] $ and $ t \in [ 0, 1 ] $. $ \mathbf{X} $ is considered as a random element in $ L_2[ 0, 1 ] $ and let $ \| \cdot \|_2 $ denote the $ L_2 $-norm. In the first simulation model, we take the conditional distribution of the response $ \mathbf{Y} $ given $ \mathbf{X} $ to be bivariate normal with mean vector $ ( 0, 0 ) $ and dispersion matrix $ ( 1 + a \| \mathbf{X} \|_2 ) \Sigma_2 $, where $ a $ is a positive parameter. In the second simulation model, the conditional distribution of the response $ \mathbf{Y} $ given $ \mathbf{X} $ is trivariate normal with mean vector $ ( 0, 0, 0 ) $ and dispersion matrix $ ( 1 + a \| \mathbf{X} \|_2 ) \Sigma_2 $.
We estimate the power of the test of heteroscedasticity for the conditional halfspace depth, the conditional spatial depth and the conditional projection depth, for different sample sizes and different values of the constant $ a $ based on 500 independent replications, with the number of random permutations in each replication being 1000. The estimated powers for different depths, parameter values and sample sizes are presented in \autoref{table:1}.
\begin{table}[h]
\caption{Estimated powers for the test of heteroscedasticity in simulated models based on different depths.}
\begin{center}
\begin{tabular}{|cc|ccccc|ccccc|}
\hline
\multicolumn{12}{|c|}{Estimated powers for the test based on halfspace depth}					\\\hline
\multirow{2}{*}{$ n $}	& \multirow{2}{*}{level}	& \multicolumn{5}{c|}{Bivariate $ \mathbf{Y} $}			& \multicolumn{5}{c|}{Trivariate $ \mathbf{Y} $}	\\
&							& $ a = 0 $	& $ a = 1 $	& $ a = 2 $	& $ a = 3 $	& $ a = 4 $	& $ a = 0 $	& $ a = 1 $	& $ a = 2 $	& $ a = 3 $	& $ a = 4 $	\\\hline
200						& 5\%						& 0.048    	& 0.334   	& 0.616    	& 0.764    	& 0.838		& 0.046    	& 0.18    	& 0.306  	& 0.43    	& 0.498		\\
200						& 1\%						& 0.012		& 0.142		& 0.348		& 0.504		& 0.622		& 0.002    	& 0.078    	& 0.128    	& 0.21    	& 0.27		\\\hline
500						& 5\%						& 0.058    	& 0.696    	& 0.966    	& 0.996    	& 0.996		& 0.044    	& 0.394    	& 0.808    	& 0.96    	& 0.992		\\
500						& 1\%						& 0.006		& 0.478		& 0.882		& 0.974		& 0.99		& 0.01    	& 0.14    	& 0.512    	& 0.718    	& 0.858		\\\hline
\hline
\multicolumn{12}{|c|}{Estimated powers for the test based on spatial depth}						\\\hline
\multirow{2}{*}{$ n $}	& \multirow{2}{*}{level}	& \multicolumn{5}{c|}{Bivariate $ \mathbf{Y} $}			& \multicolumn{5}{c|}{Trivariate $ \mathbf{Y} $}	\\
&							& $ a = 0 $	& $ a = 1 $	& $ a = 2 $	& $ a = 3 $	& $ a = 4 $	& $ a = 0 $	& $ a = 1 $	& $ a = 2 $	& $ a = 3 $	& $ a = 4 $	\\\hline
200						& 5\%						& 0.068    	& 0.324   	& 0.656    	& 0.822    	& 0.9		& 0.042    	& 0.396    	& 0.7    	& 0.856    	& 0.922		\\
200						& 1\%						& 0.022		& 0.162		& 0.428		& 0.586		& 0.72		& 0.012    	& 0.182    	& 0.458    	& 0.67    	& 0.764		\\\hline
500						& 5\%						& 0.078    	& 0.714    	& 0.968    	& 0.994    	& 0.998		& 0.048    	& 0.756    	& 0.978    	& 0.998    	& 0.996		\\
500						& 1\%						& 0.012		& 0.476		& 0.912		& 0.986		& 0.994		& 0.002    	& 0.502    	& 0.894    	& 0.974    	& 0.994		\\\hline
\hline
\multicolumn{12}{|c|}{Estimated powers for the test based on projection depth}					\\\hline
\multirow{2}{*}{$ n $}	& \multirow{2}{*}{level}	& \multicolumn{5}{c|}{Bivariate $ \mathbf{Y} $}			& \multicolumn{5}{c|}{Trivariate $ \mathbf{Y} $}	\\
&							& $ a = 0 $	& $ a = 1 $	& $ a = 2 $	& $ a = 3 $	& $ a = 4 $	& $ a = 0 $	& $ a = 1 $	& $ a = 2 $	& $ a = 3 $	& $ a = 4 $	\\\hline
200						& 5\%						& 0.044    	& 0.33   	& 0.59    	& 0.714    	& 0.79		& 0.05    	& 0.346    	& 0.618    	& 0.772    	& 0.842		\\
200						& 1\%						& 0.006		& 0.144		& 0.348		& 0.512		& 0.556		& 0.008    	& 0.152    	& 0.398    	& 0.508   	& 0.626		\\\hline
500						& 5\%						& 0.052    	& 0.592    	& 0.918    	& 0.988    	& 0.996		& 0.04    	& 0.638    	& 0.942    	& 0.986    	& 0.99		\\
500						& 1\%						& 0.016		& 0.354		& 0.77		& 0.93		& 0.978		& 0.012    	& 0.404    	& 0.83    	& 0.928    	& 0.966		\\\hline
\end{tabular}
\end{center}
\label{table:1}
\end{table}
Note that the estimated power corresponding to $ a = 0 $ is nothing but the estimated size of the test.

Next, we demonstrate the test of conditional skewness in two simulation models.
In both the models described below, the functional covariate $ \mathbf{X} $ is the same as that in the heteroscedastic simulation model with $ \mathbf{X}( t ) = \mathbf{B} e^{t} $, where $ \mathbf{B} \sim Uniform[ 0, 1 ] $ and $ t \in [ 0, 1 ] $.
We consider mixture distributions with unequal proportions as skewed models.
Let $ \Sigma_p $ and $ a $ be as defined before.
In the first simulation model for conditional skewness, we consider a bivariate normal random vector $ \mathbf{Z} $ with mean $ ( 0, 0 ) $ and dispersion matrix $ \Sigma_2 $. We take the conditional distribution of the response $ \mathbf{Y} $ given $ \mathbf{X} $ to be a mixture of bivariate normal distributions given by $ 5 R \mathbf{Z} + ( 1 - R ) ( ( a \| \mathbf{X} \|_2 ( 1, 1 ) ) + \mathbf{Z} ) $, where $ R $ is a Bernoulli random variable with $ \mathbb{P}[ R = 1 ] = 0.7 $, and $ R $ and $ \mathbf{Z} $ are independent of $ \mathbf{X} $ and $ \mathbf{Y} $. In the second simulation model for conditional skewness, we consider a trivariate normal random vector $ \mathbf{Z} $ following a trivariate normal distribution with mean $ ( 0, 0, 0 ) $ and dispersion matrix $ \Sigma_3 $, and take the conditional distribution of the response $ \mathbf{Y} $ given $ \mathbf{X} $ to be a mixture of trivariate normal distributions given by $ 5 R \mathbf{Z} + ( 1 - R ) ( ( a \| \mathbf{X} \|_2 ( 1, 1, 1 ) ) + \mathbf{Z} ) $, where $ R $ is the Bernoulli random variable described before, and $ R $ and $ \mathbf{Z} $ are again independent of $ \mathbf{X} $ and $ \mathbf{Y} $.

To carry out the test of conditional skewness, we randomly generate 50 values from the covariate distribution and estimate the power of the test at those chosen covariate values. The p-value of the test is computed based on 1000 bootstrap samples, and the estimated powers are computed based on 500 independent replications.
For the first simulation model, we plot the estimated power of the test of conditional skewness based on several depths for different values of the parameter $ a $ and sample sizes in \autoref{fig:skewtest90pcsim2d} and \autoref{fig:skewtestlocalsim2d} against the first principal component scores of the fixed covariate values based on the sample covariance operator of the functional covariate, denoted as $ P_1 $.
Note that the parameter value $ a = 0 $ corresponds to the estimated size of the test.
Similarly, for the second simulation model, we plot the estimated power of the test based on several depths for different values of $ a $ and sample sizes in \autoref{fig:skewtest90pcsim3d} and \autoref{fig:skewtestlocalsim3d} against $ P_1 $.
\begin{figure}
\centering
\includegraphics[width=1\linewidth]{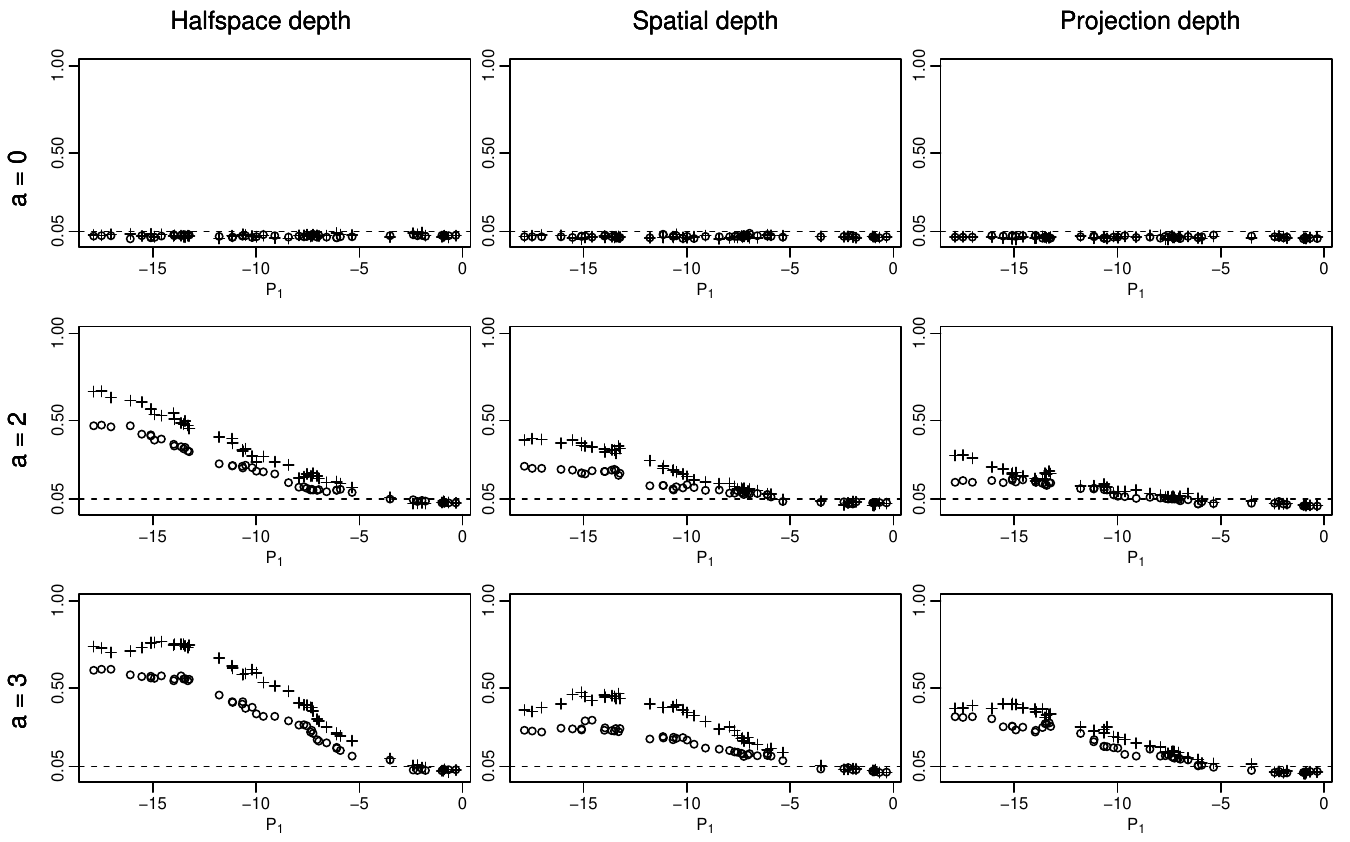}
\caption[Plots of the estimated power for the test of conditional skewness based on $ \Psi_{1,n}( 0.1, 0.9 \,|\, \mathbf{x} ) $ in the simulated data with bivariate response]{Estimated powers at nominal level 5\% at different values of the parameter $ a $ and different depths for the test of conditional skewness based on $ \Psi_{1,n}( 0.1, 0.9 \,|\, \mathbf{x} ) $ for sample size $ 200 $ (`$ \circ $' sign) and $ n = 500 $ (`$ + $' sign) in the simulated data with bivariate response.}
\label{fig:skewtest90pcsim2d}
\end{figure}
\begin{figure}
\centering
\includegraphics[width=1\linewidth]{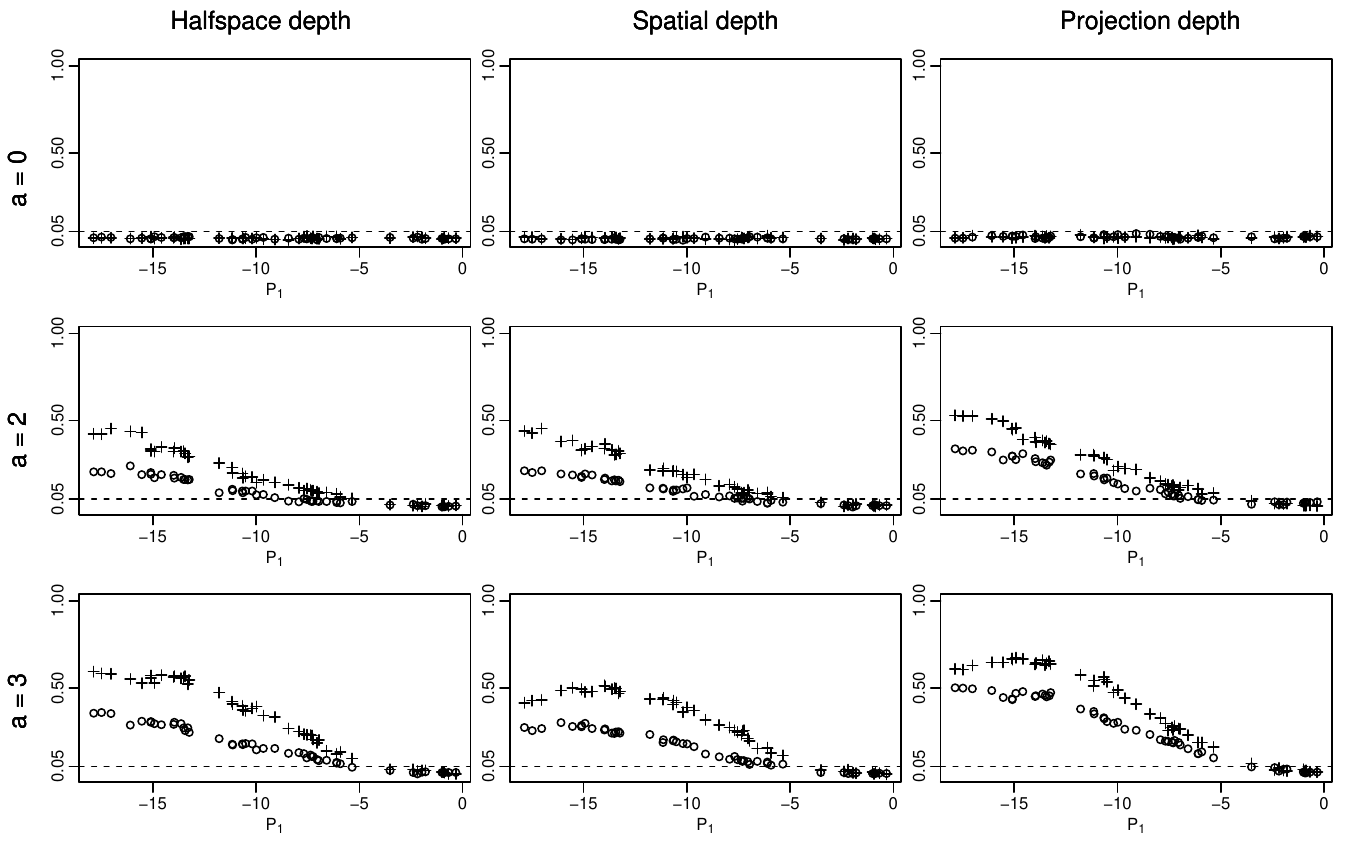}
\caption[Plots of the estimated power for the test of conditional skewness based on $ \Psi_{2,n}( 0.5 \,|\, \mathbf{x} ) $ in the simulated data with bivariate response]{Estimated powers at nominal level 5\% at different values of the parameter $ a $ and different depths for the test of conditional skewness based on $ \Psi_{2,n}( 0.5 \,|\, \mathbf{x} ) $ for sample size $ 200 $ (`$ \circ $' sign) and $ n = 500 $ (`$ + $' sign) in the simulated data with bivariate response.}
\label{fig:skewtestlocalsim2d}
\end{figure}
\begin{figure}
\centering
\includegraphics[width=1\linewidth]{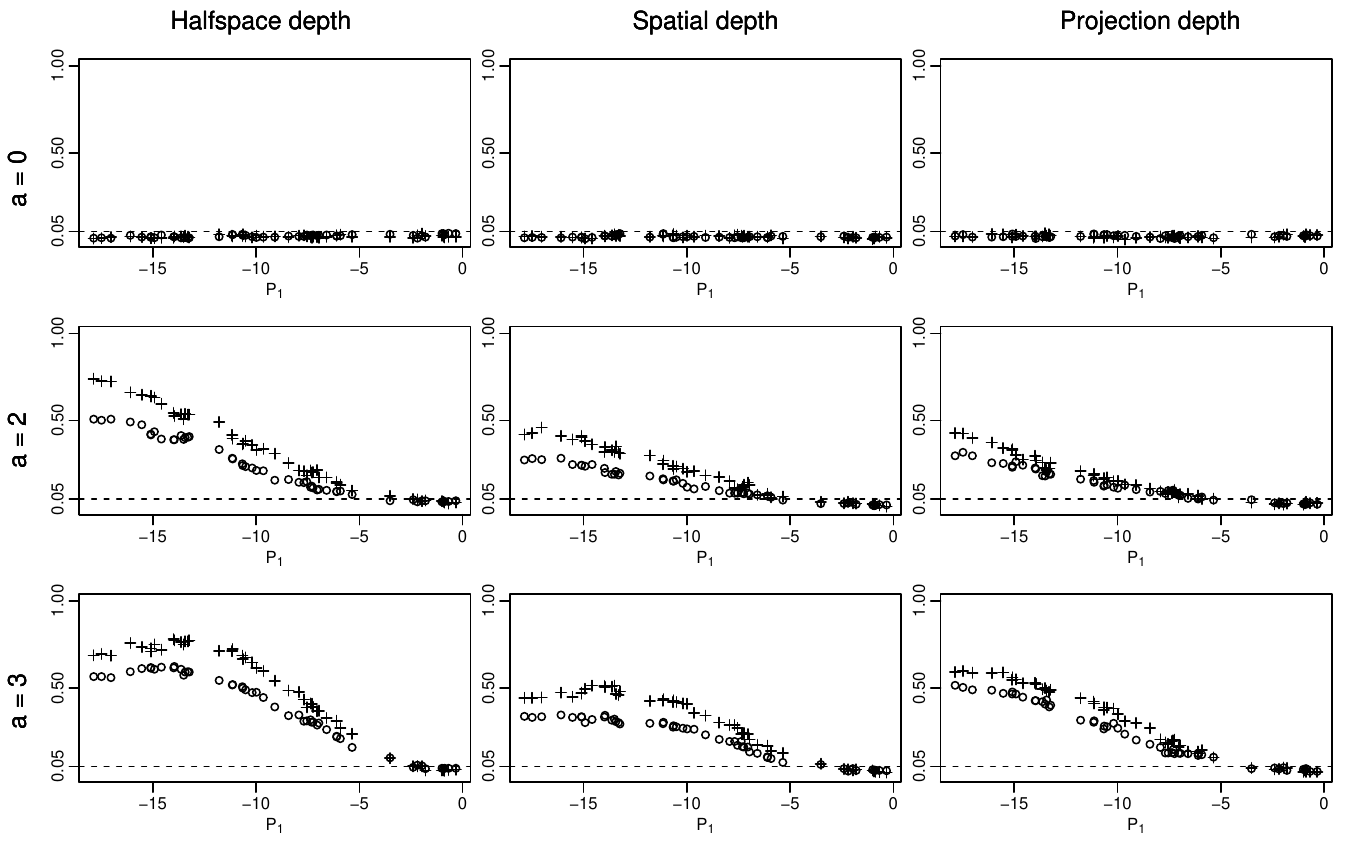}
\caption[Plots of the estimated power for the test of conditional skewness based on $ \Psi_{1,n}( 0.1, 0.9 \,|\, \mathbf{x} ) $ in the simulated data with trivariate response]{Estimated powers at nominal level 5\% at different values of the parameter $ a $ and different depths for the test of conditional skewness based on $ \Psi_{1,n}( 0.1, 0.9 \,|\, \mathbf{x} ) $ for sample size $ 200 $ (`$ \circ $' sign) and $ n = 500 $ (`$ + $' sign) in the simulated data with trivariate response.}
\label{fig:skewtest90pcsim3d}
\end{figure}
\begin{figure}
\centering
\includegraphics[width=1\linewidth]{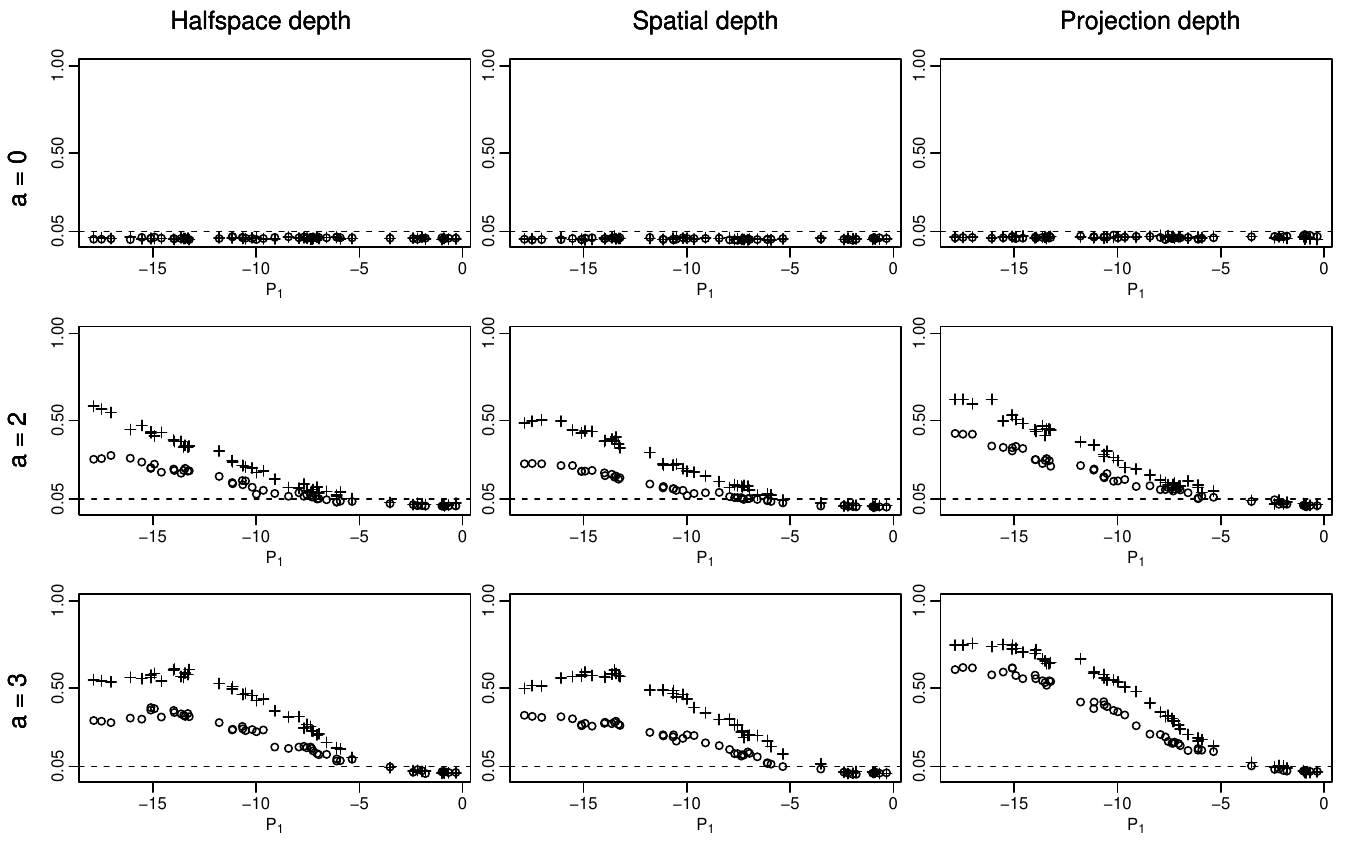}
\caption[Plots of the estimated power for the test of conditional skewness based on $ \Psi_{2,n}( 0.5 \,|\, \mathbf{x} ) $ in the simulated data with trivariate response]{Estimated powers at nominal level 5\% at different values of the parameter $ a $ and different depths for the test of conditional skewness based on $ \Psi_{2,n}( 0.5 \,|\, \mathbf{x} ) $ for sample size $ 200 $ (`$ \circ $' sign) and $ n = 500 $ (`$ + $' sign) in the simulated data with trivariate response.}
\label{fig:skewtestlocalsim3d}
\end{figure}

We plotted the estimated powers against the first principal component scores because the first principal component of the covariate captures the major portion of the variation present in the covariate distribution. Further, we did not find any clear pattern in the plots of the estimated power against the second principal component scores.

From \autoref{fig:skewtest90pcsim2d}, \autoref{fig:skewtestlocalsim2d}, \autoref{fig:skewtest90pcsim3d} and \autoref{fig:skewtestlocalsim3d}, we observe that for the conditional halfspace depth, the test based on $ \Psi_{1,n}( 0.1, 0.9 \,|\, \mathbf{x} ) $ exhibits relatively better performance than the test based on $ \Psi_{2,n}( 0.5 \,|\, \mathbf{x} ) $. However, the reverse is true for conditional projection depth. For conditional spatial depth, the two tests exhibit similar performance.
For both the test of heteroscedasticity and the test of conditional skewness, it appears that the estimated powers of the tests tend to one as the sample size grows. In our choice of the models, the levels of heteroscedasticity and conditional skewness increase with the parameter $ a $, which is also reflected in the growing powers of the tests as the value of $ a $ increases.

\section{Additional mathematical details} \label{supsec:2}
In this section, we provide the proofs of all the lemmas stated in \autoref{sec:8}.
\begin{proof}[Proof of \autoref{lemma1}]
For any $ \epsilon > 0 $,
\begin{align}
& \mathbb{P}\left[ \left| \mu_n( B \,|\, \mathbf{x} ) - \mu( B \,|\, \mathbf{x} ) \right| > \epsilon \middle\arrowvert \mathbf{X}_1, \mathbf{X}_2, \ldots \right]
\nonumber\\
& \le \mathbb{P}\left[ \left| \sum_{i=1}^{n} \left[ \mathbb{I}( \mathbf{Y}_i \in B ) - \mu( B \,|\, \mathbf{X}_i ) \right] W_{i,n}( \mathbf{x} ) \right| > \frac{\epsilon}{2} \middle\arrowvert \mathbf{X}_1, \mathbf{X}_2, \ldots \right] \nonumber\\
& + \mathbb{P}\left[ \left| \sum_{i=1}^{n} \left[ \mu( B \,|\, \mathbf{X}_i ) - \mu( B \,|\, \mathbf{x} ) \right] W_{i,n}( \mathbf{x} ) \right| > \frac{\epsilon}{2} \middle\arrowvert \mathbf{X}_1, \mathbf{X}_2, \ldots \right] .
\label{lemma1eq1}
\end{align}
Now, from the Bernstein inequality (see \citet[p.~95, Lemma A]{serfling2009approximation}), we have
\begin{align*}
& \mathbb{P}\left[ \left| \sum_{i=1}^{n} \left[ \mathbb{I}( \mathbf{Y}_i \in B ) - \mu( B \,|\, \mathbf{X}_i ) \right] W_{i,n}( \mathbf{x} ) \right| > \frac{\epsilon}{2} \middle\arrowvert \mathbf{X}_1, \mathbf{X}_2, \ldots \right]
\le \nonumber\\
& 2 \exp\left[ - \frac{\epsilon^2}{2 \sum_{i=1}^{n} W_{i,n}^2( \mathbf{x} ) + 2 \epsilon \max_{1 \le i \le n} W_{i,n}( \mathbf{x} )} \right] .
\end{align*}
From \ref{a1}, we get that there is an integer $ N_1 $ such that for all $ n \ge N_1 $,
\begin{align*}
& ( \log n ) \sum_{i=1}^{n} W_{i,n}^2( \mathbf{x} ) 
< \frac{\epsilon^2}{4 k} 
\quad \text{and} \quad
( \log n ) \max_{1 \le i \le n} W_{i,n}( \mathbf{x} ) 
< \frac{\epsilon}{4 k} 
\end{align*}
\emph{almost surely}.
$ N_1 $ depends on $ \mathbf{X}_1, \mathbf{X}_2, \ldots $, $ k $ and $ \epsilon $.
So,
\begin{align}
& \mathbb{P}\left[ \left| \sum_{i=1}^{n} \left[ \mathbb{I}( \mathbf{Y}_i \in B ) - \mu( B \,|\, \mathbf{X}_i ) \right] W_{i,n}( \mathbf{x} ) \right| > \frac{\epsilon}{2} \middle\arrowvert \mathbf{X}_1, \mathbf{X}_2, \ldots \right]
< 2 \exp\left[ - k \log n \right]
\label{lemma1eq2}
\end{align}
\emph{almost surely} for all sufficiently large $ n $.
On the other hand, for any $ \delta > 0 $, we have
\begin{align*}
& \mathbb{P}\left[ \left| \sum_{i=1}^{n} \left[ \mu( B \,|\, \mathbf{X}_i ) - \mu( B \,|\, \mathbf{x} ) \right] W_{i,n}( \mathbf{x} ) \right| > \frac{\epsilon}{2} \middle\arrowvert \mathbf{X}_1, \mathbf{X}_2, \ldots \right]
\\
& \le
\mathbb{P}\left[ \left| \sum_{i=1}^{n} \left[ \mu( B \,|\, \mathbf{X}_i ) - \mu( B \,|\, \mathbf{x} ) \right] \mathbb{I}( d( \mathbf{x}, \mathbf{X}_i ) \ge \delta ) W_{i,n}( \mathbf{x} ) \right| > \frac{\epsilon}{4} \middle\arrowvert \mathbf{X}_1, \mathbf{X}_2, \ldots \right] \\
& \quad
+ \mathbb{P}\left[ \left| \sum_{i=1}^{n} \left[ \mu( B \,|\, \mathbf{X}_i ) - \mu( B \,|\, \mathbf{x} ) \right] \mathbb{I}( d( \mathbf{x}, \mathbf{X}_i ) < \delta ) W_{i,n}( \mathbf{x} ) \right| > \frac{\epsilon}{4} \middle\arrowvert \mathbf{X}_1, \mathbf{X}_2, \ldots \right] .
\end{align*}
Since $ B $ is a $ \mu( \cdot \,|\, \mathbf{x} ) $-continuity set, from \ref{a3}, we can find $ \delta > 0 $ such that $ \left| \mu( B \,|\, \mathbf{z} ) - \mu( B \,|\, \mathbf{x} ) \right| < (\epsilon / 5) $ whenever $ d( \mathbf{x}, \mathbf{z} ) < \delta $. From \ref{a1}, we can find an integer $ N_2 $ such that for all $ n \ge N_2 $,
\begin{align*}
\sum_{i=1}^{n} W_{i,n}( \mathbf{x} ) \mathbb{I}( d( \mathbf{x}, \mathbf{X}_i ) \ge \delta )
< \frac{\epsilon}{5}
\end{align*}
\emph{almost surely}.
$ N_2 $ depends on $ \mathbf{X}_1, \mathbf{X}_2, \ldots $, $ \delta $ and $ \epsilon $.
Therefore, for all $ n \ge N_2 $,
\begin{align*}
& \mathbb{P}\left[ \left| \sum_{i=1}^{n} \left[ \mu( B \,|\, \mathbf{X}_i ) - \mu( B \,|\, \mathbf{x} ) \right] \mathbb{I}( d( \mathbf{x}, \mathbf{X}_i ) \ge \delta ) W_{i,n}( \mathbf{x} ) \right| > \frac{\epsilon}{4} \middle\arrowvert \mathbf{X}_1, \mathbf{X}_2, \ldots \right]
\\
& \le
\mathbb{P}\left[ \sum_{i=1}^{n} \mathbb{I}( d( \mathbf{x}, \mathbf{X}_i ) \ge \delta ) W_{i,n}( \mathbf{x} ) > \frac{\epsilon}{4} \middle\arrowvert \mathbf{X}_1, \mathbf{X}_2, \ldots \right]
= 0 
\quad
\text{\emph{almost surely}},
\\
\intertext{and}
& \mathbb{P}\left[ \left| \sum_{i=1}^{n} \left[ \mu( B \,|\, \mathbf{X}_i ) - \mu( B \,|\, \mathbf{x} ) \right] \mathbb{I}( d( \mathbf{x}, \mathbf{X}_i ) < \delta ) W_{i,n}( \mathbf{x} ) \right| > \frac{\epsilon}{4} \middle\arrowvert \mathbf{X}_1, \mathbf{X}_2, \ldots \right]
\\
& \le
\mathbb{P}\left[ \sum_{i=1}^{n} \left| \mu( B \,|\, \mathbf{X}_i ) - \mu( B \,|\, \mathbf{x} ) \right| \mathbb{I}( d( \mathbf{x}, \mathbf{X}_i ) < \delta ) W_{i,n}( \mathbf{x} ) > \frac{\epsilon}{4} \middle\arrowvert \mathbf{X}_1, \mathbf{X}_2, \ldots \right]
\\
& = 0
\quad\text{\emph{almost surely}} .
\end{align*}
Hence,
\begin{align}
\mathbb{P}\left[ \left| \sum_{i=1}^{n} \left[ \mu( B \,|\, \mathbf{X}_i ) - \mu( B \,|\, \mathbf{x} ) \right] W_{i,n}( \mathbf{x} ) \right| > \frac{\epsilon}{2} \middle\arrowvert \mathbf{X}_1, \mathbf{X}_2, \ldots \right]
= 0 
\label{lemma1eq3}
\end{align}
\emph{almost surely} for all sufficiently large $ n $.
Therefore, from \eqref{lemma1eq1}, \eqref{lemma1eq2} and \eqref{lemma1eq3}, we get
\begin{align}
\mathbb{P}\left[ \left| \mu_n( B \,|\, \mathbf{x} ) - \mu( B \,|\, \mathbf{x} ) \right| > \epsilon \middle\arrowvert \mathbf{X}_1, \mathbf{X}_2, \ldots \right]
<
2 n^{-k}
\label{lemma1eq4}
\end{align}
\emph{almost surely} for all sufficiently large $ n $.

Now, if $ \mathcal{B} $ is a VC class of $ \mu( \cdot \,|\, \mathbf{x} ) $-continuity sets, using \eqref{lemma1eq4} and some standard arguments involving the properties of VC-classes (see, e.g., \citet[p.~13--24]{pollard1984convergence}), we get
$ \mathbb{P}\left[ \sup\{ \left| \mu_n( B \,|\, \mathbf{x} ) - \mu( B \,|\, \mathbf{x} ) \right| \,|\, B \in \mathcal{B} \} > \epsilon \middle\arrowvert \mathbf{X}_1, \mathbf{X}_2, \ldots \right]
\stackrel{a.s.}{\longrightarrow} 0 $
as $ n \to \infty $.
\end{proof}

\begin{proof}[Proof of \autoref{lemma2}]
Under the condition in the lemma on the conditional density of $ \mathbf{Y} $ given $ \mathbf{X} = \mathbf{z} $, we get that there is a neighborhood of $ \mathbf{x} $ such that for $ \mathbf{z} $ lying in that neighborhood and for every $ \mathbf{u} $, $ \mathbf{u}^t \mathbf{Y} $ given $ \mathbf{X} = \mathbf{z} $ has a continuous strictly increasing conditional distribution function, which is also continuous as a function of $ \mathbf{u} $.
Consequently, $ m( \mathbf{u}^t \mathbf{Y} \,|\, \mathbf{x} ) $ and $ m( | \mathbf{u}^t \mathbf{Y} - m( \mathbf{u}^t \mathbf{Y} \,|\, \mathbf{x} ) | \,|\, \mathbf{x} ) $ are unique for every $ \mathbf{u} $ and continuous as functions of $ \mathbf{u} $.
Denote the conditional distribution function of $ \mathbf{u}^t \mathbf{Y} $ given $ \mathbf{X} = \mathbf{x} $ as $ F( \cdot \,|\, \mathbf{u}, \mathbf{x} ) $.
Let $ F_n( \cdot \,|\, \mathbf{u}, \mathbf{x} ) $ be the corresponding weighted empirical distribution function of $ \mathbf{u}^t \mathbf{Y} $, which is defined by
\begin{align*}
F_n( v \,|\, \mathbf{u}, \mathbf{x} )
= \sum_{i=1}^{n} \mathbb{I}( \mathbf{u}^t \mathbf{Y}_i \le v ) W_{i,n}( \mathbf{x} ) .
\end{align*}
Under the condition on the conditional density of $ \mathbf{Y} $ given $ \mathbf{X} = \mathbf{z} $, using \ref{a1} and \ref{a3} and arguments similar to those in the proof of \autoref{lemma1}, we get that
$ F_n( v \,|\, \mathbf{u}, \mathbf{x} ) \stackrel{a.s.}{\longrightarrow} F( v \,|\, \mathbf{u}, \mathbf{x} ) $ as $ n \to \infty $ for all $ v \in \mathbb{R} $ and $ \mathbf{u} \in \mathbb{R}^p $.
Recall that a halfspace in $ \mathbb{R}^p $ is defined as $ H( \mathbf{u}, v ) = \{ \mathbf{y} \in \mathbb{R}^p \,|\, \mathbf{u}^t \mathbf{y} \le v \} $, where $ \mathbf{u} \in \mathbb{R}^p $ and $ v \in \mathbb{R} $. Since the class of all halfspaces $ \{ H( \mathbf{u}, v ) \,|\, \mathbf{u} \in \mathbb{R}^p ,\; v \in \mathbb{R} \} $ forms a VC class (see \citet[p.~152]{van1996weak}), using some standard arguments involving the properties of VC-classes (see, e.g., \citet[p.~13--24]{pollard1984convergence}), we get that
\begin{align}
\sup_{ \mathbf{u} \in \mathbb{R}^p ,\; v \in \mathbb{R} } \left| F_n( v \,|\, \mathbf{u}, \mathbf{x} ) - F( v \,|\, \mathbf{u}, \mathbf{x} ) \right|
\stackrel{a.s.}{\longrightarrow}
0
\quad \text{as } n \to \infty .
\label{lemma2eq1}
\end{align}
Now, from \eqref{lemma2eq1} and \ref{a1}, we have
\begin{align}
& \sup_{ \| \mathbf{u} \| = 1 } \left| F( m_n( \mathbf{u}^t \mathbf{Y} \,|\, \mathbf{x} ) \,|\, \mathbf{u}, \mathbf{x} ) - F( m( \mathbf{u}^t \mathbf{Y} \,|\, \mathbf{x} ) \,|\, \mathbf{u}, \mathbf{x} ) \right|
\nonumber\\
& \le
\sup_{ \| \mathbf{u} \| = 1 } \left| F_n( m_n( \mathbf{u}^t \mathbf{Y} \,|\, \mathbf{x} ) \,|\, \mathbf{u}, \mathbf{x} ) - 0.5 \right| \nonumber\\
& \quad
+ \sup_{ \| \mathbf{u} \| = 1 } \left| F_n( m_n( \mathbf{u}^t \mathbf{Y} \,|\, \mathbf{x} ) \,|\, \mathbf{u}, \mathbf{x} ) - F( m_n( \mathbf{u}^t \mathbf{Y} \,|\, \mathbf{x} ) \,|\, \mathbf{u}, \mathbf{x} ) \right|
\nonumber\\
& \le
\max_{1 \le i \le n} W_{i,n}( \mathbf{x} ) 
+ \sup_{ \mathbf{u} \in \mathbb{R}^p ,\; v \in \mathbb{R} } \left| F_n( v \,|\, \mathbf{u}, \mathbf{x} ) - F( v \,|\, \mathbf{u}, \mathbf{x} ) \right|
\stackrel{a.s.}{\longrightarrow} 0
\label{lemma2eq2}
\end{align}
as $ n \to \infty $.
Since $ F( \cdot \,|\, \mathbf{u}, \mathbf{x} ) $ is a continuous and strictly increasing distribution function, its inverse $ F^{-1}( \cdot \,|\, \mathbf{u}, \mathbf{x} ) $ is well-defined. Further, under the condition in the lemma on the conditional density of $ \mathbf{Y} $ given $ \mathbf{X} = \mathbf{z} $, $ F^{-1}( \cdot \,|\, \mathbf{u}, \mathbf{x} ) $ is also continuous in $ \mathbf{u} $.
Since $ \{ \mathbf{u} \in \mathbb{R}^p \,|\, \| \mathbf{u} \| = 1 \} $ is compact, $ F^{-1}( \cdot \,|\, \mathbf{u}, \mathbf{x} ) $ is uniformly continuous in $ \mathbf{u} \in \{ \mathbf{u} \in \mathbb{R}^p \,|\, \| \mathbf{u} \| = 1 \} $. Hence, from \eqref{lemma2eq2}, we get
\begin{align}
& \sup_{ \| \mathbf{u} \| = 1 } \left| m_n( \mathbf{u}^t \mathbf{Y} \,|\, \mathbf{x} ) - m( \mathbf{u}^t \mathbf{Y} \,|\, \mathbf{x} ) \right|
\nonumber\\
& = \sup_{ \| \mathbf{u} \| = 1 } \left| F^{-1}( F( m_n( \mathbf{u}^t \mathbf{Y} \,|\, \mathbf{x} ) \,|\, \mathbf{u}, \mathbf{x} ) \,|\, \mathbf{u}, \mathbf{x} ) - F^{-1}( F( m( \mathbf{u}^t \mathbf{Y} \,|\, \mathbf{x} ) \,|\, \mathbf{u}, \mathbf{x} ) \,|\, \mathbf{u}, \mathbf{x} ) \right|
\nonumber\\
& \stackrel{a.s.}{\longrightarrow} 0
\quad \text{as } n \to \infty .
\label{lemma2eq3}
\end{align}
One can similarly show that
\begin{align}
\sup_{ \| \mathbf{u} \| = 1 } \left| m_n( | \mathbf{u}^t \mathbf{Y} - m_n( \mathbf{u}^t \mathbf{Y} \,|\, \mathbf{x} ) | \,|\, \mathbf{x} ) - m( | \mathbf{u}^t \mathbf{Y} - m( \mathbf{u}^t \mathbf{Y} \,|\, \mathbf{x} ) | \,|\, \mathbf{x} ) \right|
\stackrel{a.s.}{\longrightarrow} 0
\label{lemma2eq4}
\end{align}
as $ n \to \infty $.
Now, from \eqref{lemma2eq3} and \eqref{lemma2eq4}, we get that for all sufficiently large $ n $,
\begin{align*}
& \sup_{\| \mathbf{u} \| = 1} \frac{| \mathbf{u}^t \mathbf{y} - m_n( \mathbf{u}^t \mathbf{Y} \,|\, \mathbf{x} ) |}{m_n( | \mathbf{u}^t \mathbf{Y} - m_n( \mathbf{u}^t \mathbf{Y} \,|\, \mathbf{x} ) | \,|\, \mathbf{x} )} \\
& \ge \sup_{\| \mathbf{u} \| = 1} \frac{| | \mathbf{u}^t \mathbf{y} | - m_n( \mathbf{u}^t \mathbf{Y} \,|\, \mathbf{x} ) |}{m_n( | \mathbf{u}^t \mathbf{Y} - m_n( \mathbf{u}^t \mathbf{Y} \,|\, \mathbf{x} ) | \,|\, \mathbf{x} )} \\
& \ge \frac{\| \mathbf{y} \| - \sup_{\| \mathbf{u} \| = 1} m_n( \mathbf{u}^t \mathbf{Y} \,|\, \mathbf{x} )}{\sup_{\| \mathbf{u} \| = 1} m_n( | \mathbf{u}^t \mathbf{Y} - m_n( \mathbf{u}^t \mathbf{Y} \,|\, \mathbf{x} ) | \,|\, \mathbf{x} )}
\stackrel{a.s.}{\longrightarrow} \infty \text{ as } \| \mathbf{y} \| \to \infty .
\end{align*}
So, for all sufficiently large $ n $, $ \sup_{\| \mathbf{y} \| > C} \rho_n( \mathbf{y} \,|\, \mathbf{x} ) \stackrel{a.s.}{\longrightarrow} 0 $ as $ C \to \infty $, and $ \sup_{\| \mathbf{y} \| > C} \rho( \mathbf{y} \,|\, \mathbf{x} ) \allowbreak \to 0 $ as $ C \to \infty $.

Next, we have
\begin{align*}
& \sup_{\| \mathbf{y} \| \le C} \left| \sup_{\| \mathbf{u} \| = 1} \frac{| \mathbf{u}^t \mathbf{y} - m_n( \mathbf{u}^t \mathbf{Y} \,|\, \mathbf{x} ) |}{m_n( | \mathbf{u}^t \mathbf{Y} - m_n( \mathbf{u}^t \mathbf{Y} \,|\, \mathbf{x} ) | \,|\, \mathbf{x} )} 
- \sup_{\| \mathbf{u} \| = 1} \frac{| \mathbf{u}^t \mathbf{y} - m( \mathbf{u}^t \mathbf{Y} \,|\, \mathbf{x} ) |}{m( | \mathbf{u}^t \mathbf{Y} - m( \mathbf{u}^t \mathbf{Y} \,|\, \mathbf{x} ) | \,|\, \mathbf{x} )} \right| \\
& \le \sup_{\| \mathbf{y} \| \le C} \sup_{\| \mathbf{u} \| = 1} \left| \frac{\mathbf{u}^t \mathbf{y}}{m_n( | \mathbf{u}^t \mathbf{Y} - m_n( \mathbf{u}^t \mathbf{Y} \,|\, \mathbf{x} ) | \,|\, \mathbf{x} )} 
- \frac{\mathbf{u}^t \mathbf{y}}{m( | \mathbf{u}^t \mathbf{Y} - m( \mathbf{u}^t \mathbf{Y} \,|\, \mathbf{x} ) | \,|\, \mathbf{x} )} \right| \\
& \quad + \sup_{\| \mathbf{y} \| \le C} \sup_{\| \mathbf{u} \| = 1} \left| \frac{m_n( \mathbf{u}^t \mathbf{Y} \,|\, \mathbf{x} )}{m_n( | \mathbf{u}^t \mathbf{Y} - m_n( \mathbf{u}^t \mathbf{Y} \,|\, \mathbf{x} ) | \,|\, \mathbf{x} )} 
- \frac{m( \mathbf{u}^t \mathbf{Y} \,|\, \mathbf{x} )}{m( | \mathbf{u}^t \mathbf{Y} - m( \mathbf{u}^t \mathbf{Y} \,|\, \mathbf{x} ) | \,|\, \mathbf{x} )} \right| .
\end{align*}
Since $ m_n( \mathbf{u}^t \mathbf{Y} \,|\, \mathbf{x} ) \stackrel{a.s.}{\longrightarrow} m( \mathbf{u}^t \mathbf{Y} \,|\, \mathbf{x} ) $ and $ m_n( | \mathbf{u}^t \mathbf{Y} - m_n( \mathbf{u}^t \mathbf{Y} \,|\, \mathbf{x} ) | \,|\, \mathbf{x} ) \stackrel{a.s.}{\longrightarrow} m( | \mathbf{u}^t \mathbf{Y} - m( \mathbf{u}^t \mathbf{Y} \,|\, \mathbf{x} ) | \,|\, \mathbf{x} ) $ as $ n \to \infty $ uniformly in $ \{ \mathbf{u} \in \mathbb{R}^p \,|\, \| \mathbf{u} \| = 1 \} $, it follows that
\begin{align*}
& \sup_{\| \mathbf{y} \| \le C} \sup_{\| \mathbf{u} \| = 1} \left| \frac{\mathbf{u}^t \mathbf{y}}{m_n( | \mathbf{u}^t \mathbf{Y} - m_n( \mathbf{u}^t \mathbf{Y} \,|\, \mathbf{x} ) | \,|\, \mathbf{x} )} 
- \frac{\mathbf{u}^t \mathbf{y}}{m( | \mathbf{u}^t \mathbf{Y} - m( \mathbf{u}^t \mathbf{Y} \,|\, \mathbf{x} ) | \,|\, \mathbf{x} )} \right| \stackrel{a.s.}{\longrightarrow} 0 , \\
& \sup_{\| \mathbf{y} \| \le C} \sup_{\| \mathbf{u} \| = 1} \left| \frac{m_n( \mathbf{u}^t \mathbf{Y} \,|\, \mathbf{x} )}{m_n( | \mathbf{u}^t \mathbf{Y} - m_n( \mathbf{u}^t \mathbf{Y} \,|\, \mathbf{x} ) | \,|\, \mathbf{x} )} 
- \frac{m( \mathbf{u}^t \mathbf{Y} \,|\, \mathbf{x} )}{m( | \mathbf{u}^t \mathbf{Y} - m( \mathbf{u}^t \mathbf{Y} \,|\, \mathbf{x} ) | \,|\, \mathbf{x} )} \right| \stackrel{a.s.}{\longrightarrow} 0
\end{align*}
as $ n \to \infty $.
Hence,
\begin{align*}
& \sup_{\| \mathbf{y} \| \le C} \left| \sup_{\| \mathbf{u} \| = 1} \frac{| \mathbf{u}^t \mathbf{y} - m_n( \mathbf{u}^t \mathbf{Y} \,|\, \mathbf{x} ) |}{m_n( | \mathbf{u}^t \mathbf{Y} - m_n( \mathbf{u}^t \mathbf{Y} \,|\, \mathbf{x} ) | \,|\, \mathbf{x} )} 
- \sup_{\| \mathbf{u} \| = 1} \frac{| \mathbf{u}^t \mathbf{y} - m( \mathbf{u}^t \mathbf{Y} \,|\, \mathbf{x} ) |}{m( | \mathbf{u}^t \mathbf{Y} - m( \mathbf{u}^t \mathbf{Y} \,|\, \mathbf{x} ) | \,|\, \mathbf{x} )} \right| 
\stackrel{a.s.}{\longrightarrow} 0
\end{align*}
as $ n \to \infty $.
\end{proof}

\begin{proof}[Proof of \autoref{lemma3}]
If $ \alpha + \delta > \max_\mathbf{y} \rho( \mathbf{y} \,|\, \mathbf{x} ) $, then $ D( \alpha + \delta \,|\, \mathbf{x} ) = \emptyset \subset D_n( \alpha_n \,|\, \mathbf{x} ) $ with probability 1. On the other hand, if $ \alpha - \delta < \min_\mathbf{y} \rho( \mathbf{y} \,|\, \mathbf{x} ) $, then $ D( \alpha - \delta \,|\, \mathbf{x} ) $ is the entire response space, and as a result $ D_n( \alpha_n \,|\, \mathbf{x} ) \subset D( \alpha - \delta \,|\, \mathbf{x} ) $ with probability 1. So we assume
\begin{align}
\delta \le \min\left\{ \max_\mathbf{y} \rho( \mathbf{y} \,|\, \mathbf{x} ) - \alpha, \, \alpha - \min_\mathbf{y} \rho( \mathbf{y} \,|\, \mathbf{x} ) \right\} .
\label{lemma3eq1}
\end{align}
Given $ \epsilon > 0 $ and $ \delta > 0 $, we take
\begin{align*}
\epsilon_1 = \epsilon_2 = \frac{\epsilon}{4} ,\; \delta_1 = \delta_2 = \frac{\delta}{4} .
\end{align*}
There exists $ N_1( \epsilon_1, \delta_1 ) $ such that for all $ n \ge N_1( \epsilon_1, \delta_1 ) $,
\begin{align}
\mathbb{P}[ \alpha - \delta_1 < \alpha_n < \alpha + \delta_1 ] > 1 - \epsilon_1 .
\label{lemma3eq2}
\end{align}
Similarly, under \eqref{eq2}, there exists $ N_2( \epsilon_2, \delta_2 ) $ such that for all $ n \ge N_2( \epsilon_2, \delta_2 ) $,
\begin{align}
\mathbb{P}[ \rho( \mathbf{y} \,|\, \mathbf{x} ) - \delta_2 < \rho_n( \mathbf{y} \,|\, \mathbf{x} ) < \rho( \mathbf{y} \,|\, \mathbf{x} ) + \delta_2 \text{ for all } \mathbf{y} ] > 1 - \epsilon_2 .
\label{lemma3eq3}
\end{align}

Using \eqref{lemma3eq1}, \eqref{lemma3eq2}, and \eqref{lemma3eq3}, we have for all $ n \ge \max\{ N_1( \epsilon_1, \delta_1 ), \allowbreak N_2( \epsilon_2, \delta_2 ) \} $,
\begin{align}
& \mathbb{P}\left[ \alpha - \delta_1 < \alpha_n < \alpha + \delta_1 \text{ and }
\rho( \mathbf{y} \,|\, \mathbf{x} ) - \delta_2 < \rho_n( \mathbf{y} \,|\, \mathbf{x} ) < \rho( \mathbf{y} \,|\, \mathbf{x} ) + \delta_2 \text{ for all } \mathbf{y} \right] \nonumber\\
& \quad
> 1 - ( \epsilon_1 + \epsilon_2 ) 
\nonumber\\
& \Rightarrow
\mathbb{P}\left[ \alpha_n - \delta_1 < \alpha \text{ and }
\alpha + \delta - \delta_2 \le \rho( \mathbf{y} \,|\, \mathbf{x} ) - \delta_2 < \rho_n( \mathbf{y} \,|\, \mathbf{x} ) \text{ for all } \mathbf{y} \text{ such that } \right. \nonumber\\
& \qquad\quad\left. \rho( \mathbf{y} \,|\, \mathbf{x} ) \ge \alpha + \delta \right] \nonumber\\
& \quad > 1 - ( \epsilon_1 + \epsilon_2 ) 
\nonumber\\
& \Rightarrow
\mathbb{P}\left[ D( \alpha + \delta \,|\, \mathbf{x} ) \subseteq D_n( \alpha_n + \delta - (\delta_1 + \delta_2) \,|\, \mathbf{x} ) \right]
> 1 - ( \epsilon_1 + \epsilon_2 ) .
\label{lemma3eq4}
\end{align}
Also, again using \eqref{lemma3eq1}, \eqref{lemma3eq2}, and \eqref{lemma3eq3}, we have for all $ n \ge \max\{ N_1( \epsilon_1, \delta_1 ),\allowbreak N_2( \epsilon_2, \delta_2 ) \} $,
\begin{align}
& \mathbb{P}\left[ \alpha - \delta_1 < \alpha_n < \alpha + \delta_1 \text{ and }
\rho( \mathbf{y} \,|\, \mathbf{x} ) - \delta_2 < \rho_n( \mathbf{y} \,|\, \mathbf{x} ) < \rho( \mathbf{y} \,|\, \mathbf{x} ) + \delta_2 \text{ for all } \mathbf{y} \right] \nonumber\\
& \quad
> 1 - ( \epsilon_1 + \epsilon_2 ) 
\nonumber\\
& \Rightarrow
\mathbb{P}\left[ \alpha < \alpha_n + \delta_1 \text{ and }
\rho_n( \mathbf{y} \,|\, \mathbf{x} ) < \rho( \mathbf{y} \,|\, \mathbf{x} ) + \delta_2 < \alpha - \delta + \delta_2 \text{ for all } \mathbf{y} \text{ such that } \right. \nonumber\\
& \qquad\quad\left. \rho( \mathbf{y} \,|\, \mathbf{x} ) < \alpha - \delta \right] \nonumber\\
& \quad
> 1 - ( \epsilon_1 + \epsilon_2 ) 
\nonumber\\
& \Rightarrow
\mathbb{P}\left[ \rho_n( \mathbf{y} \,|\, \mathbf{x} ) < \alpha_n + (\delta_1 + \delta_2) - \delta \text{ for all } \mathbf{y} \text{ such that } \rho( \mathbf{y} \,|\, \mathbf{x} ) < \alpha - \delta \right] \nonumber\\
& \qquad
> 1 - ( \epsilon_1 + \epsilon_2 ) 
\nonumber\\
& \Leftrightarrow
\mathbb{P}\left[ D_n( \alpha_n - ( \delta - (\delta_1 + \delta_2) ) \,|\, \mathbf{x} ) \subseteq D( \alpha - \delta \,|\, \mathbf{x} ) \right] 
> 1 - ( \epsilon_1 + \epsilon_2 ) .
\label{lemma3eq5}
\end{align}
So, from \eqref{lemma3eq4} and \eqref{lemma3eq5} we get that, for all $ n \ge \max\{ N_1( \epsilon_1, \delta_1 ), N_2( \epsilon_2, \delta_2 ) \} $,
\begin{align*}
& \mathbb{P}\left[ D( \alpha + \delta \,|\, \mathbf{x} ) \subseteq D_n( \alpha_n \,|\, \mathbf{x} ) \subseteq D( \alpha - \delta \,|\, \mathbf{x} ) \right]
\nonumber\\
& \ge
\mathbb{P}\left[ D( \alpha + \delta \,|\, \mathbf{x} ) \subseteq D_n( \alpha_n + ( \delta - (\delta_1 + \delta_2) ) \,|\, \mathbf{x} ) \; \text{and} \right. \nonumber\\
& \qquad\quad\left. \; D_n( \alpha_n - ( \delta - (\delta_1 + \delta_2) ) \,|\, \mathbf{x} ) \subseteq D( \alpha - \delta \,|\, \mathbf{x} ) \right]
\nonumber\\
& > 1 - 2 ( \epsilon_1 + \epsilon_2 ) = 1 - \epsilon .
\end{align*}
\end{proof}

\begin{proof}[Proof of \autoref{lemma4}]
Note that
\begin{align}
& | \mu_n( D_n( \beta \,|\, \mathbf{x} ) \,|\, \mathbf{x} ) - \mu( D( \beta \,|\, \mathbf{x} ) \,|\, \mathbf{x} ) | \nonumber\\
& \le
| \mu_n( D_n( \beta \,|\, \mathbf{x} ) \,|\, \mathbf{x} ) - \mu_n( D( \beta \,|\, \mathbf{x} ) \,|\, \mathbf{x} ) | \nonumber\\
& \quad
+ | \mu_n( D( \beta \,|\, \mathbf{x} ) \,|\, \mathbf{x} ) - \mu( D( \beta \,|\, \mathbf{x} ) \,|\, \mathbf{x} ) | . 
\label{lemma4eq1}
\end{align}
From \ref{a4} and \eqref{eq1}, we get that
\begin{align}
| \mu_n( D( \beta \,|\, \mathbf{x} ) \,|\, \mathbf{x} ) - \mu( D( \beta \,|\, \mathbf{x} ) \,|\, \mathbf{x} ) | \stackrel{P}{\longrightarrow} 0
\quad
\text{as } n \to \infty .
\label{lemma4eq2}
\end{align}
We shall show that $ | \mu_n( D_n( \beta \,|\, \mathbf{x} ) \,|\, \mathbf{x} ) - \mu_n( D( \beta \,|\, \mathbf{x} ) \,|\, \mathbf{x} ) | \stackrel{P}{\longrightarrow} 0 $ as $ n \to \infty $.

Given any $ \epsilon , \delta > 0 $, take $ \epsilon_1 = \epsilon / 2 $ and $ \delta_1 = \delta / 2 $.
Since $ \mu( \{ \mathbf{y} \,|\, \rho( \mathbf{y} \,|\, \mathbf{x} ) = \beta \} \,|\, \mathbf{x} ) = 0 $ under \ref{a4}, we can find $ \delta_2 > 0 $ such that
\begin{align}
\mu( \{ \mathbf{y} \,|\, \beta - \delta_2 \le \rho( \mathbf{y} \,|\, \mathbf{x} ) < \beta + \delta_2 \} \,|\, \mathbf{x} ) < \delta_1 .
\label{lemma4eq3}
\end{align}
Denote $ B( \beta, \delta_2 ) = \{ \mathbf{y} \,|\, \beta - \delta_2 \le \rho( \mathbf{y} \,|\, \mathbf{x} ) < \beta + \delta_2 \} $. Under \eqref{eq1} and from \eqref{lemma4eq3}, we get that there exists $ N_1( \delta_1, \epsilon_1 ) $ such that for all $ n \ge N_1( \delta_1, \epsilon_1 ) $,
\begin{align}
& \mathbb{P}[ \mu_n( B( \beta, \delta_2 ) \,|\, \mathbf{x} ) < 2 \delta_1 ] \nonumber\\
& \ge \mathbb{P}[ | \mu_n( B( \beta, \delta_2 ) \,|\, \mathbf{x} ) - \mu( B( \beta, \delta_2 ) \,|\, \mathbf{x} ) | < \delta_1 ] 
> 1 - \epsilon_1 .
\label{lemma4eq4}
\end{align}
Denote the events
\begin{align*}
& A_n^{(1)}( \beta, \delta_1, \delta_2 )
= \{ \mu_n( B( \beta, \delta_2 ) \,|\, \mathbf{x} ) < 2 \delta_1 \} ,
\\
& A_n^{(2)}( \beta, \delta_2 ) = \{ D( \beta + \delta_2 \,|\, \mathbf{x} ) \subseteq D_n( \beta \,|\, \mathbf{x} ) \subseteq D( \beta - \delta_2 \,|\, \mathbf{x} ) \} .
\end{align*}
From \autoref{lemma3}, we get that there exists $ N_2( \epsilon_1, \delta_2 ) $ such that for all $ n \ge N_2( \epsilon_1, \delta_2 ) $,
\begin{align}
\mathbb{P}[ A_n^{(2)}( \beta, \delta_2 ) ] > 1 - \epsilon_1 .
\label{lemma4eq5}
\end{align}
Since $ D( \beta + \delta_1 \,|\, \mathbf{x} ) \subseteq D( \beta \,|\, \mathbf{x} ) \subseteq D( \beta - \delta_1 \,|\, \mathbf{x} ) $, when $ A_n^{(2)}( \beta, \delta_2 ) $ occurs, we have
\begin{align}
| \mu_n( D_n( \beta \,|\, \mathbf{x} ) \,|\, \mathbf{x} ) - \mu_n( D( \beta \,|\, \mathbf{x} ) \,|\, \mathbf{x} ) |
& \le \mu_n( B( \beta, \delta_2 ) \,|\, \mathbf{x} ) .
\label{lemma4eq6}
\end{align}
So, from \eqref{lemma4eq4}, \eqref{lemma4eq5} and \eqref{lemma4eq6}, we get that for all $ n \ge \max\{ N_1( \delta_1, \epsilon_1 ), \allowbreak N_2( \epsilon_1, \delta_2 ) \} $,
\begin{align}
& \mathbb{P}[ | \mu_n( D_n( \beta \,|\, \mathbf{x} ) \,|\, \mathbf{x} ) - \mu_n( D( \beta \,|\, \mathbf{x} ) \,|\, \mathbf{x} ) | < \delta ]
\nonumber\\
& \ge
\mathbb{P}[ | \mu_n( D_n( \beta \,|\, \mathbf{x} ) \,|\, \mathbf{x} ) - \mu_n( D( \beta \,|\, \mathbf{x} ) \,|\, \mathbf{x} ) | \le \mu_n( B( \beta, \delta_2 ) \,|\, \mathbf{x} )
\text{ and } \nonumber\\
& \qquad\quad \mu_n( B( \beta, \delta_2 ) \,|\, \mathbf{x} ) < 2 \delta_1 ]
\nonumber\\
& \ge
\mathbb{P}\left[ A_n^{(2)}( \beta, \delta_2 ) \cap A_n^{(1)}( \beta, \delta_1, \delta_2 ) \right]
> 1 - \epsilon .
\label{lemma4eq7}
\end{align}
The proof follows from \eqref{lemma4eq1}, \eqref{lemma4eq2} and \eqref{lemma4eq7}.
\end{proof}

\begin{proof}[Proof of \autoref{lemma5}]
Given $ \epsilon > 0 $ and $ \delta > 0 $, denote
\begin{align*}
& \delta_1 = \frac{r - \mu( D( \alpha( r ) + \delta \,|\, \mathbf{x} ) \,|\, \mathbf{x} )}{2} ,
\quad \delta_2 = \frac{\mu( D( \alpha( r ) - \delta \,|\, \mathbf{x} ) \,|\, \mathbf{x} ) - r}{2} .
\end{align*}
By definition of $ \alpha( r ) $, $ \delta_1 > 0 $. From \ref{a5}, we have $ \mu( \{ \mathbf{y} \,|\, \alpha( r ) - \delta \le \rho( \mathbf{y} \,|\, \mathbf{x} ) < \alpha( r ) \} \,|\, \mathbf{x} ) > 0 $, which ensures $ \delta_2 > 0 $. Define the events
\begin{align*}
& A_n( \delta ) = \{ \alpha_n( r ) > \alpha( r ) + \delta \} , \\
& B_n( \delta ) = \{ \alpha_n( r ) < \alpha( r ) - \delta \} , \\
& C_n^{(1)}( \delta_1 ) = \{ | \mu_n( D_n( \alpha( r ) + \delta \,|\, \mathbf{x} ) \,|\, \mathbf{x} ) - \mu( D( \alpha( r ) + \delta \,|\, \mathbf{x} ) \,|\, \mathbf{x} ) | < \delta_1 \} , \\
& C_n^{(2)}( \delta_2 ) = \{ | \mu_n( D_n( \alpha( r ) - \delta \,|\, \mathbf{x} ) \,|\, \mathbf{x} ) - \mu( D( \alpha( r ) - \delta \,|\, \mathbf{x} ) \,|\, \mathbf{x} ) | < \delta_2 \} .
\end{align*}
From \autoref{lemma4}, we get that there exists $ N( \epsilon, \delta_1, \delta_2 ) $ such that for all $ n \ge N( \epsilon, \delta_1, \delta_2 ) $,
\begin{align}
& \mathbb{P}\left[ C_n^{(1)}( \delta_1 ) \right] > 1 - (\epsilon / 2) 
\, \text{ and } \,
\mathbb{P}\left[ C_n^{(2)}( \delta_2 ) \right] > 1 - (\epsilon / 2) .
\label{lemma5eq1}
\end{align}
Note that,
\begin{align*}
\text{when } A_n( \delta ) \text{ occurs, }
& \mu_n( D_n( \alpha( r ) + \delta \,|\, \mathbf{x} ) \,|\, \mathbf{x} )
\ge r , 
\\
\text{when } B_n( \delta ) \text{ occurs, }
& \mu_n( D_n( \alpha( r ) - \delta \,|\, \mathbf{x} ) \,|\, \mathbf{x} )
< r , 
\\
\text{when } C_n^{(1)}( \delta_1 ) \text{ occurs, }
& \mu_n( D_n( \alpha( r ) + \delta \,|\, \mathbf{x} ) \,|\, \mathbf{x} )
< \mu( D( \alpha( r ) + \delta \,|\, \mathbf{x} ) \,|\, \mathbf{x} ) + \delta_1
< r , 
\\
\text{when } C_n^{(2)}( \delta_2 ) \text{ occurs, }
& \mu_n( D_n( \alpha( r ) - \delta \,|\, \mathbf{x} ) \,|\, \mathbf{x} )
> \mu( D( \alpha( r ) - \delta \,|\, \mathbf{x} ) \,|\, \mathbf{x} ) - \delta_2
> r .
\end{align*}
So, $ A_n( \delta ) \cap C_n^{(1)}( \delta_1 ) = \emptyset $ and $ B_n( \delta ) \cap C_n^{(2)}( \delta_2 ) = \emptyset $. Consequently, from \eqref{lemma5eq1}, it follows that for all $ n \ge N( \epsilon, \delta_1, \delta_2 ) $,
\begin{align*}
& \mathbb{P}[ | \alpha_n( r ) - \alpha( r ) | > \delta ]
= \mathbb{P}\left[ A_n( \delta ) \cup B_n( \delta ) \right]
\le 1 - \mathbb{P}\left[ C_n^{(1)}( \delta_1 ) \cap C_n^{(2)}( \delta_2 ) \right]
\le \epsilon ,
\end{align*}
which completes the proof.
\end{proof}

\begin{proof}[Proof of \autoref{lemma6}]
Recall that for any pair of sets $ A $ and $ B $,
\begin{align*}
d_H( A, B ) = \max\left\{ \sup_{\mathbf{u} \in A} \inf_{\mathbf{v} \in B} \| \mathbf{u} - \mathbf{v} \|, \, \sup_{\mathbf{v} \in B} \inf_{\mathbf{u} \in A} \| \mathbf{u} - \mathbf{v} \| \right\} .
\end{align*}
Since $ D( \alpha + \delta \,|\, \mathbf{x} ) \subseteq D( \alpha \,|\, \mathbf{x} ) \subseteq D( \alpha - \delta \,|\, \mathbf{x} ) $ for any $ \delta > 0 $, we have
\begin{align}
& d_H( D( \alpha + \delta \,|\, \mathbf{x} ), D( \alpha - \delta \,|\, \mathbf{x} ) ) \nonumber\\
& \le
d_H( D( \alpha \,|\, \mathbf{x} ) , D( \alpha - \delta \,|\, \mathbf{x} ) )
+ d_H( D( \alpha + \delta \,|\, \mathbf{x} ) , D( \alpha \,|\, \mathbf{x} ) ) .
\label{lemma6eq1}
\end{align}
Denote $ d_0( \mathbf{u}, A ) = \inf\{ \| \mathbf{u} - \mathbf{v} \| \,|\, \mathbf{v} \in A \} $ for a point $ \mathbf{u} $ and a set $ A $. So, $ d_H( D( \alpha \,|\, \mathbf{x} ) , D( \alpha - \delta \,|\, \mathbf{x} ) ) = \sup\{ d_0( \mathbf{y}, D( \alpha \,|\, \mathbf{x} ) ) \,|\, \mathbf{y} \in D( \alpha - \delta \,|\, \mathbf{x} ) \} $. Suppose, if possible, $ \lim_{\delta \to 0^+} d_H( D( \alpha \,|\, \mathbf{x} ) , \allowbreak D( \alpha - \delta \,|\, \mathbf{x} ) ) > 0 $. Then, there exists a constant $ \delta_1 > 0 $ and a sequence $ \{ \mathbf{y}_n \} $ such that $ \rho( \mathbf{y}_n \,|\, \mathbf{x} ) \ge \alpha( 1 - 0.5 (1 / n) ) $ and $ d_0( \mathbf{y}_n, D( \alpha \,|\, \mathbf{x} ) ) \ge \delta_1 $ for all $ n $. Under \ref{a6}, $ D( ( \alpha / 2 ) \,|\, \mathbf{x} ) $ is bounded, and so $ \{ \mathbf{y}_n \} $ is a bounded sequence in $ \mathbb{R}^p $, the response space. Consequently, $ \{ \mathbf{y}_n \} $ has convergent subsequence $ \{ \mathbf{y}_{n_k} \} $, with, say, $ \mathbf{y}_{n_k} \to \mathbf{y}_0 $ as $ k \to \infty $. From \ref{a6}, we have $ \rho( \mathbf{y}_0 \,|\, \mathbf{x} ) = \lim_k \rho( \mathbf{y}_{n_k} \,|\, \mathbf{x} ) \ge \alpha $ and $ d_0( \mathbf{y}_0, D( \alpha \,|\, \mathbf{x} ) ) = \lim_k d_0( \mathbf{y}_{n_k}, D( \alpha \,|\, \mathbf{x} ) ) \ge \delta_1 > 0 $, which contradict themselves. Therefore,
\begin{align}
\lim_{\delta \to 0^+} d_H( D( \alpha \,|\, \mathbf{x} ) , D( \alpha - \delta \,|\, \mathbf{x} ) ) = 0 .
\label{lemma6eq2}
\end{align}
Next, we look into the term $ d_H( D( \alpha + \delta \,|\, \mathbf{x} ) , D( \alpha \,|\, \mathbf{x} ) ) $. We have $ d_H( D( \alpha + \delta \,|\, \mathbf{x} ) , \allowbreak D( \alpha \,|\, \mathbf{x} ) ) = \sup\{ d_0( \mathbf{y}, D( \alpha + \delta \,|\, \mathbf{x} ) ) \,|\, \mathbf{y} \in D( \alpha \,|\, \mathbf{x} ) \} $. Suppose, if possible, $ \lim_{\delta \to 0^+} d_H( D( \alpha + \delta \,|\, \mathbf{x} ) , D( \alpha \,|\, \mathbf{x} ) ) \allowbreak > d_H( D_0( \alpha \,|\, \mathbf{x} ) , D( \alpha \,|\, \mathbf{x} ) ) $. Then, there are $ \delta_2, \delta_3 > 0 $ and a sequence $ \{ \mathbf{y}_n \} \subset D( \alpha \,|\, \mathbf{x} ) $ such that $ d_0( \mathbf{y}_n, D( \alpha ( 1 + (1/n) ) \,|\, \mathbf{x} ) )
> \delta_2 > \delta_3
> d_H( D_0( \alpha \,|\, \mathbf{x} ) , D( \alpha \,|\, \mathbf{x} ) ) $ for all sufficiently large $ n $. From \ref{a6}, we have $ D( \alpha \,|\, \mathbf{x} ) $ closed and bounded. So, $ \{ \mathbf{y}_n \} $ has a convergent subsequence $ \{ \mathbf{y}_{n_k} \} $ with $ \lim_k \mathbf{y}_{n_k} = \mathbf{y}_0 \in D( \alpha \,|\, \mathbf{x} ) $. Since $ d_0( \mathbf{y}_0, D( \alpha ( 1 + (1/n_k) ) \,|\, \mathbf{x} ) ) \ge d_0( \mathbf{y}_{n_k}, D( \alpha ( 1 + (1/n_k) ) \,|\, \mathbf{x} ) ) - \| \mathbf{y}_{n_k} - \mathbf{y}_0 \| $, it follows that $ d_0( \mathbf{y}_0, D( \alpha ( 1 + (1/n_k) ) \,|\, \mathbf{x} ) ) > \delta_3 > d_0( \mathbf{y}_0 , D_0( \alpha \,|\, \mathbf{x} ) ) $ for all sufficiently large $ k $. This implies that there is $ \mathbf{y}' \in D_0( \alpha \,|\, \mathbf{x} ) $ with $ \delta_3 > \| \mathbf{y}_0 - \mathbf{y}' \| $. Since $ \rho( \mathbf{y}' \,|\, \mathbf{x} ) > \alpha $, we have $ \mathbf{y}' \in D( \alpha ( 1 + (1/n_k) ) \,|\, \mathbf{x} ) $ for all sufficiently large $ k $, but this leads to a contradiction as then $ \| \mathbf{y}_0 - \mathbf{y}' \| \ge d_0( \mathbf{y}_0, D( \alpha ( 1 + (1/n_k) ) \,|\, \mathbf{x} ) ) > \delta_3 $. Hence, we have
\begin{align}
\lim_{\delta \to 0^+} d_H( D( \alpha + \delta \,|\, \mathbf{x} ) , D( \alpha \,|\, \mathbf{x} ) ) = d_H( D_0( \alpha \,|\, \mathbf{x} ) , D( \alpha \,|\, \mathbf{x} ) ) .
\label{lemma6eq3}
\end{align}
From \ref{a7} and \eqref{lemma6eq3}, we get
\begin{align}
\lim_{\delta \to 0^+} d_H( D( \alpha + \delta \,|\, \mathbf{x} ) , D( \alpha \,|\, \mathbf{x} ) ) = 0 .
\label{lemma6eq4}
\end{align}
Therefore, from \eqref{lemma6eq1}, \eqref{lemma6eq2} and \eqref{lemma6eq4}, we get
$ d_H( D( \alpha + \delta \,|\, \mathbf{x} ), D( \alpha - \delta \,|\, \mathbf{x} ) ) \to 0 $ as $ \delta \to 0^+ $.
Given any $ \epsilon > 0 $ and $ \delta > 0 $, take $ \delta_1 > 0 $ such that
\begin{align*}
& d_H( D( \alpha + \delta_1 \,|\, \mathbf{x} ), D( \alpha - \delta_1 \,|\, \mathbf{x} ) )
< \delta .
\end{align*}
Denote the event
$ A_n( \delta_1 ) = \left\{ D( \alpha + \delta_1 \,|\, \mathbf{x} ) \subseteq D_n( \alpha_n \,|\, \mathbf{x} ) \subseteq D( \alpha - \delta_1 \,|\, \mathbf{x} ) \right\} $.
From \autoref{lemma3}, we get that there exists $ N( \epsilon, \delta_1 ) $ such that for all $ n \ge N( \epsilon, \delta_1 ) $, $ \mathbb{P}\left[ A_n( \delta_1 ) \right] > 1 - \epsilon $.
Also, when the event $ A_n( \delta_1 ) $ occurs, we have
\begin{align*}
d_H( D_n( \alpha_n \,|\, \mathbf{x} ) , D( \alpha \,|\, \mathbf{x} ) ) 
\le 
d_H( D( \alpha + \delta_1 \,|\, \mathbf{x} ), D( \alpha - \delta_1 \,|\, \mathbf{x} ) )
< \delta ,
\end{align*}
which implies
\begin{align*}
\mathbb{P}[ d_H( D_n( \alpha_n \,|\, \mathbf{x} ) , D( \alpha \,|\, \mathbf{x} ) ) 
\le \delta ]
\ge \mathbb{P}\left[ A_n( \delta_1 ) \right] > 1 - \epsilon
\end{align*}
for all $ n \ge N( \epsilon, \delta_1 ) $.
Hence, $ d_H( D_n( \alpha_n \,|\, \mathbf{x} ), D( \alpha \,|\, \mathbf{x} ) ) \stackrel{P}{\longrightarrow} 0 $ as $ n \to \infty $
\end{proof}

\begin{proof}[Proof of \autoref{lemma7}]
We have
\begin{align}
& \left\| \int \mathbf{y} \mathbb{I}( \mathbf{y} \in D_n( \alpha_n \,|\, \mathbf{x} ) ) \mu_n( \mathrm{d} \mathbf{y} \,|\, \mathbf{x} )
- \int \mathbf{y} \mathbb{I}( \mathbf{y} \in D( \alpha \,|\, \mathbf{x} ) ) \mu( \mathrm{d} \mathbf{y} \,|\, \mathbf{x} ) \right\| \nonumber\\
& \le \left\| \int \mathbf{y} \mathbb{I}( \mathbf{y} \in D_n( \alpha_n \,|\, \mathbf{x} ) ) \mu_n( \mathrm{d} \mathbf{y} \,|\, \mathbf{x} )
- \int \mathbf{y} \mathbb{I}( \mathbf{y} \in D( \alpha \,|\, \mathbf{x} ) ) \mu_n( \mathrm{d} \mathbf{y} \,|\, \mathbf{x} ) \right\| \nonumber\\
& \quad + \left\| \int \mathbf{y} \mathbb{I}( \mathbf{y} \in D( \alpha \,|\, \mathbf{x} ) ) \mu_n( \mathrm{d} \mathbf{y} \,|\, \mathbf{x} )
- \int \mathbf{y} \mathbb{I}( \mathbf{y} \in D( \alpha \,|\, \mathbf{x} ) ) \mu( \mathrm{d} \mathbf{y} \,|\, \mathbf{x} ) \right\| .
\label{lemma7eq1}
\end{align}

Define the event $ A_n( \delta ) = \{ D( \alpha + \delta \,|\, \mathbf{x} ) \subseteq D_n( \alpha_n \,|\, \mathbf{x} ) \subseteq D( \alpha - \delta \,|\, \mathbf{x} ) \} $, where $ \delta > 0 $.
Since $ \alpha > 0 $, from \ref{a6}, we get that there is $ \delta_1 > 0 $ such that for all $ 0 < \delta \le \delta_1 $, $ D( \alpha - \delta \,|\, \mathbf{x} ) $ is a bounded set. For any $ 0 < \delta \le \delta_1 $, when the event $ A_n( \delta ) $ occurs, we have
\begin{align}
& \left\| \int \mathbf{y} \mathbb{I}( \mathbf{y} \in D_n( \alpha_n \,|\, \mathbf{x} ) ) \mu_n( \mathrm{d} \mathbf{y} \,|\, \mathbf{x} )
- \int \mathbf{y} \mathbb{I}( \mathbf{y} \in D( \alpha \,|\, \mathbf{x} ) ) \mu_n( \mathrm{d} \mathbf{y} \,|\, \mathbf{x} ) \right\|
\nonumber\\
& \le
\int \left\| \mathbf{y} \right\| | \mathbb{I}( \mathbf{y} \in D_n( \alpha_n \,|\, \mathbf{x} ) ) - \mathbb{I}( \mathbf{y} \in D( \alpha \,|\, \mathbf{x} ) ) | \mu_n( \mathrm{d} \mathbf{y} \,|\, \mathbf{x} )
\nonumber\\
& \le
\int \left\| \mathbf{y} \right\| \mathbb{I}( \mathbf{y} \in D( \alpha - \delta \,|\, \mathbf{x} ) \cap ( D( \alpha + \delta \,|\, \mathbf{x} ) )^c ) \mu_n( \mathrm{d} \mathbf{y} \,|\, \mathbf{x} )
\nonumber\\
& \le
\sup\left\{ \left\| \mathbf{y} \right\| \,|\, \mathbf{y} \in D( \alpha - \delta \,|\, \mathbf{x} ) \right\}
\mu_n\left( \{ \mathbf{y} \,|\, \alpha - \delta \le \rho( \mathbf{y} \,|\, \mathbf{x} ) < \alpha + \delta \} \,|\, \mathbf{x} \right)
\nonumber\\
& \le
\sup\left\{ \left\| \mathbf{y} \right\| \,|\, \mathbf{y} \in D( \alpha - \delta_1 \,|\, \mathbf{x} ) \right\}
\mu_n\left( \{ \mathbf{y} \,|\, \alpha - \delta \le \rho( \mathbf{y} \,|\, \mathbf{x} ) < \alpha + \delta \} \,|\, \mathbf{x} \right) .
\label{lemma7eq2}
\end{align}
Since $ \sup\left\{ \left\| \mathbf{y} \right\| \,|\, \mathbf{y} \in D( \alpha - \delta_1 \,|\, \mathbf{x} ) \right\} < \infty $, from \ref{a4} and \eqref{eq1}, \autoref{lemma3} and \eqref{lemma7eq2}, we get
\begin{align}
& \left\| \int \mathbf{y} \mathbb{I}( \mathbf{y} \in D_n( \alpha_n \,|\, \mathbf{x} ) ) \mu_n( \mathrm{d} \mathbf{y} \,|\, \mathbf{x} )
- \int \mathbf{y} \mathbb{I}( \mathbf{y} \in D( \alpha \,|\, \mathbf{x} ) ) \mu_n( \mathrm{d} \mathbf{y} \,|\, \mathbf{x} ) \right\| \nonumber\\
& \stackrel{P}{\longrightarrow} 0
\quad
\text{as } n \to \infty .
\label{lemma7eq3}
\end{align}

Next, recall that $ d_0( \mathbf{u}, A ) = \inf\{ \| \mathbf{u} - \mathbf{v} \| \,|\, \mathbf{v} \in A \} $ for a point $ \mathbf{u} $ and a set $ A $. Denote $ r_+ = r \mathbb{I}( r > 0 ) $ for $ r \in \mathbb{R} $. Define the functions $ f_1( \cdot ) : \mathbb{R}^p \to \mathbb{R}^p $ and $ f_2( \cdot ) : \mathbb{R}^p \to \mathbb{R}^p $ by
\begin{align*}
& f_1( \mathbf{y} ) = \mathbf{y} \mathbb{I}( \mathbf{y} \in D( \alpha \,|\, \mathbf{x} ) )  
\quad \text{and} \quad
f_2( \mathbf{y} ) = \mathbf{y} \left( 1 - \delta^{-1} d_0( \mathbf{y}, D( \alpha \,|\, \mathbf{x} ) ) \right)_+ ,
\end{align*}
where $ \delta > 0 $.
From \ref{a6}, we get that the set of discontinuity points of $ f_1( \cdot ) $ is $ \{ \mathbf{y} \,|\, \rho( \mathbf{y} \,|\, \mathbf{x} ) = \alpha \} $. Therefore, from \ref{a4} and \eqref{eq1} and Theorem 2.7 in \citet[p.~21]{billingsley2013convergence}, we get
\begin{align}
\mu_n\left( f_1^{-1}( \cdot ) \middle\arrowvert \mathbf{x} \right) \stackrel{w}{\longrightarrow} \mu\left( f_1^{-1}( \cdot ) \middle\arrowvert \mathbf{x} \right)
\quad
\text{\emph{almost surely} as }
n \to \infty .
\label{lemma7eq4}
\end{align}
Since $ \alpha > 0 $, from \ref{a6}, we get that $ D( \alpha \,|\, \mathbf{x} ) $ is a bounded set, and hence, $ f_2( \cdot ) $ is a bounded continuous function on $ \mathbb{R}^p $. Therefore, from \eqref{lemma7eq4}, we have
\begin{align}
& \left\| \int f_2 \mathrm{d} \mu_n\left( f_1^{-1} \middle\arrowvert \mathbf{x} \right)
- \int f_2 \mathrm{d} \mu\left( f_1^{-1} \middle\arrowvert \mathbf{x} \right) \right\|
\stackrel{a.s.}{\longrightarrow} 0
\quad
\text{as } n \to \infty .
\label{lemma7eq5}
\end{align}
Note that
\begin{align*}
& \int f_2 \mathrm{d} \mu_n\left( f_1^{-1} \middle\arrowvert \mathbf{x} \right)
= \int \mathbf{y} \mathbb{I}( \mathbf{y} \in D( \alpha \,|\, \mathbf{x} ) ) \mu_n( \mathrm{d} \mathbf{y} \,|\, \mathbf{x} ) \nonumber\\
& \text{and}\quad
\int f_2 \mathrm{d} \mu\left( f_1^{-1} \middle\arrowvert \mathbf{x} \right)
= \int \mathbf{y} \mathbb{I}( \mathbf{y} \in D( \alpha \,|\, \mathbf{x} ) ) \mu( \mathrm{d} \mathbf{y} \,|\, \mathbf{x} ) .
\end{align*}
Hence, from \eqref{lemma7eq5}, we have
\begin{align}
& \left\| \int \mathbf{y} \mathbb{I}( \mathbf{y} \in D( \alpha \,|\, \mathbf{x} ) ) \mu_n( \mathrm{d} \mathbf{y} \,|\, \mathbf{x} )
- \int \mathbf{y} \mathbb{I}( \mathbf{y} \in D( \alpha \,|\, \mathbf{x} ) ) \mu( \mathrm{d} \mathbf{y} \,|\, \mathbf{x} ) \right\| \nonumber\\
& \stackrel{a.s.}{\longrightarrow} 0
\quad\text{as } n \to \infty .
\label{lemma7eq6}
\end{align}
From \eqref{lemma7eq1}, \eqref{lemma7eq3} and \eqref{lemma7eq6}, we get
\begin{align*}
\int \mathbf{y} \mathbb{I}( \mathbf{y} \in D_n( \alpha_n \,|\, \mathbf{x} ) ) \mu_n( \mathrm{d} \mathbf{y} \,|\, \mathbf{x} )
\stackrel{P}{\longrightarrow}
\int \mathbf{y} \mathbb{I}( \mathbf{y} \in D( \alpha \,|\, \mathbf{x} ) ) \mu( \mathrm{d} \mathbf{y} \,|\, \mathbf{x} )
\end{align*}
as $ n \to \infty $.
\end{proof}

\begin{proof}[Proof of \autoref{lemma8}]
Given any $ \delta > 0 $, define the event
\begin{align*}
A_n( \delta ) = \{ D( \alpha + \delta \,|\, \mathbf{x} ) \subseteq D_n( \alpha_n \,|\, \mathbf{x} ) \subseteq D( \alpha - \delta \,|\, \mathbf{x} ) \} .
\end{align*}
When the event $ A_n( \delta ) $ occurs, we have
\begin{align}
\left| \lambda( D_n( \alpha_n \,|\, \mathbf{x} ) )
- \lambda( D( \alpha \,|\, \mathbf{x} ) ) \right|
\le
\lambda( \{ \mathbf{y} \,|\, \alpha - \delta \le \rho( \mathbf{y} \,|\, \mathbf{x} ) < \alpha + \delta \} ) .
\label{lemma8eq1}
\end{align}
From the assumptions in the lemma and using the Radon-Nikodym theorem, we get that $ \lambda( \cdot ) $ has a Radon-Nikodym derivative $ f( \cdot ) $ with respect to $ \mu( \cdot \,|\, \mathbf{x} ) $, which is continuous and bounded over bounded sets. From \ref{a4}, we have $ \mu( \{ \mathbf{y} \,|\, \alpha - \delta \le \rho( \mathbf{y} \,|\, \mathbf{x} ) < \alpha + \delta \} \,|\, \mathbf{x} )
\to 0 $ as $ \delta \to 0^+ $. So,
\begin{align}
\lambda( \{ \mathbf{y} \,|\, \alpha - \delta \le \rho( \mathbf{y} \,|\, \mathbf{x} ) < \alpha + \delta \} )
\to 0
\quad \text{as } \delta \to 0^+ .
\label{lemma8eq2}
\end{align}
Therefore, from \autoref{lemma3}, \eqref{lemma8eq1} and \eqref{lemma8eq2}, we have
$ \lambda( D_n( \alpha_n \,|\, \mathbf{x} ) ) )
\stackrel{P}{\longrightarrow} \lambda( D( \alpha \,|\, \mathbf{x} ) ) $
as $ n \to \infty $.
\end{proof}

\begin{proof}[Proof of \autoref{lemma9}]
Recall that $ d_H( A, B ) = \inf\{ \epsilon \,|\, A \subseteq B^\epsilon , B \subseteq A^\epsilon \} $, where $ A^\epsilon $ and $ B^\epsilon $ denote the $ \epsilon $-neighborhoods of $ A $ and $ B $, respectively.
Take any $ \mathbf{u}, \mathbf{v} \in A $. From the definition of $ d_H( A, B ) $, we get that given any $ \epsilon > 0 $, there are $ \mathbf{u}', \mathbf{v}' \in B $ such that
\begin{align*}
\| \mathbf{u} - \mathbf{u}' \| < d_H( A, B ) + \epsilon 
\quad \text{and} \quad
\| \mathbf{v} - \mathbf{v}' \| < d_H( A, B ) + \epsilon .
\end{align*}
Since $ \| \mathbf{u}' - \mathbf{v}' \| \le \text{Diameter}( B ) $, we have
\begin{align*}
\| \mathbf{u} - \mathbf{v} \|
\le 
\| \mathbf{u} - \mathbf{u}' \|
+ \| \mathbf{v} - \mathbf{v}' \|
+ \| \mathbf{u}' - \mathbf{v}' \|
\le
2 d_H( A, B ) + 2 \epsilon + \text{Diameter}( B ) .
\end{align*}
Since $ \mathbf{u}, \mathbf{v} $ are arbitrary points in $ A $, we have
\begin{align}
\text{Diameter}( A ) - \text{Diameter}( B )
\le 2 d_H( A, B ) + 2 \epsilon .
\label{lemma9eq1}
\end{align}
Similarly, we can show that
\begin{align}
\text{Diameter}( B ) - \text{Diameter}( A )
\le 2 d_H( A, B ) + 2 \epsilon .
\label{lemma9eq2}
\end{align}
Since $ \epsilon > 0 $ is arbitrary, from \eqref{lemma9eq1} and \eqref{lemma9eq2}, we have $ | \text{Diameter}( A ) - \text{Diameter}( B ) | \le 2 d_H( A, B ) $.
\end{proof}

\begin{proof}[Proof of \autoref{lemma10}]
Note that
\begin{align}
& d_H( D_n'( \alpha_n \,|\, \mathbf{x} ), D( \alpha \,|\, \mathbf{x} ) ) \nonumber\\
& = \max\left\{ \sup_{ \mathbf{u} \in D_n'( \alpha_n \,|\, \mathbf{x} ) } d_0( \mathbf{u}, D( \alpha \,|\, \mathbf{x} ) ), \sup_{ \mathbf{v} \in D( \alpha \,|\, \mathbf{x} ) } d_0( \mathbf{v}, D_n'( \alpha_n \,|\, \mathbf{x} ) ) \right\} .
\label{lemma10eq1}
\end{align}
Now, given any $ \epsilon , \delta > 0 $, it follows from \autoref{lemma3} that
\begin{align*}
\mathbb{P}[ D_n'( \alpha_n \,|\, \mathbf{x} ) \subset D_n( \alpha_n \,|\, \mathbf{x} ) \subseteq D( \alpha - \delta \,|\, \mathbf{x} ) ] > 1 - \epsilon
\end{align*}
for all sufficiently large $ n $.
Denote the event
\begin{align*}
A_n( \delta ) = \{ D_n'( \alpha_n \,|\, \mathbf{x} ) \subseteq D( \alpha - \delta \,|\, \mathbf{x} ) \} .
\end{align*}
When the event $ A_n( \delta ) $ occurs, we have
\begin{align*}
& \sup\{ d_0( \mathbf{u}, D( \alpha \,|\, \mathbf{x} ) ) \,|\, \mathbf{u} \in D_n'( \alpha_n \,|\, \mathbf{x} ) \} 
\le \sup\{ d_0( \mathbf{u}, D( \alpha \,|\, \mathbf{x} ) ) \,|\, \mathbf{u} \in D( \alpha - \delta \,|\, \mathbf{x} ) \} .
\end{align*}
Since $ \sup\{ d_0( \mathbf{u}, D( \alpha \,|\, \mathbf{x} ) ) \,|\, \mathbf{u} \in D( \alpha - \delta \,|\, \mathbf{x} ) \} = d_H( D( \alpha \,|\, \mathbf{x} ) , D( \alpha - \delta \,|\, \mathbf{x} ) ) $, from \eqref{lemma6eq2} it follows that
\begin{align}
\sup_{ \mathbf{u} \in D_n'( \alpha_n \,|\, \mathbf{x} ) } d_0( \mathbf{u}, D( \alpha \,|\, \mathbf{x} ) ) \stackrel{P}{\longrightarrow} 0
\quad
\text{as }
n \to \infty .
\label{lemma10eq2}
\end{align}

Next, we consider the term $ \sup\{ d_0( \mathbf{v}, D_n'( \alpha_n \,|\, \mathbf{x} ) ) \,|\, \mathbf{v} \in D( \alpha \,|\, \mathbf{x} ) \} $. Let $ \epsilon , \delta > 0 $ be any given numbers.
From \ref{a6}, we get that $ D( \alpha \,|\, \mathbf{x} ) $ is a compact subset of $ \mathbb{R}^p $. So, we can cover $ D( \alpha \,|\, \mathbf{x} ) $ by a finite number of pairwise disjoint semi-open hypercubes of long-diagonal length $ (\delta / 2) $, such that the interior of each hypercube has a non-empty intersection with $ D( \alpha \,|\, \mathbf{x} ) $. Let the number of such hypercubes be $ n_1 $ and the collection of hypercubes be denoted as $ \{ C_1, \ldots, C_{n_1} \} $. So, under the assumption of the lemma, we have
\begin{align}
& \min\{ \mathbb{P}[ \mathbf{Y} \in C_i \cap D( \alpha \,|\, \mathbf{x} ) ] \,|\, i = 1, \ldots, n_1 \} \nonumber\\
& \ge \min\{ \mathbb{P}[ \mathbf{Y} \in \text{int}( C_i ) \cap D( \alpha \,|\, \mathbf{x} ) ] \,|\, i = 1, \ldots, n_1 \}
> 0 ,
\label{lemma10eq3}
\end{align}
where $ \text{int}( C ) $ denotes the interior of a set $ C $.
Consider the event
\begin{align*}
A_n = \{ \text{Each } ( C_i \cap D( \alpha \,|\, \mathbf{x} ) ) \text{ contains at least one } \mathbf{Y}_j, j = 1, \ldots, n \}.
\end{align*}
Since the hypercubes are pairwise disjoint, from \eqref{lemma10eq3}, we have $ \mathbb{P}[ A_n ] \to 1 $ as $ n \to \infty $. Let $ n_2 $ be an integer such that for all $ n \ge n_2 $,
\begin{align}
\mathbb{P}[ A_n ] > 1 - (\epsilon / 2) .
\label{lemma10eq4}
\end{align}
From \eqref{lemma6eq4}, we can find $ \delta_1 > 0 $ sufficiently small such that
\begin{align}
\sup\{ d_0( \mathbf{v}, D( \alpha + 2 \delta_1 \,|\, \mathbf{x} ) ) \,|\, \mathbf{v} \in D( \alpha \,|\, \mathbf{x} ) \}
< (\delta / 2) .
\label{lemma10eq5}
\end{align}
So, using the triangle inequality and \eqref{lemma10eq5}, we have
\begin{align}
& \sup\{ d_0( \mathbf{v}, D_n'( \alpha_n \,|\, \mathbf{x} ) ) \,|\, \mathbf{v} \in D( \alpha \,|\, \mathbf{x} ) \} \nonumber\\
& \le \sup\{ d_0( \mathbf{v}, D( \alpha + 2 \delta_1 \,|\, \mathbf{x} ) ) \,|\, \mathbf{v} \in D( \alpha \,|\, \mathbf{x} ) \} \nonumber\\
& \quad
+ \sup\{ d_0( \mathbf{v}, D_n'( \alpha_n \,|\, \mathbf{x} ) ) \,|\, \mathbf{v} \in D( \alpha + 2 \delta_1 \,|\, \mathbf{x} ) \}
\nonumber\\
& < (\delta / 2)
+ \sup\{ d_0( \mathbf{v}, D_n'( \alpha_n \,|\, \mathbf{x} ) ) \,|\, \mathbf{v} \in D( \alpha + 2 \delta_1 \,|\, \mathbf{x} ) \} .
\label{lemma10eq6}
\end{align}
Define the event
\begin{align*}
B_n = \left\{ | \alpha_n - \alpha | < \delta_1 \text{ and } \sup_{ \mathbf{y} \in \mathbb{R}^p } | \rho_n( \mathbf{y} \,|\, \mathbf{x} ) - \rho( \mathbf{y} \,|\, \mathbf{x} ) | < \delta_1 \right\} .
\end{align*}
Since $ \alpha_n \stackrel{P}{\longrightarrow} \alpha $ as $ n \to \infty $, and \eqref{eq2} is assumed to be satisfied, there is an integer $ n_3 $ such that for all $ n \ge n_3 $,
\begin{align}
\mathbb{P}[ B_n ] > 1 - (\epsilon / 2) .
\label{lemma10eq7}
\end{align}
When the event $ B_n $ occurs, we have $ \rho_n( \mathbf{Y}_j \,|\, \mathbf{x} ) \ge \alpha_n $ for any sample observation $ \mathbf{Y}_j \in D( \alpha + 2 \delta_1 \,|\, \mathbf{x} ) $. Consequently, when the event $ ( A_n \cap B_n ) $ occurs, we have
\begin{align}
\sup\{ d_0( \mathbf{v}, D_n'( \alpha_n \,|\, \mathbf{x} ) ) \,|\, \mathbf{v} \in D( \alpha + 2 \delta_1 \,|\, \mathbf{x} ) \} 
< (\delta / 2)
\label{lemma10eq8}
\end{align}
since $ D( \alpha + 2 \delta_1 \,|\, \mathbf{x} ) \subset D( \alpha \,|\, \mathbf{x} ) $.
Therefore, from \eqref{lemma10eq4}, \eqref{lemma10eq6}, \eqref{lemma10eq7} and \eqref{lemma10eq8}, we have for all $ n \ge \max\{ n_2, n_3 \} $,
\begin{align}
& \mathbb{P}\left[ \sup\{ d_0( \mathbf{v}, D_n'( \alpha_n \,|\, \mathbf{x} ) ) \,|\, \mathbf{v} \in D( \alpha \,|\, \mathbf{x} ) \} < \delta \right] 
\nonumber\\
& \ge
\mathbb{P}\left[ \sup\{ d_0( \mathbf{v}, D_n'( \alpha_n \,|\, \mathbf{x} ) ) \,|\, \mathbf{v} \in D( \alpha + 2 \delta_1 \,|\, \mathbf{x} ) \} < (\delta / 2) \right] 
\nonumber\\
& \ge
\mathbb{P}\left[ \sup\{ d_0( \mathbf{v}, D_n'( \alpha_n \,|\, \mathbf{x} ) ) \,|\, \mathbf{v} \in D( \alpha + 2 \delta_1 \,|\, \mathbf{x} ) \} < (\delta / 2) \middle\arrowvert A_n \cap B_n \right] \mathbb{P}\left[ A_n \cap B_n \right] \nonumber\\
& = \mathbb{P}\left[ A_n \cap B_n \right]
> 1 - \epsilon .
\label{lemma10eq9}
\end{align}
From \eqref{lemma10eq9}, we get
\begin{align}
\sup_{\mathbf{v} \in D( \alpha \,|\, \mathbf{x} )} d_0( \mathbf{v}, D_n'( \alpha_n \,|\, \mathbf{x} ) )
\stackrel{P}{\longrightarrow} 0
\quad
\text{as }
n \to \infty .
\label{lemma10eq10}
\end{align}
The proof is complete from \eqref{lemma10eq1}, \eqref{lemma10eq2} and \eqref{lemma10eq10}.
\end{proof}

\bibliographystyle{apa}
\bibliography{bibliography_database}

\end{document}